\documentclass[11pt]{article}
\usepackage{amsmath}
\usepackage{amsthm}
\usepackage{amssymb}
\usepackage{algorithmic}
\usepackage{algorithm}
\usepackage{subfig}
\usepackage{color}
\usepackage[english]{babel}
\usepackage{graphicx}
\usepackage{wrapfig,epsfig}
\usepackage{epstopdf}
\usepackage{url}
\usepackage{graphicx}
\usepackage{color}
\usepackage{epstopdf}
\usepackage{scrextend}
\usepackage[T1]{fontenc}
\usepackage{bbm}
\usepackage{comment}

\usepackage{tikz}
\usepackage{hyperref}  
\hypersetup{colorlinks=true,citecolor=blue,linkcolor=blue}
\usetikzlibrary{arrows}
\usepackage[margin=1in]{geometry}
\graphicspath{{./figs/}}

\usepackage{tikz}\usepackage{amsmath}
\usepackage{amsfonts}
\usepackage{amssymb}
\usetikzlibrary{positioning}
\definecolor{processblue}{cmyk}{0.96,0,0,0}

\author{
}

\date{}
\title{Streaming Balanced Clustering}

\allowdisplaybreaks

\newtheorem{theorem}{Theorem}[section]
\newtheorem{lemma}[theorem]{Lemma}
\newtheorem{definition}[theorem]{Definition}

\newtheorem{fact}[theorem]{Fact}

\newtheorem{claim}[theorem]{Claim}

\newcommand{\wt}{\widetilde}

\newcommand{\Dun}{\dot{\bigcup}}
\newcommand{\dun}{\dot{\cup}}

\renewcommand{\varepsilon}{\epsilon}
\renewcommand{\tilde}{\wt}

\newcommand{\cost}{{\rm{cost}}^{(r)}}

\DeclareMathOperator*{\E}{{\bf {E}}}

\DeclareMathOperator{\OPT}{OPT}

\DeclareMathOperator{\poly}{poly}

\DeclareMathOperator{\dist}{dist}

\makeatletter
\newcommand*{\RN}[1]{\expandafter\@slowromancap\romannumeral #1@}
\makeatother

\usepackage{lineno}

\begin{document}\sloppy

\author{Hossein Esfandiari\\Google Research\\esfandiari@google.com\and Vahab Mirrokni \\ Google Research\\ mirrokni@google.com\and Peilin Zhong \\ Columbia University\\ pz2225@columbia.edu}

\begin{titlepage}
  \maketitle
  \begin{abstract}
Clustering of data points in metric space is among the most fundamental problems in computer science with plenty of applications in data mining, information retrieval and machine learning. Due to the necessity of clustering of large datasets, several streaming algorithms have been developed for different variants of clustering problems such as $k$-median and $k$-means problems. However, despite the importance of the context, the current understanding of balanced clustering (or more generally capacitated clustering) in the streaming setting is very limited. The only  previously known streaming approximation algorithm for capacitated clustering requires three passes and only handles insertions.

In this work, we develop \emph{the first single pass streaming algorithm} for a general class of clustering problems that includes capacitated $k$-median and capacitated $k$-means in Euclidean space, using only $\poly ( k d \log \Delta)$ space, where $k$ is the number of clusters, $d$ is the dimension and $\Delta$ is the maximum relative range of a coordinate\footnote{Note that $d\log \Delta$ is the space required to represent one point.}. This algorithm only violates the capacity constraint by a $1+\epsilon$ factor. Interestingly, unlike the previous algorithm, our algorithm handles both insertions and deletions of points. To provide this result we define a decomposition of the space via some curved half-spaces. We used this decomposition to design a strong coreset of size $\poly ( k d \log \Delta)$ for balanced clustering. Then, we show that this coreset is implementable in the streaming and distributed settings.
\end{abstract}

  \thispagestyle{empty}
\end{titlepage}

\section{Introduction}
Clustering of data points in metric space is among the most fundamental problems in computer science with plenty of applications in data mining, information retrieval and machine learning. In many applications there are some natural constraints on the size of the clusters. To capture this constraint, balanced and capacitated clustering have been introduced and widely studied in the classical setting~\cite{an2017lp,adamczyk2018constant,byrka2016approximation,li2017uniform,demirci2016constant,xu2019constant}.

It is fairly well-known that classical algorithms are not practical for large datasets. Streaming setting is one of the most popular settings to design algorithms for large datasets. In the streaming setting, we have a space sublinear in the size of the input, and we are usually restricted to take only one pass over the input. The input stream may contain both insertion and deletion of data points, or may be restricted to only contain insertion of data points.

Due to the necessity of clustering of large datasets, several streaming algorithms have been developed for different variants of clustering problems such as $k$-median~\cite{cop03,fs05,fl11,gmmo00,hm04,chen09,birw16,bfl16,bflsy17} and $k$-means problems~\cite{fms07,chen09,fl11,braverman2016clustering,bfl16,hsyz18}.
However, despite the importance of the context, there is no one-pass streaming algorithm known for balanced or capacitated version of these problems with non-trivial guarantees.
The only previously known approximation algorithm in this context is a three-pass insertion-only streaming algorithm for a general class of capacitated $k$-clustering in $\ell_r$~\cite{bblm14}.
In capacitated $k$-clustering in $\ell_r$, the objective is to assign all of the points into $k$ centers such that, while respecting the capacity constraints, it minimizes the total sum of $r$-th power of the distances. Note that this definition extends capacitated $k$-median (for $r=1$), capacitated $k$-means (for $r=2$) and capacitated $k$-center (for $r=\infty$). Let us say a solution is $(\alpha,\beta)$-approximate solution if its cost is at most $\alpha$ times that of the optimum and it violates the capacity constraints by at most a factor $\beta$. Given a regular sequential $(\alpha,\beta)$-approximation algorithm for capacitated $k$-clustering in $\ell_r$, the previous paper provides an $(O(r \alpha),\beta)$-approximation three-pass streaming algorithm. Unfortunately, there is a large constant hidden in $O(r \alpha)$.

In this paper, we develop the first \emph{single-pass} streaming algorithm for capacitated $k$-clustering in $\ell_r$, using only $\poly ( k d \log \Delta)$ space, where $d$ is the dimension and $\Delta$ is the maximum relative range of a coordinate\footnote{This is roughly the ratio between the maximum distance and the minimum distance. We define this notation in the next subsection and give a further discussion in Section~\ref{sec:preli}.}. Given an $(\alpha,\beta)$-approximation algorithm for weighted capacitated $k$-clustering in $\ell_r$, and arbitrary positive numbers $\eta$ and $\varepsilon$, our algorithm provides an $((1+\varepsilon)\alpha,(1+\eta)\beta)$-approximate solution\footnote{When $\eta$, $\varepsilon$ or $\beta$ are not constant, they appear polynomially in the space.}. Interestingly, unlike the previous algorithm, this algorithm handles \emph{both insertion and deletion} of data points.

Regardless of time complexities, for arbitrary $\varepsilon,\eta \in (0,1]$ our result directly implies $(1+\varepsilon,1+\eta)$-approximation streaming algorithms for capacitated $k$-median and capacitated $k$-means. By applying the $(O(1/\epsilon),1+\epsilon)$-approximation algorithm of~\cite{demirci2016constant} for capacitated $k$-median, and using proper parameters, we have a polynomial time $(O(1/\epsilon),1+\epsilon)$-approximation streaming algorithm for capacitated $k$-median in $\poly ( k d \log \Delta)$ space. Similarly, by applying the $69+\epsilon$-approximation fixed parameter tractable algorithm of~\cite{xu2019constant} for capacitated $k$-means we have a fixed parameter tractable $(69+\epsilon,1+\epsilon)$-approximation streaming algorithm for capacitated $k$-means in $\poly ( k d \log \Delta)$ space.
Furthermore, if $d$ is much larger than $k/\varepsilon$, we can apply~\cite{mmr19} to reduce the dimension to $\poly(k/\varepsilon)$. 
Then our streaming algorithm only needs $d\cdot \poly(k\log\Delta)$ space though the dependence on $k$ and $1/\varepsilon$ becomes slightly larger.

It is easy to convert our streaming algorithm to a distributed algorithm. We use the same distributed model as of ~\cite{kvw14,wz16,bwz16,swz17,swz19}.
In this model we have $s$ machines, where the $i$-th machine holds a subset of input points. There is one coordinator and the communication is only between the coordinator and other machines. The goal is to bound the communication. In Subsection \ref{subsec:results} we present our results, and then in Subsection~\ref{subsec:technique} we discuss our technical contributions.
 
\paragraph{Other Related Works.}
Guha et al. studied the $k$-median problem in the streaming setting and provided the first single pass $2^{O(1/\epsilon)}$-approximation algorithm for this problem using $O(n^{\epsilon})$ space~\cite{gmmo00}. Next, Charikar et al. improved the approximation factor and the space requirement of this problem and provided a constant approximation streaming algorithm that stores $O(k \log^2 n)$ points~\cite{cop03}. 
Braverman et al. developed constant approximation algorithm in streaming sliding window model for both $k$-median and $k$-means using $O(k^3  \log^6 n)$ space~\cite{braverman2016clustering}.

In the Euclidean space, Har-Peled and Mazumdar show a $(1+\epsilon)$-approximation insertion-only streaming algorithm for $k$-means and $k$-median~\cite{hm04}. They developed a coreset that requires $ O(k\epsilon^{-d}\log n)$ space and used it to provide a streaming algorithm that requires $O(k \epsilon^{-d}\log^{2d+2}n)$ space. Later, Har-Peled and Kushal provided a coreset of size $O(k^2\epsilon^{-d})$ for these problems~\cite{hk05}. In high dimensional spaces, Chen presented a streaming algorithm for both problems in $O(k^2d\epsilon^{-2}\log^8 n)$ space~\cite{chen09}.

\subsection{Our Results}\label{subsec:results}
In this paper, we suppose all input and output points are in $\{1,2,\cdots,\Delta\}^d$ for some $\Delta,d\in\mathbb{Z}_{\geq 1}$. 
This assumption is without loss of generality since if the clustering cost is non-zero, we can always discretize the space by changing the cost by an arbitrary small multiplicative error~\cite{bflsy17,hsyz18}.
Given a point set $Q\in[\Delta]^d$, a strong $(\eta,\varepsilon)$-coreset of $Q$ for capacitated $k$-clustering in $\ell_r$ is a subset of points $Q'\subseteq Q$ with weights $w':Q'\rightarrow \mathbb{R}_{>0}$ such that for any capacity $t\geq \lceil|Q|/k\rceil$ and any set of $k$ centers $Z=\{z_1,z_2,\cdots,z_k\}$,
\begin{align*}
\frac{1}{1+\varepsilon}\cdot \cost_{(1+\eta)^2t}(Q,Z)\leq \cost_{(1+\eta)t}(Q',Z,w') \leq (1+\varepsilon)\cdot \cost_t(Q,Z),
\end{align*}
where $\cost_{t}(Q,Z)$ indicates the capacitated clustering cost in $\ell_r$ with respect to centers $Z$ and capacity $t$, and similarly, $\cost_{t'}(Q',Z,w')$ indicates the weighted version of $\ell_r$ capacitated clustering cost.
We refer readers to Section~\ref{sec:preli} for the formal definition of $\cost_t(Q,Z)$ and $\cost_{t'}(Q',Z,w')$.

In this paper, we give the first strong coreset construction for capacitated $k$-clustering.
The size of our coreset is $\poly(\varepsilon^{-1}\eta^{-1}kd\log\Delta)$. 
Our coreset can be constructed in near linear time.

\begin{theorem}[Restatement of Theorem~\ref{thm:offline}]
For a constant $r\geq 1$, given $k\in\mathbb{Z}_{\geq 1},\varepsilon,\eta\in(0,0.5)$ and a point set $Q\subseteq[\Delta]^d$ with $|Q|=n$, there is a randomized algorithm which takes $O(nd\log^2(nd\Delta))$ time and outputs a subset of points $Q'\subseteq Q$ with weights $w':Q'\rightarrow \mathbb{R}_{>0}$ such that with probability at least $0.9$, $(Q',w')$ is a strong $(\eta,\varepsilon)$-coreset of $Q$ for capacitated $k$-clustering in $\ell_r$ and $|Q'|\leq \poly(\varepsilon^{-1}\eta^{-1}kd\log\Delta)$.
\end{theorem}

Our coreset can also be constructed in streaming and distributed setting efficiently.
The streaming model studied in this paper is the dynamic streaming model which allows both insertion and deletion of points. 
\begin{theorem}[Restatement of Theorem~\ref{thm:streaming}]
For a constant $r\geq 1$, given $k\in\mathbb{Z}_{\geq 1},\varepsilon,\eta\in(0,0.5)$ and a point set $Q\subseteq[\Delta]^d$ obtained by a stream of insertions and deletions, there is a streaming algorithm which takes one pass over the stream and with probability at least $0.9$ outputs a strong $(\eta,\varepsilon)$-coreset $(Q',w')$ of $Q$ for capacitated $k$-clustering in $\ell_r$.
Furthermore, both $|Q'|$ and the space of the streaming algorithm is at most $\poly(\varepsilon^{-1}\eta^{-1}kd\log\Delta)$.
\end{theorem}

In the distributed model, the input is distributed into machines. 
Each machine can only communicate with the coordinator.
The goal in this model is to design a protocol with small communication cost.
\begin{theorem}[Restatement of Theorem~\ref{thm:distributed}]
For a constant $r\geq 1$, given $k\in\mathbb{Z}_{\geq 1},\varepsilon,\eta\in(0,0.5)$ and a point set $Q\subseteq[\Delta]^d$ partitioned into $s$ machines, there is a distributed protocol which on termination with probability at least $0.9$ leaves a subset of points $Q'\subseteq Q$ with weights $w':Q'\rightarrow \mathbb{R}_{>0}$ such that $(Q',w')$ is a strong $(\eta,\varepsilon)$-coreset of $Q$ for capacitated $k$-clustering in $\ell_r$ and the size of the coreset is at most $\poly(\varepsilon^{-1}\eta^{-1}kd\log\Delta)$.
Furthermore, the total communication cost is at most $s\cdot \poly(\varepsilon^{-1}\eta^{-1}kd\log\Delta)$ bits.
\end{theorem}

\subsection{Our Techniques}\label{subsec:technique}
Let us first discuss how to construct a strong coreset for capacitated $k$-means.
Later we will show how to generalize the idea for capacitated $k$-clustering in $\ell_r$ for general $r\geq 1$.

Our starting point is a common partitioning approach~\cite{chen09,bflsy17,hsyz18} for $k$-clustering coreset construction.
The size $n$ input point set $Q\subseteq[\Delta]^d$ is partitioned into $\poly(kd\log\Delta)$ parts of points $P_1,P_2,\cdots,P_s$ such that if we move all points in each part to an arbitrary point in this part, the optimal $k$-means cost for moved points should not change too much.
In other words,
\begin{align}
\sum_{i=1}^s |P_i| \cdot \left(\max_{p,q\in P_i} \dist(p,q)\right)^2 \leq \poly(kd\log\Delta)\cdot \OPT_{k\text{-means}}, \label{eq:bound_additive_error}
\end{align}
where $\OPT_{k\text{-means}}=\min_{Z\subset[\Delta]^d:|Z|=k} \sum_{p\in Q} \dist^2(p,Z)$.

Let us briefly review the sampling based strong coreset construction for standard $k$-means problem.
Consider a fixed set of $k$ centers $Z\subset[\Delta]^d$. 
Each point $p\in P_i$ is sampled with probability $\poly(\varepsilon^{-1}kd\log\Delta)/|P_i|$ and each sampled point is assigned a weight $w(p)$ which is inverse sampling probability.
The expected number of points sampled is $\poly(\varepsilon^{-1}kd\log\Delta)$
 and $\sum_{\text{sampled }p\in P_i} w(p)\cdot \dist^2(p,Z)$ is an unbiased estimator of $\sum_{p\in P_i}\dist^2(p,Z)$.
By triangle inequality, $\forall p',q'\in P_i$, $|\dist(p',Z)-\dist(q',Z)|\leq \max_{p,q\in P_i}\dist(p,q)$.
This can upper bound the variance of the cost of each sampled point.
By Bernstein inequality, with high probability, for every part $P_i$, 
the difference between
$\sum_{\text{sampled }p\in P_i} w(p)\cdot \dist^2(p,Z)$ and  $\sum_{p\in P_i}\dist^2(p,Z)$ is at most $\varepsilon \sum_{p\in P_i}\dist^2(p,Z) +\frac{\varepsilon}{\poly(kd\log\Delta)} |P_i|\cdot \left(\max_{p,q\in P_i}\dist(p,q)\right)^2$.
The additive error term $\frac{\varepsilon}{\poly(kd\log\Delta)} |P_i|\cdot \left(\max_{p,q\in P_i}\dist(p,q)\right)^2$ is acceptable since the total additive error is bounded by $\varepsilon\cdot \OPT_{k\text{-means}}$ (Equation~\eqref{eq:bound_additive_error}) and thus becomes a small relative error.
Notice that the total number of choices of $Z$ is at most $\Delta^{kd}$.
By taking union bound over all possible choices of $Z$, with high probability, the sampled points together with their weights become a strong coreset for $k$-means with size $\poly(\varepsilon^{-1}kd\log\Delta)$.
We refer readers to \cite{chen09,fl11,bfl16,bflsy17,hsyz18} for more details and history.

Unfortunately, the above analysis breaks for capacitated $k$-means.
Due to capacity constraints, each point may not be assigned to the closest center in the capacitated $k$-means solution.
If we look at the samples from the given point set, it is unclear how to determine the cost of each sampled point without looking at entire point set.
This is an obstacle to obtain an unbiased estimator of the capacitated $k$-means cost.
To construct an unbiased estimator, we need to find a simple way to determine the cost of each sampled point.
Again consider a fixed set of $k$ centers $Z=\{z_1,z_2,\cdots,z_k\}\subset[\Delta]^d$.
If we know that each point $p\in Q$ is assigned to the center $\pi(p)\in Z$, then $\sum_{\text{sampled }p} w(p)\dist^2(p,\pi(p))$ is an unbiased estimator of the clustering cost with respect to the assignment $\pi:Q\rightarrow Z$, i.e., $\sum_{p} \dist^2(p,\pi(p))$.
In addition, $\sum_{\text{sampled }p:~\pi(p)=z_i} w(p)$ is an unbiased estimator of the number of points assigned to the center $z_i$.
If for every assignment $\pi$, $\sum_{\text{sampled }p} w(p)\dist^2(p,\pi(p))$ is a good approximation to $\sum_{p} \dist^2(p,\pi(p))$ and $\forall i\in[k]$, $\sum_{\text{sampled }p:~\pi(p)=z_i} w(p)$ is a good estimation of the size of the cluster with center $z_i$, then the capacitated clustering cost of samples is a good approximation of the capacitated clustering cost of $Q$ if we allow some relaxation of capacity constraints.
However, the possible choices of assignment $\pi$ can be as large as $k^n$.
It implies that if we want to obtain a good estimation for every assignment $\pi$, we need at least $\Omega(\log(k^n))=\Omega(n\log k)$ samples which is even worse than taking entire point set $Q$.
The issue of the above attempt is that we want to have uniformly good estimations for all possible assignments. 
To handle this issue, we should reduce the number of assignments that we care about.
An observation is that if an assignment $\pi$ is not an optimal assignment for any capacity constraint, then we do not care the quality of the estimated cost for $\pi$.
Then we hope the number of assignments which can be optimal for some capacity constraint is small.

Our main technical contribution is finding a good structure of possible optimal assignments, and such structure can be used to upper bound the number of those assignments. 
Consider an assignment $\pi:Q\rightarrow Z$ which is optimal for some capacity constraint.
For two centers $z_i,z_j$, by Pythagorean theorem there must be a $(d-1)$-dimensional hyperplane separating the points in the cluster with center $z_i$ and the points in the cluster with center $z_j$ (see Figure~\ref{fig:hyperplane_separation}), and the hyperplane  is perpendicular to the line connecting $z_i$ and $z_j$.
\begin{figure}
\centering
\includegraphics[width=\textwidth]{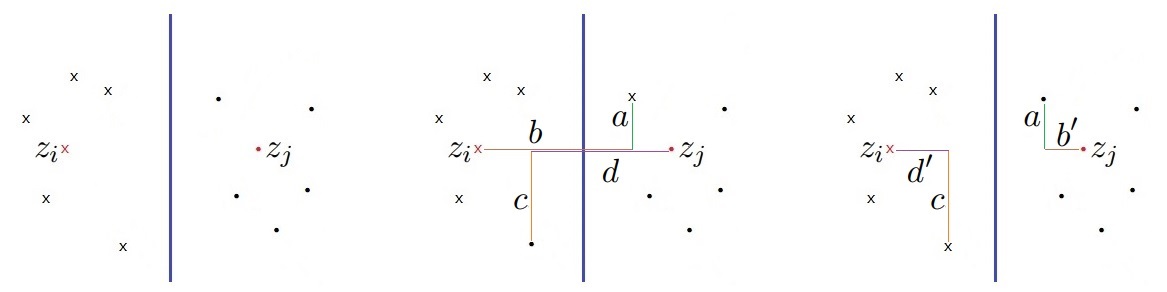}
\caption{A cross denotes a point assigned to $z_i$ and a dot denotes a point assigned to $z_j$. We can find a hyperplane separating two clusters. 
Suppose a point on the right side is assigned to $z_i$ and a point on the left side is assigned to $z_j$. 
By Pythagorean theorem, the total cost of these two points is $a^2+b^2+c^2+d^2$. 
If we switch the assignments of these two points, the cost is $a^2+b'^2+c^2+d'^2$ which is smaller since $b'^2+d'^2<b^2+d^2$.
Thus, if two clusters cannot be separated by a hyperplane, the assignment cannot be optimal.
}\label{fig:hyperplane_separation}
\end{figure}
Let $H_{(i,j)}$ denote the half-space which is one side of the hyperplane containing $z_i$.
Similarly we denote $H_{(j,i)}$ as another side containing $z_j$.
For every pair of centers $z_i,z_j$, we can always define the half-space $H_{(i,j)}$ in the above way.
It is clear to see that $\pi(p)=z_i$ if and only if $p\in \bigcap_{j\not =i} H_{(i,j)}$.
In other words, the assignment $\pi$ can be determined by the set of all half-spaces $\{H_{(i,j)}\mid i\not=j\}$.
Since $Q\subseteq[\Delta]^d$, the number of possible half-space $H_{(i,j)}$ for $i,j$ is at most $\Delta^d$.
The total number of possible set of half-spaces $\{H_{(i,j)}\mid i\not=j\}$ is at most $(\Delta^d)^{k\choose 2}\leq \Delta^{O(dk^2)}$.
It implies that the number of possible optimal assignments is at most $\Delta^{O(dk^2)}$ which is much less than the total number of all possible assignments.

Although we are able to get an unbiased estimator of the cost of the optimal assignment $\pi$, there is an issue remaining for bounding the variance of the cost of samples.
For two points $p,q$ from the same part $P_i$, the difference between $\dist(p,\pi(p))$ and $\dist(q,\pi(q))$ may be much larger than $\dist(p,q)$ (see Figure~\ref{fig:unbounded_variance} for an example).
\begin{figure}[t!]
    \centering
    \includegraphics[width=0.8\textwidth]{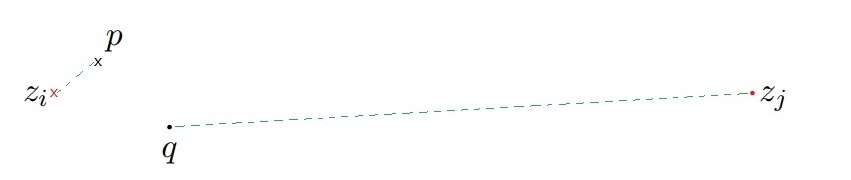}
    \caption{Due to capacity constraint, $p$ is assigned to the center $z_1$ and $q$ is assigned to the center $z_2$. 
    The difference between $\dist(p,z_1)$ and $\dist(q,z_2)$ depends on $\dist(z_1,z_2)$ and thus can be arbitrarily large.}
    \label{fig:unbounded_variance}
\end{figure}
Despite the difference between $\dist(p,\pi(p))$ and $\dist(q,\pi(q))$ can be arbitrarily large in general, the difference is upper bounded by $\dist(p,q)$ if $p$ and $q$ are assigned to the same center, i.e., $\pi(p)=\pi(q)$.
This observation motivates us to further conceptually partition $P_i$ into $k$ regions where each region contains the points assigned to the same center.
If each region either contains no point from $P_i$ or contains at least $|P_i|/\poly(\varepsilon^{-1}kd\log\Delta)$ points, then we can estimate $\sum_{p\in P_i}\dist^2(p,\pi(p))$ by sampling each point $p\in P_i$ with probability $\poly(\varepsilon^{-1}kd\log\Delta)/|P_i|$, and with high probability, the total estimated error can becomes a small relative error as discussed in the early paragraph in this section.
Unfortunately, there could be some region in which number of points is much less than $|P_i|/\poly(\varepsilon^{-1}kd\log\Delta)$.
These points may contribute a lot to the cost since they may be far away from their center. 
But we may sample none of them due to relatively low sampling rate and thus it can cause a large approximation error.
To handle this, we develop a method to transfer the optimal assignment $\pi$ to an another assignment $\pi':Q\rightarrow Z$ such that $\pi'$ approximately satisfies the capacity constraint and does not increase the cost by too much.
Furthermore, for each part $P_i$ and center $z_j$, either none of point in $P_i$ is assigned to $z_j$ by $\pi'$ or there are at least $|P_i|/\poly(\varepsilon^{-1}kd\log\Delta)$ number of points in $P_i$ are assigned to $z_j$ by $\pi'$.
If we estimate the cost of $\pi'$, the variance of the cost of sampled points can have a good upper bound.
Thus, with high probability, we can estimate the cost of $\pi'$.
By taking union bound over all the possible choices of $Z$ and the possible transferred assignments $\pi'$, we can prove that the sampled points form a strong coreset for capacitated $k$-means with high probability. 

Next, let us discuss how to extend the above idea for capacitated $k$-means to general capacitated $k$-clustering in $\ell_r$ for $r\geq 1$. 
The main difficulty to extend our capacitated $k$-means to $\ell_r$ case is that we may not find a hyperplane to separate two clusters since we cannot apply Pythagorean theorem for $\ell_r$ cost (see Figure~\ref{fig:generalized_hyperplane}).
\begin{figure}[b!]
    \centering
    \includegraphics[width=0.8\textwidth]{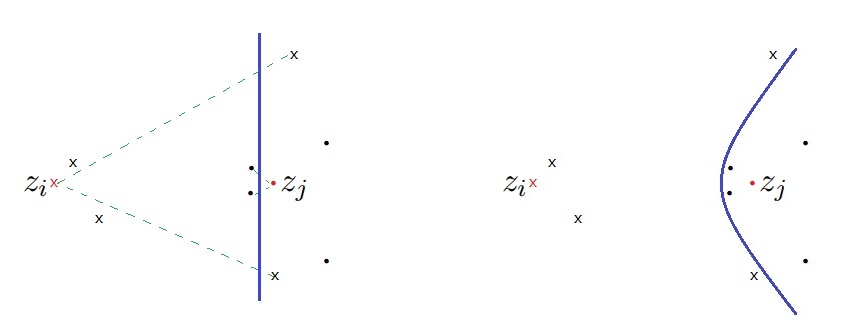}
    \caption{The optimal clusters for capacitated $k$-median may not be separated by a hyperplane. 
    But it is possible to use a curved hyperplane to separate the clusters. 
    For example, two clusters of capacitated $k$-median in $2$-dimensional space may be separated by a branch of a hyperbola.
    }
    \label{fig:generalized_hyperplane}
\end{figure}
But fortunately, we can find a curved hyperplane to separate two clusters. 
Consider two centers $z_i$ and $z_j$.
We can define a curved hyperplane as $\{x\in\mathbb{R}^d\mid \dist^r(x,z_i)-\dist^r(x,z_j) = a\}$ for some parameter $a\in\mathbb{R}$.
For an assignment $\pi:Q\rightarrow Z$, if $\dist^r(p,z_i)-\dist^r(p,z_j) < a$ and $\dist^r(q,z_i)-\dist^r(q,z_j) > a$ but $\pi(p)=z_j,\pi(q)=z_i$, then $\pi$ cannot be an optimal assignment since 
\begin{align*}
    \dist^r(p,z_i)+\dist^r(q,z_j) < \dist^r(q,z_i)+\dist^r(p,z_j),
\end{align*}
which implies switching the assigned centers of $p,q$ can give an assignment with smaller cost.
Thus, for an optimal assignment and any two clusters under the assignment, there always exists a curved hyperplane separating two clusters.
If we replace the half-spaces in previous paragraphs with the half-spaces defined by the above curved hyperplanes,  then the argument works for $k$-clustering in $\ell_r$ for general $r\geq 1$.

Since the partition $P_1,P_2,\cdots,P_s$ and its $\ell_r$ variant version can be computed in the streaming model (with both insertion and deletion) and distributed model~\cite{bflsy17,hsyz18} and we only need to sample points in each part with uniform sampling rate, our strong coreset construction can be easily implemented in the streaming setting and distributed setting. 
\section{Preliminaries}\label{sec:preli}

We use $[n]$ to denote the set $\{1,2,\cdots, n\}$.
For any $x\in\mathbb{R},a\in\mathbb{R}_{>0}$, we use $x\pm a$ to denote the interval $(x-a,x+a)$.
For any $x\in\mathbb{R}_{\geq 0},\varepsilon\in(0,1)$, we use $(1\pm \varepsilon)\cdot x$ to denote the interval $((1-\varepsilon)\cdot x, (1+\varepsilon)\cdot x)$.
Consider two points $x,y\in\mathbb{R}^d$. 
If $\exists i\in[d]$ such that $x_1=y_1,x_2=y_2,\cdots,x_{i-1}=y_{i-1}$ and $x_i<y_i$ then $x$ is smaller than $y$ in the alphabetical order.
For $r\geq 1$, we use $\|x\|_r$ to denote the $\ell_r$ norm of $x\in\mathbb{R}^d$, i.e., $\|x\|_r=(\sum_{i=1}^d |x_i|^r)^{1/r}$.
We use $\dist(x,y)$ to denote the Euclidean distance between $x$ and $y$, i.e., $\dist(x,y)=\|x-y\|_2$. 

\begin{fact}\label{fac:approx_tri}
For $r\geq 1$ and any $x,y,z\in\mathbb{R}^d$, $\dist^r(x,z)\leq 2^{r-1}\left(\dist^r(x,y)+\dist^r(y,z)\right).$
\end{fact}
\begin{proof}
\begin{align*}
\dist^r(x,z) \leq \left(\dist(x,y)+\dist(y,z)\right)^r \leq \frac{(2\dist(x,y))^r+(2\dist(y,z))^r}{2}\leq 2^{r-1} \left(\dist^r(x,y)+\dist^r(y,z)\right),
\end{align*}
where the first step follows from triangle inequality and the second step follows from convexity.
\end{proof}

\begin{definition}[Half-space]\label{def:halfspace}
Consider $r\geq 1$ and two points $z_1,z_2\in[\Delta]^d$. 
All points $x_1,x_2,\cdots,x_{\Delta^d}\in[\Delta]^d$ are sorted such that $\forall i\in[\Delta^d-1]$, either $\dist^r(x_i,z_1)-\dist^r(x_i,z_2)<\dist^r(x_{i+1},z_1)-\dist^r(x_{i+1},z_2)$ or $\dist^r(x_i,z_1)-\dist^r(x_i,z_2)=\dist^r(x_{i+1},z_1)-\dist^r(x_{i+1},z_2)$ and $x_i$ is smaller than $x_{i+1}$ in the alphabetical order. 
Let $t$ be an arbitrary integer in $\left[\Delta^d\right]$, then the set
$
H=\{x_1,x_2,\cdots,x_t\}
$
is an $\ell_r$-half-space corresponding to $(z_1,z_2,t)$.
\end{definition}

Given a set of points $Z\subset \mathbb{R}^d$ and a point $x\in\mathbb{R}^d$, we define $\dist(x,Z)=\dist(Z,x)=\min_{y\in Z}\dist(x,y)$.
For two sets of points $P,Q\subset\mathbb{R}^d$, we define the distance between $P,Q$ as $\dist(P,Q)=\min_{p\in P,q\in Q}\dist(p,q)$.
Consider an arbitrary set $S$.
If $S=S_1\cup S_2\cup\cdots\cup S_s$ and $\forall i\not=j\in[s],S_i\cap S_j=\emptyset$, then $S_1,S_2,\cdots,S_s$ is a partition of $S$, and $S_i$ is called a part of the partition.
We use $S=S_1\dun S_2\dun\cdots\dun S_s$ to denote that $S_1,S_2,\cdots,S_s$ partitions $S$.
For a point set $Q\subset [\Delta]^d$, a set of centers $Z=\{z_1,z_2,\cdots,z_k\}\subseteq [\Delta]^d$ with $|Z|=k$ and a size parameter $t\geq |Q|/k$, we define 
\begin{align*}
\cost_{t}(Q,Z)=\min_{\underset{\underset{\forall i\in[k],|S_i|\leq t}{Q = S_1\dun S_2\dun \cdots \dun S_k,}}{S_1,S_2,\cdots,S_k:}} \sum_{i=1}^k \sum_{p\in S_i} \dist^r(p,z_i).
\end{align*}
For $t=\infty$, we define
\begin{align*}
\cost_{\infty}(Q,Z)=\sum_{p\in Q} \dist^r(p, Z).
\end{align*}
For convenience of the notation, we use $\cost(Q,Z)$ to denote $\cost_{\infty}(Q,Z)$ for short.
Similarly, we can define a weighted version of the cost function.
Suppose each point $p\in Q$ has a weight $w(p)$.
We define
\begin{align*}
\cost_t(Q,Z,w) = \min_{\underset{\underset{\forall i\in[k],\sum_{p\in S_i} w(p)\leq t}{Q = S_1\dun S_2\dun \cdots \dun S_k,}}{S_1,S_2,\cdots,S_k:}} \sum_{i=1}^k \sum_{p\in S_i} w(p)\cdot \dist^r(p,z_i).
\end{align*}
If there is no partition $Q = S_1\dun S_2\dun \cdots \dun S_k$ satisfying $\forall i\in[k],\sum_{p\in S_i} w(p)\leq t$, we define $\cost_t(Q,Z,w)=\infty$.
We denote
\begin{align*}
\cost(Q,Z,w)=\cost_{\infty}(Q,Z,w) = \sum_{p\in Q} w(p)\cdot \dist^r(p,Z). 
\end{align*}

\begin{fact}
Consider a point set $Q\subseteq[\Delta]^d$ and a parameter $k\in\mathbb{Z}_{\geq 1}$. 
For $r\geq 1,\varepsilon,\eta\in(0,0.5)$, if $Q'\subseteq Q,w':Q'\rightarrow \mathbb{R}_{> 0}$ satisfies that $\forall t\geq |Q|/k$ and $\forall Z\subset[\Delta]^d$ with $|Z|=k$,
\begin{align*}
\frac{1}{1+\varepsilon}\cdot\cost_{(1+\eta)^2t}(Q,Z)\leq \cost_{(1+\eta)t}(Q',Z,w')\leq(1+\varepsilon)\cdot\cost_t(Q,Z),
\end{align*}
then for $\hat{Z}\subset[\Delta]^d$ with $|\hat{Z}|=k$ which satisfies
\begin{align*}
\cost_{(1+\eta)\beta t}(Q',\hat{Z},w')\leq \alpha \min_{Z\subset[\Delta]^d:|Z|=k} \cost_{(1+\eta) t}(Q',Z,w')
\end{align*}
for some $\alpha,\beta\geq 1$, we have
\begin{align*}
\cost_{(1+O(\eta))\beta t}(Q,\hat{Z})\leq (1+O(\varepsilon))\alpha \min_{Z\subseteq[\Delta]^d:|Z|=k} \cost_t(Q,Z).
\end{align*}
\end{fact}
\begin{proof}
We have
\begin{align*}
&\cost_{(1+\eta)^2\beta t}(Q,\hat{Z})\\
\leq & (1+\varepsilon)\cdot \cost_{(1+\eta)\beta t}(Q',\hat{Z},w')\\
\leq & (1+\varepsilon)\alpha\cdot \min_{Z\subset [\Delta]^d:|Z|=k} \cost_{(1+\eta)t} \cost(Q',Z,w')\\
\leq & (1+\varepsilon)^2\alpha \cdot \min_{Z\subset[\Delta]^d:|Z|=k} \cost_{t}(Q,Z)
\end{align*}
\end{proof}
where the first and the last step follows from the strong coreset property of $(Q',w')$, and the second step follows from that $\hat{Z}$ is an $(\alpha,\beta)$-approximation.
Notice that $(1+\eta)^2=(1+O(\eta))$ and $(1+\varepsilon)^2=(1+O(\varepsilon))$ since $\eta,\varepsilon\in(0,0.5)$.
We complete the proof.
\section{The Coreset Construction}

In this section we show an offline construction of the coreset and give an analysis of the algorithm.
In Section~\ref{sec:points_partitioning}, we show how to partition the input point set. 
The partitioning scheme follows the common thread of work \cite{chen09,bflsy17,hsyz18}. 
Next, we describe our algorithm in Algorithm~\ref{alg:coreset_construction}. 
To prove the correctness of our algorithm, we develop a novel \emph{half-space} argument.
We show the details of the argument in Section~\ref{sec:construction}. 
This is the main technical contribution of this section.
In Section~\ref{sec:assignment_via_coreset}, given a set of centers and capacity constraints, we show how to efficiently compute a good assignment for the coreset and build a good representation of a good assignment for the original point set.

\subsection{Points Partitioning}\label{sec:points_partitioning}
In this section, we use a common approach to partition the point set. 
We refer readers to \cite{chen09,bflsy17,hsyz18} for more details and history of this partitioning approach.
We put all missing proofs in this section to Appendix~\ref{sec:missing_proofs_partition}

Let us partition the space $[\Delta]^d$ by a randomly shifted hierarchical grid structure.
Without loss of generality, we suppose $\Delta=2^L$ for some integer $L$.
We choose a vector $v\in\mathbb{R}^d$ such that each entry is an i.i.d. random sample drawn uniformly from $[0,\Delta]$.
Then we can impose $L+1$ level grids $G_0,G_1,\cdots,G_L$, where the grid $G_i$ partitions the space $\mathbb{R}^d$ into cells with side length $g_i=\Delta/2^i$, and there is a cell which has a corner with location $v$. 
More precisely, 
\begin{align*}
&\forall i\in \{0,1,\cdots,L\}, \\
&G_i=\left\{C~\Big|~C=[v_1+g_i t_1, v_1+g_i (t_1+1)  )\times\cdots\times [v_d+g_i t_d, v_d+g_i (t_d+1) ),t_1,t_2,\cdots,t_d\in\mathbb{Z} \right\}.
\end{align*} 
For convenience, we also define the gird $G_{-1}$ in the similar way. 
Since the cell in $G_{-1}$ has side length $g_{-1}=2\Delta$, there must be a cell which contains all the points in $[\Delta]^d$.
If $C\in G_i,C'\in G_j$ and $C\subseteq C'$, then we call cell $C'$ an ancestor of cell $C$.
For a point $p\in\mathbb{R}^d$, if $p\in C$ for some cell $C\in G_i$, then we define $c_i(p)=C$.
Similarly, for a point set $P\subset \mathbb{R}^d$, if $P\subseteq C$ for some cell $C\in G_i$, then we denote $c_i(P)=C$.

Consider an input point set $Q\subseteq [\Delta]^d$ and a parameter $k\in\mathbb{Z}_{\geq 1}$.
We denote
\begin{align*}
\OPT^{(r)}_{k\text{-clus}}=\min_{Z\subseteq [\Delta]^d:~|Z|\leq k} \cost(Q,Z).
\end{align*}
Let us review the heavy cell partitioning scheme (Algorithm~\ref{alg:heavy_light}). 
\begin{algorithm}[t!]
	\small
	\begin{algorithmic}[1]\caption{Partitioning via Heavy Cells}\label{alg:heavy_light}
		\STATE {\bfseries Predetermined:}  $o\in \left[1,\Delta^d\cdot \left(\sqrt{d}\Delta\right)^{r}\right]$ which is a guess of the optimal standard $l_r$ $k$-clustering cost
		\STATE {\bfseries Input:} $Q\subseteq [\Delta]^d$, $k\in\mathbb{Z}_{\geq 1}$
		\STATE Impose randomly shifted grids $G_{-1},G_0,\cdots,G_L$.
		\FOR{$i:=-1\rightarrow L-1$}
			\STATE $T_i(o)\gets 0.01\cdot \frac{o}{(\sqrt{d}g_i)^r}$.
			\FOR{$C\in G_i:~C\cap Q\not = \emptyset$}
				\STATE Estimate the size of $|C\cap Q|$ up to some precision, and let $\tau(C\cap Q)$ be the estimated size. \label{sta:estimated_size}
				\STATE If $\tau(C\cap Q) \geq T_i(o)$ and all the ancestors of $C$ are \emph{heavy}, mark $C$ as heavy. 
				\STATE Otherwise, if all the ancestors of $C$ are marked as heavy, mark $C$ as \emph{crucial}.
			\ENDFOR
		\ENDFOR
		\STATE For $C\in G_L$, if all the ancestors of $C$ are heavy, mark $C$ as crucial.
		\STATE $\forall i\in\{0,1,\cdots,L\},$ let $s_{i}$ denote the number of heavy cells in $G_{i-1}$.
		\STATE Partition $Q=\Dun_{i=0}^L\Dun_{j=1}^{s_i} Q_{i,j}$: $Q_{i,j}=\bigcup_{\text{crucial }\hat{C}\in G_i:\hat{C}\subset C} (\hat{C}\cap Q)$ where $C$ is the $j$-th heavy cell in $G_{i-1}$.\label{sta:conceptually_partitioning}
		\STATE {\bfseries Output:} All cells marked as heavy in $G_{-1},G_0,G_1,\cdots,G_{L-1}$
	\end{algorithmic}
\end{algorithm}
In Algorithm~\ref{alg:heavy_light}, once the heavy cells are determined, the partitioning of $Q=\Dun_{i=0}^L\Dun_{j=1}^{s_i} Q_{i,j}$ is determined.
Thus, we explicitly store all the heavy cells and conceptually partition $Q$ into $Q_{0,1},Q_{0,2},\cdots,Q_{0,s_0},Q_{1,1},\cdots,Q_{1,s_1},\cdots,Q_{L,s_L}$ for analysis.

\begin{definition}\label{def:good_estimation_alg1}
If the estimated size $\tau(C\cap Q)$ in line~\ref{sta:estimated_size} of Algorithm~\ref{alg:heavy_light} satisfies either $\tau(C\cap Q)\in |C\cap Q|\pm 0.1T_i(o)$ or $\tau(C\cap Q)\in (1\pm 0.01)\cdot |C\cap Q|$, then the estimated size $\tau(C\cap Q)$ is good for $C$.
\end{definition}

If the guess $o$ is close to $\OPT^{(r)}_{k\text{-clus}}$, the number of heavy cells cannot be too large with a good probability. 
The main reason is that there cannot be too many \emph{center cells}.
Let $Z^*\subset[\Delta]^d$ with $|Z^*|\leq k$ be an optimal solution of the standard $\ell_r$ $k$-clustering problem of $Q$, i.e., $\cost(Q,Z^*)=\OPT^{(r)}_{k\text{-clus}}$. 
We call a cell $C\in G_i$ a center cell if $\dist(C,Z^*)\leq g_i/d$.
Let $\mathcal{F}$ denote the event that the total number of center cells is at most $2000kL$.

\begin{lemma}[Lemma 14 of \cite{hsyz18}]\label{lem:number_of_center_cells}
$\mathcal{F}$ happens with probability at least $0.99$.
\end{lemma}

In the remaining of the paper, we condition on the event $\mathcal{F}$. 
Since $\mathcal{F}$ happens, it is able to show that there cannot be too many heavy cells.

\begin{lemma}[Number of heavy cells]\label{lem:num_heavy_cells}
Suppose $o\leq \OPT^{(r)}_{k\text{-clus}}$.
If the estimated size $\tau(C\cap Q)$ in line~\ref{sta:estimated_size} of Algorithm~\ref{alg:heavy_light} is good (Definition~\ref{def:good_estimation_alg1}) for every cell $C$, then condition on $\mathcal{F}$,
the number of heavy cells outputted by Algorithm~\ref{alg:heavy_light}, $\sum_{i=0}^L s_i$, is at most $2000 (k+d^{1.5r})L\cdot\frac{\OPT^{(r)}_{k\text{-clus}}}{o}$.
\end{lemma}

In the following, we show that removal of small parts will not change the cost of balanced (capacitated) $k$-clustering by too much.
\begin{lemma}\label{lem:small_parts_removal}
$\forall i\in\{0,1,\cdots,L\},$ let $\mathcal{P}_i=\{Q_{i,j}\mid j\in[s_i]\}$, where $Q_{i,j}$ are parts computed by  line~\ref{sta:conceptually_partitioning} of Algorithm~\ref{alg:heavy_light}. 
Let $\mathcal{P}^N_0\subseteq \mathcal{P}_0,\mathcal{P}^N_1\subseteq \mathcal{P}_1,\cdots,\mathcal{P}^N_L\subseteq \mathcal{P}_L$ be arbitrary subsets of parts satisfying $\forall i\in\{0,1,\cdots,L\},P\in \mathcal{P}_i^N$, $|P|\leq 2\gamma T_i(o)$, where $\gamma = \min\left(\frac{\eta}{8\cdot 2^{r}kL},\frac{\varepsilon}{4000\cdot 2^{2r}\left(k+d^{1.5r}\right)L}\right)$ for some arbitrary $\varepsilon,\eta\in(0,0.5)$.
Suppose $o\leq \OPT_{k\text{-means}}$, $\sum_{i=0}^L s_i\leq 20000(k+d^{1.5r})L$, and $\forall i\in \{-1,0,\cdots,L-1\},$ every heavy cell $C\in G_i$ satisfies $|C\cap Q|\geq 0.5 T_i(o)$.
Then $\forall t\geq |Q|/k$ and $Z\subseteq[\Delta]^d$ with $|Z|= k$, 
\begin{align*}
\begin{array}{ccc}
\cost_{t}(Q\setminus Q^N,Z)\leq \cost_{t}(Q,Z) & \text{and} & \cost_{(1+\eta)\cdot t}(Q,Z)\leq (1+\varepsilon)\cost_{t}(Q\setminus Q^N,Z),
\end{array}
\end{align*}
where $Q^N=\bigcup_{i=0}^L\bigcup_{P\in\mathcal{P}_i^N} P$.
\end{lemma}

\subsection{Construction and Analysis}\label{sec:construction}
Our coreset construction is shown in Algorithm~\ref{alg:coreset_construction}. In the remaining of the section, let us analyze the algorithm.
Before we starts our proof, let us assume that both $\tau\left(\bigcup_{j=1}^{s_i} Q_{i,j}\right)$ and $\tau(Q_{i,j})$ are good estimations of $\left|\bigcup_{j=1}^{s_i} Q_{i,j}\right|$ and $|Q_{i,j}|$ respectively.
\begin{algorithm}[t!]
	\small
	\begin{algorithmic}[1]\caption{Coreset Construction}\label{alg:coreset_construction}
		\STATE {\bfseries Predetermined:} $o\in \left[1,\Delta^d\cdot \left(\sqrt{d}\Delta\right)^{r}\right]$ which is a guess of the optimal standard $l_r$ $k$-clustering cost
		\STATE {\bfseries Input:} $Q\subseteq [\Delta]^d$, $k\in\mathbb{Z}_{\geq 1},\eta,\varepsilon\in(0,0.5)$
		\STATE $L\gets\log\Delta,\gamma\gets 2^{-2(r+10)}\min\left(\frac{\eta}{kL},\frac{\varepsilon}{(k+d^{1.5r})L}\right)$, $\xi\gets 2^{-2(r+10)} \frac{\min(\varepsilon,\eta)}{k(k+d^{1.5r})L^2}$,$\lambda\gets 10^6r k^3 dL\lceil\log(kdL)\rceil$.
		\STATE Compute the partition $Q=\Dun_{i=0}^L\Dun_{j=1}^{s_i} Q_{i,j}$. Set $T_i(o)\gets 0.01 \cdot \frac{o}{(\sqrt{d}g_i)^r}$.\hfill{//Algorithm~\ref{alg:heavy_light}} \label{sta:call_alg1}
		\STATE If $\sum_{i=0}^L s_i> 20000(k+d^{1.5r})L$, return FAIL. \label{sta:return_fail_in_alg1}
		\STATE For $0\leq i\leq L$, let $\tau\left(\bigcup_{j=1}^{s_i} Q_{i,j}\right)$ be an estimation of $\sum_{j=1}^{s_i}|Q_{i,j}|$ up to some precision. Return FAIL if 
		\begin{align*}
		    \exists i\in\{0,1,\cdots,L\},\tau\left(\bigcup_{j=1}^{s_i} Q_{i,j}\right)> 10000(kL+d^{1.5r})T_i(o).
		\end{align*}\label{sta:return_fail_in_alg2}
		\FOR{$i=0\rightarrow L$}
		    \STATE $\mathcal{P}^I_i\gets \emptyset,Q'_i\gets \emptyset$, $\phi_i \gets \min\left(1,2^{2(r+10)}\cdot\frac{\lambda}{\xi^3\gamma T_i(o)}\right)$.
		    \STATE For $j\in[s_i]$, let $\tau(Q_{i,j})$ be an estimation of $|Q_{i,j}|$ up to some precision. If $\tau(Q_{i,j})\geq \gamma T_i(o)$,
		    \begin{align*}
		        \mathcal{P}^I_i\gets \mathcal{P}^I_i\cup \{Q_{i,j}\}.
		    \end{align*} \label{sta:lb_of_each_final_part}
		    \STATE Let $\hat{h}_i:[\Delta]^d\rightarrow \{0,1\}$ be a $\lambda$-wise independent hash function s.t. $\forall p\in [\Delta]^d,$ $\Pr[\hat{h}_i(p)=1]=\phi_i$.
		    \STATE For each $P\in\mathcal{P}^I_i$ and for each $p\in P$, add $p$ into $Q'_i$ if $\hat{h}_i(p)=1$. \label{sta:sampling_step}
		\ENDFOR
	    \STATE {\bfseries Output:} $Q'=\bigcup_{i=0}^L Q'_i,w':Q'\rightarrow \mathbb{R}_{>0}$.
	\end{algorithmic}
\end{algorithm}



\begin{definition}\label{def:good_estimation_alg2}
If the estimated size $\tau\left(\bigcup_{j=1}^{s_i} Q_{i,j}\right)$ in Algorithm~\ref{alg:coreset_construction} satisfies either $\tau\left(\bigcup_{j=1}^{s_i} Q_{i,j}\right)\in \sum_{j=1}^{s_i}|Q_{i,j}| \pm 0.1 T_i(o)$ or $\tau\left(\bigcup_{j=1}^{s_i} Q_{i,j}\right)\in (1\pm 0.1)\sum_{j=1}^{s_i}|Q_{i,j}|$, then $\tau\left(\bigcup_{j=1}^{s_i} Q_{i,j}\right)$ is good.
If the estimated size $\tau(Q_{i,j})$ in Algorithm~\ref{alg:coreset_construction} satisfies either $\tau(Q_{i,j})\in |Q_{i,j}|\pm 0.1\gamma T_i(o)$ or $\tau(Q_{i,j})\in (1\pm 0.1)|Q_{i,j}|$, then $\tau(Q_{i,j})$ is good.
\end{definition}

Consider a solution of $k$-clustering problem. 
Each cluster can be seen as the set of points which are assigned to the same center.
Thus, a solution of $k$-clustering can be represented by an assignment mapping.
\begin{definition}[Assignment]\label{def:assignment}
Given a point set $P\subseteq[\Delta]^d$ where each point $p\in P$ has a weight $w(p)\in\mathbb{R}_{\geq 0}$, and a set of centers $Z\subset [\Delta]^d$ with $|Z|= k$ for some $k\in\mathbb{Z}_{\geq 1}$, an assignment of $P$ to the centers $Z$ is a mapping $\pi:P\rightarrow Z$. 
The clustering cost of $\pi$ is denoted as $\cost(\pi)=\sum_{p\in P}w(p)\cdot \dist^r(p,\pi(p))$.
The size vector $s(\pi)\in\mathbb{R}^k$ of clusters is defined as $\forall i\in[k],s(\pi)_i = \sum_{p\in P:\pi(p)=z_i} w(p)$.
\end{definition}

If we do not specify the weight function $w(\cdot)$ explicitly in the context, we suppose each point $p\in P$ has $w(p)=1$.

Next, we show that some assignment mapping can be defined by a set of half-spaces.
\begin{definition}[Assignment half-spaces]\label{def:assignment_halfspace}
Given a set of centers $Z=\{z_1,z_2,\cdots,z_k\}\subset[\Delta]^d$ with $|Z|=k$ for some $k\in\mathbb{Z}_{\geq 1}$, a set of assignment half-spaces $\mathcal{H}$ corresponding to $Z$ has $k\choose 2$ half-spaces (Definition~\ref{def:halfspace}), i.e., 
\begin{align*}
\mathcal{H}=\left\{H_{(i,j)}\mid 1\leq i<j\leq k\right\},
\end{align*}
where $H_{(i,j)}$ is a half-space corresponding to $\left(z_i,z_j,t_{(i,j)}\right)$ for some integer $t_{(i,j)}\in \left[\Delta^d\right]$. 
For $j<i$, we denote $H_{(i,j)}$ as $[\Delta]^d\setminus H_{(j,i)}$.
For a point set $P\subseteq[\Delta]^d$, if $\forall p\in P$, there always exists a unique $i\in[k]$ such that $\forall j\in[k],j\not =i,$ it has $p\in H_{(i,j)}$, we say $\mathcal{H}$ is valid for $P$, and the assignment mapping $\pi:P\rightarrow Z$ corresponding to $\mathcal{H}$ is defined as:
\begin{align*}
\begin{array}{cc}
\forall p\in P, \pi(p) = z_i, &  \text{where }i\text{ satisfies $\forall j\not =i, p\in H_{(i,j)}$}.
\end{array}
\end{align*}
\end{definition}

Consider the assignment mapping of capacitated clustering problem for point set $Q$ and centers $Z$.
We can always find a set of half-spaces such that we can use these half-spaces to determine the assigned center for each point $p\in Q$ without looking at other points in $Q$.
In the following, we formalize the argument and extend it to the weighted case.
\begin{lemma}[Cost and assignment half-spaces]\label{lem:opt_halfspace}
Consider a point set $Q$ with at most $m$ different weights, i.e., $Q=Q_1\dun Q_2\dun\cdots \dun Q_m\subseteq[\Delta]^d$, where $\forall i\in[m],p\in Q_i,$ $p$ has a weight $w(p)=w_i$.
For any set of centers $Z=\{z_1,z_2,\cdots,z_k\}\subset [\Delta]^d$ with $|Z|=k$ for some $k\in\mathbb{Z}_{\geq 1}$, and any $t\geq 0$ with $\cost_t(Q,Z,w)\not=\infty$, there always exist $m$ sets of assignment half-spaces $\mathcal{H}^{(1)},\mathcal{H}^{(2)},\cdots,\mathcal{H}^{(m)}$ corresponding to $Z$ such that $\forall i\in[m]$, $\mathcal{H}^{(i)}$ is valid for $Q_i$, and  
\begin{align*}
\begin{array}{ccc}
\cost_t(Q,Z,w) = \sum_{i=1}^m \cost(\pi_i)    &  \text{and} & \|\sum_{i=1}^m s(\pi_i)\|_{\infty}\leq t,
\end{array}
\end{align*}
where $\pi_i:Q_i\rightarrow Z$ is an assignment mapping corresponding to $\mathcal{H}^{(i)}$.
\end{lemma}
\begin{proof}
Let $\pi^*:Q\rightarrow Z$ be an optimal assignment mapping, i.e., $\cost_t(Q,Z,w)=\cost(\pi^*)$.
For $l\in[m]$, let us construct $\mathcal{H}^{(l)}=\{H^{(l)}_{(i,j)}\mid i<j\in[k]\}$ as the following.
Consider $Q_l$, $z_i,$ and $z_j$ $(i<j)$.
We can sort $Q_l=\{p_1,p_2,\cdots,p_{|Q_{l}|}\}$ such that $\forall a\in [|Q_l|-1]$, either $\dist^r(p_a,z_i)-\dist^r(p_{a},z_j)<\dist^r(p_{a+1},z_i)-\dist^r(p_{a+1},z_j)$ or $\dist^r(p_a,z_i)-\dist^r(p_{a},z_j)=\dist^r(p_{a+1},z_i)-\dist^r(p_{a+1},z_j)$ and $p_a$ is smaller than $p_{a+1}$ in the alphabetic order.
Consider the largest $a$ such that $\pi^*(p_a)=z_i$ and the smallest $a'$ such that $\pi^*(p_{a'})=z_j$.
\begin{claim}\label{cla:able_to_switch}
If $a>a'$, then 
\begin{align*}
w(p_a)\cdot \dist^r(p_a,z_j) + w(p_{a'})\cdot \dist^r(p_{a'},z_i) = w(p_a)\cdot \dist^r(p_a,z_i) + w(p_{a'})\cdot \dist^r(p_{a'},z_j),
\end{align*}
and the alphabetic order of $p_{a'}$ is smaller than $p_{a}$.
\end{claim}
\begin{proof}
Since both $p_a$ and $p_{a'}$ are from $Q_l$, we know that $w(p_a)=w(p_{a'})$.
Since $a>a'$, we have:
\begin{align*}
\dist^r(p_{a'},z_i)-\dist^r(p_{a'},z_j)\leq \dist^r(p_a,z_i)-\dist^r(p_a,z_j).
\end{align*}
If 
\begin{align*}
\dist^r(p_{a'},z_i)-\dist^r(p_{a'},z_j)< \dist^r(p_a,z_i)-\dist^r(p_a,z_j),
\end{align*}
then 
\begin{align*}
w(p_a)\cdot \dist^r(p_a,z_j) + w(p_{a'})\cdot \dist^r(p_{a'},z_i) < w(p_a)\cdot \dist^r(p_a,z_i) + w(p_{a'})\cdot \dist^r(p_{a'},z_j)
\end{align*}
which implies that if we switch the assignment of $p_a$ and $p_{a'}$, we can get a better solution, and thus it contradicts to that $\pi^*$ is the optimal assignment.
\end{proof}
If $a< a'$, we can find a half-space $H^{(l)}_{(i,j)}$ such that $\forall p\in Q_{l}$ with $\pi^*(p)=z_i$ satisfies $p\in H^{(l)}_{(i,j)}$ and $\forall p\in Q_{l}$ with $\pi^*(p)=z_j$ satisfies $p\in H^{(l)}_{(j,i)}$.
Otherwise, by Claim~\ref{cla:able_to_switch} we can switch the assignment of $p_a$ and $p_{a'}$ which neither increases the cost nor changes the number of points assigned to each center, and we can try to construct $\mathcal{H}^{(l)}$ 
to the switched assignment mapping.
Notice that the switching decreases the summation of the alphabetic ranks of the points assigned to $z_i$.
Therefore, the switching operation will terminate.
\end{proof}

As discussed in the above lemma, the assigned center of a point may be determined by a set of half-spaces. 
We can define the set of points (may not be in $Q$) which should be assigned to the same center according to the half-spaces as a region.
Most regions should be the intersection of several half-spaces. 
However in some situations, the assignment half-spaces may not be valid for the underlying input points, and thus there are some points which can not be assigned to any center according to the half-spaces.
In this case, we define an additional region for these points.
\begin{definition}[Regions induced by assignment half-spaces]\label{def:regions}
Given a set of assignment half-spaces $\mathcal{H}=\left\{H_{(i,j)}\mid i<j\in[k]\right\}$, if $R_0=\{p\in[\Delta]^d\mid \forall i\in[k],\exists j\not=i,p\not\in H_{(i,j)}\}$ and $\forall i\in[k], R_i=\{p\in[\Delta]^d\mid \forall j\in[k], p\in H_{(i,j)}\}$, then $(R_0,R_1,\cdots,R_k)$ are regions induced by $\mathcal{H}$.
\end{definition}

Consider a set of points $P$ and a set of assignment half-spaces $\mathcal{H}$. 
Let $(R_0,R_1,\cdots,R_k)$ be regions induced by $\mathcal{H}$.
$\mathcal{H}$ may not be valid for $P$ or there may be some region $R_i$ for $i\in[k]$ such that $R_i\cap P\not=\emptyset$ but $|P\cap R_i|$ is small.
In this case, we want to find an assignment mapping which is almost determined by $\mathcal{H}$ and each non-empty cluster is large enough.

To achieve this, we can check each point $p\in P$. 
If $p\in R_0$ or $p\in R_i$ for some $i$ such that $|P\cap R_i|$ is small, we assign $p$ to $z_{i^*}$, where $i^*$ satisfies that there are lots of points in region $R_{i^*}$ and $i^*\not=0$.
Notice that though there may be no assignment mapping for $P$ corresponding to $\mathcal{H}$ since $\mathcal{H}$ may be invalid for $P$, we can always define a transferred assignment mapping for $P$ according to $\mathcal{H}$.
\begin{definition}[Assignment transfer]\label{def:assign_transfer}
Given a threshold $T\in\mathbb{R}_{\geq 0}$, let $P\subset [\Delta]^d$ be a point set where each point $p\in P$ has a weight $w(p)\in\mathbb{R}_{\geq 0}$ such that $\sum_{p\in P} w(p)\geq 0.9T$.
Consider a set of centers $Z=\{z_1,z_2,\cdots,z_k\}\subset [\Delta]^d$ with $|Z|=k$ for some $k\in\mathbb{Z}_{\geq 1}$, a set of assignment half-spaces $\mathcal{H}=\{H_{(i,j)}\mid i<j\in[k]\}$ corresponding to $Z$, and $B=(b_0,b_1,\cdots,b_k)\in \mathbb{R}^{k+1}_{\geq 0}$ such that $\forall i\in\{0,1,\cdots,k\}, b_i$ satisfies either $b_i\in (1\pm \xi)\cdot \sum_{p\in R_i\cap P} w(p)$ or $b_i\in \sum_{p\in R_i\cap P} w(p)\pm \xi T$, where $\xi\in(0,0.5)$ and $(R_0,R_1,\cdots,R_k)$ are regions induced by $\mathcal{H}$.
Let $i^*=\arg\max_{i\in [k]} b_i$.
A transferred assignment mapping $\pi:P\rightarrow Z$ corresponding to $(\mathcal{H},B,\xi,T)$ is defined as:
\begin{align*}
\forall p\in P, \pi(p)=\left\{
\begin{array}{ll}
 z_i,     &   i\in[k]: b_i\geq 2\xi T,p\in R_i,\\
 z_{i^*},     & \text{otherwise.}
\end{array}
\right.
\end{align*}
\end{definition}
Similar to Definition~\ref{def:assignment}, if we do not specify the weights $w(\cdot)$, each point has weight $1$.

The following lemma shows that if assignment half-spaces $\mathcal{H}$ is valid for the point set $P$, then the cost of the transferred assignment mapping is close to the cost of the assignment mapping corresponding to $\mathcal{H}$, and furthermore, the number of points of which centers are changed is small.
\begin{lemma}[Transferred assignment does not change the cost too much]\label{lem:transferred_also_good}
Given a threshold $T\in \mathbb{R}_{\geq 0}$, let $P\subseteq[\Delta]^d$ be a point set where each point $p\in P$ has a weight $w(p)\in\mathbb{R}_{\geq 0}$ such that $\sum_{p\in P} w(p)\geq 0.9T$, and $\forall p,q\in P,\dist(p,q)\leq \sqrt{d}g$ for some $g\in\mathbb{R}_{\geq 0}$. 
Consider a set of centers $Z=\{z_1,z_2,\cdots,z_k\}\subset[\Delta]^d$ with $|Z|=k$ for some $k\in\mathbb{Z}_{\geq 1}$, a set of assignment half-spaces $\mathcal{H}$ corresponding to $Z$, and $B=(b_0,b_1,\cdots,B_k)\in \mathbb{R}_{\geq 0}^{k+1}$ which satisfies the condition mentioned in Definition~\ref{def:assign_transfer} for some $\xi\in(0,1/(100k))$. 
If $\mathcal{H}$ is valid for $P$, then 
\begin{align*}
\begin{array}{lll}
\cost(\pi')\leq (1+2^{r+4}k^2\cdot \xi)\cost(\pi) +\xi\cdot 2^{r+1}kT(\sqrt{d}g)^r  &\text{and}& \|s(\pi')-s(\pi)\|_1\leq 16k\xi\sum_{p\in P} w(p),
\end{array}
\end{align*}
where $\pi:P\rightarrow Z$ is an assignment mapping corresponding to $\mathcal{H}$, and $\pi':P\rightarrow Z$ is a transferred assignment mapping corresponding to $(\mathcal{H},B,\xi,T)$. 
\end{lemma}
\begin{proof}
Let $i^*=\arg\max_{i\in[k]} b_i$.
Let $(R_0,R_1,\cdots,R_k)$ be the regions induced by $\mathcal{H}$.
Since $\mathcal{H}$ is valid for $P$, $R_0\cap P = \emptyset$.
By Pigeonhole Principle, we know that 
\begin{align*}
\max_{i\in[k]}\sum_{p\in R_i \cap P} w(p)\geq \frac{0.9T}{k} \geq \frac{T}{2k}.
\end{align*}
Thus,
\begin{align*}
b_{i^*}\geq \min\left(\frac{T}{2k}-\xi T,\frac{1}{2}\cdot \frac{T}{2k}\cdot T\right)\geq  \frac{T}{4k},
\end{align*}
where the second inequality follows from $\xi\leq 1/(100k)$.
Therefore,
\begin{align}\label{eq:lb_of_region_of_istar}
\sum_{p\in P:\pi(p)=z_{i^*}} w(p) = \sum_{p\in R_{i^*}\cap P}w(p)\geq \min(b_{i^*}-\xi T,b_{i^*}/2)\geq \frac{T}{8k}.
\end{align}
We have
\begin{align*}
\cost(\pi')-\cost(\pi) & \leq \sum_{p\in P:\pi'(p)\not=\pi(p)} w(p) \cdot \dist^r(p,\pi'(p))\\
& = \sum_{p\in P:\pi'(p)\not=\pi(p)} w(p)\cdot \dist^r(p,z_{i^*})\\
& \leq \sum_{p\in P:\pi'(p)\not=\pi(p)}w(p)\cdot \left(2^{r-1} (\sqrt{d}g)^r +  2^{r-1} \frac{\sum_{q\in P:\pi(q)=z_{i^*}}w(q)\cdot\dist^r(q,z_{i^*})}{\sum_{q\in P:\pi(q)=z_{i^*}}w(q)}\right)\\
& \leq k\cdot 4\xi T\cdot  \left(2^{r-1} (\sqrt{d} g)^r +  2^{r-1} \frac{\sum_{q\in P:\pi(q)=z_{i^*}}w(q)\cdot\dist^r(q,z_{i^*})}{\sum_{q\in P:\pi(q)=z_{i^*}}w(q)}\right)\\
& \leq \xi \cdot \left(2^{r+1} kT(\sqrt{d}g)^r + 2^{r+4} k^2\cdot \sum_{q\in P:\pi(q)=z_{i^*}}w(q)\cdot\dist^2(q,z_{i^*})\right)\\
& \leq \xi \cdot \left(2^{r+1} kT(\sqrt{d}g)^r + 2^{r+4} k^2\cdot \cost(\pi)\right),
\end{align*}
where the third step follows from that there is a point $q'\in P$ such that $\pi(q')=z_{i^*}$ and  
\begin{align*}
\dist^r(q',z_{i^*})\leq \frac{\sum_{q\in P:\pi(q)=z_{i^*}}w(q)\cdot\dist^r(q,z_{i^*})}{\sum_{q\in P:\pi(q)=z_{i^*}}w(q)}
\end{align*}
by averaging argument, and $\dist^r(p,z_{i^*})\leq 2^{r-1}\dist^r(p,q')+2^{r-1}\dist^r(q',z_{i^*})$ (Fact~\ref{fac:approx_tri}), the forth step follows from $\sum_{\pi'(p)\not=\pi(p)} w(p)\leq \sum_{i\in\{0,1,2,\cdots,k\}\setminus \{i^*\}:b_i<2\xi T} \max(2b_i,b_i+\xi T)\leq k\cdot 4\xi T$, the fifth step follows from Equation~\eqref{eq:lb_of_region_of_istar}.

Now, let us consider $s(\pi')$.
We have
\begin{align*}
\|s(\pi') - s(\pi)\|_1 & \leq 2 \sum_{p\in P:\pi'(p)\not = \pi(p)} w(p)\\
& \leq  2\sum_{i\in\{0,1,2,\cdots,k\}\setminus \{i^*\}:b_i<2\xi T} \max(2b_i,b_i+\xi T)\\
& \leq  \xi\cdot 8kT\\
&\leq 16k \xi \sum_{p\in P} w(p),
\end{align*}
where the last step follows from $\sum_{p\in P} w(p)\geq 0.9T$.
\end{proof}

The following lemma is a concentration bound for summation of random variables with limited independence.
We need following lemma since we want to prove that our algorithm only needs limited independence.
If the fully independent random samples are allowed in the algorithm, Bernstein inequality will be enough for analysis.
\begin{lemma}[\cite{br94}]\label{lem:limited_independent_bound}
Consider an even integer $\lambda \geq 4$ and a random variable $X=\sum_{i=1}^n X_i$, where $X_i$ are $\lambda$-wise independent random variables taking values in $[0,M]$.
For any $a>0$, 
\begin{align*}
\Pr\left[|X-\mu| > a\right] \leq 8\left(\frac{\mu\lambda M+\lambda^2M^2}{a^2}\right)^{\lambda/2},
\end{align*}
where $\mu=\E[X]$.
\end{lemma}

Consider a set of points $P$ and a set of assignment half-spaces $\mathcal{H}$.
Let $(R_0,R_1,\cdots,R_k)$ be regions induced by $\mathcal{H}$.
We randomly sample a subset of points $P'$ from $P$. 
We can use the number of points sampled from each region $R_i$ to estimate the number of points in $R_i$.
We can also use the total number of sampled points to estimate the total number of points in $P$.
Based on the estimated number of points in each region, we can define transferred assignment mappings for both $P'$ and $P$.
We argue that the cost of the transferred assignment mapping for $P'$ is a good estimation of the transferred assignment mapping for $P$.
Furthermore, we can use the number of samples in $P'$ assigned to each center to estimate the number of points in $P$ assigned to each center.
\begin{lemma}[Estimating the cost of transferred assignment via sampling]\label{lem:core_for_each}
Given a threshold $T\in \mathbb{R}_{\geq 0},$ let $P\subseteq[\Delta]^d$ be a point set such that each point has weight $w(p)=1$ and $|P|\geq T$.
Furthermore, $\forall p,q\in P,$ $\dist(p,q)\leq\sqrt{d} g$ for some $g\in\mathbb{R}_{\geq 0}$.
Consider a set of centers $Z=\{z_1,z_2,\cdots,z_k\}\subset [\Delta]^d$ with $|Z|=k$ for some $k\in\mathbb{Z}_{\geq 1}$, and a set of assignment half-spaces $\mathcal{H}=\{H_{(i,j)}\mid i<j \in [k]\}$ corresponding to $Z$.
For some arbitrary $\xi,\delta\in(0,0.5)$,
let $P'$ be a random subset of $P$ such that each point $p\in P$ is chosen $\lambda$-wise independently with probability $\phi$, where $\lambda=100k \lceil\log(k/\delta)\rceil,\phi=\min\left(1,\frac{1000\cdot 2^{2r}\lambda}{\xi^3T}\right)$. 
Let each $p\in P'$ have weight $w'(p)=1/\phi$.
Let $(R_0,R_1,\cdots,R_k)$ be the regions (Definition~\ref{def:regions}) induced by $\mathcal{H}$.
Let $B=(b_0,b_1,\cdots,b_k)$ such that $\forall i\in\{0,1,\cdots,k\}$, $b_i=\sum_{p\in R_i\cap P'} w'(p)$.
With probability at least $1-\delta$, all of the following events happens:
\begin{enumerate}
\item $\forall i\in\{0,1\cdots,k\}$, either $b_i\in (1\pm \xi)\cdot |R_i\cap P|$ or $b_i\in |R_i\cap P|\pm \xi T$, \label{it:bi_is_good}\label{it:b0_is_good}
\item $\sum_{p\in P'} w'(p)\geq 0.9 T$, \label{it:preserving_the_most_weights}
\item  
$|\cost(\pi')-\cost(\pi)|\leq \xi\left(|P|(\sqrt{d}g)^r+\cost(\pi)\right)$, where both of $\pi:P\rightarrow Z$ and $\pi':P'\rightarrow Z$ are transferred assignment mappings corresponding to $(\mathcal{H},B,\xi,T)$, \label{it:sampling_estimating_cost}
\item $\|s(\pi)-s(\pi')\|_1\leq \xi |P|$. \label{it:sampling_does_not_change_to_much}
\end{enumerate}
\end{lemma}
\begin{proof}
We only need to consider the case when $\phi<1$.

Let us first consider event \ref{it:bi_is_good}. 
We have $\forall i\in\{0,1,\cdots,k\}$,
\begin{align*}
\E[b_i]=\sum_{p\in R_i\cap P} \phi\cdot \frac{1}{\phi} = |R_i\cap P|.
\end{align*}
Consider $i\in\{0,1,\cdots,k\}$.
If $|R_i\cap P|\leq T$, by Lemma~\ref{lem:limited_independent_bound}, 
\begin{align*}
\Pr[|b_i-|R_i\cap P||>\xi T] \leq 8 \left(\frac{|R_i\cap P|\lambda\cdot \frac{1}{\phi} + \lambda^2\cdot \frac{1}{\phi^2}}{\xi^2 T^2}\right)^{\lambda/2}\leq 8\left(\frac{T\lambda\cdot\frac{1}{\phi}+\lambda^2\cdot\frac{1}{\phi^2}}{\xi^2 T^2}\right)^{\lambda/2} \leq \frac{\delta}{10(k+1)},
\end{align*}
where the last inequality follows from $\lambda\geq 40\log(k/\delta)$ and $\phi\geq 10\lambda/(\xi^2T)$.
If $|R_i\cap P|> T$, by Lemma~\ref{lem:limited_independent_bound} again,
\begin{align*}
\Pr[|b_i-|R_i\cap P||> \xi |R_i\cap P|] \leq 8\left(\frac{|R_i\cap P|\lambda\cdot \frac{1}{\phi}+\lambda^2\cdot\frac{1}{\phi^2}}{\xi^2|R_i\cap P|^2}\right)^{\lambda/2}\leq \frac{\delta}{10(k+1)},
\end{align*}
where the last inequality also follows from $\lambda\geq 40\log(k/\delta)$ and $\phi\geq 10\lambda/(\xi^2T)$.
By taking union bound, event \ref{it:b0_is_good} happens with probability at least $1-\delta/10$.

Consider event \ref{it:preserving_the_most_weights}. We have
\begin{align*}
\E\left[\sum_{p\in P'} w'(p)\right] = \sum_{p\in P} \phi\cdot \frac{1}{\phi} = |P|.
\end{align*}
By Lemma~\ref{lem:useful_tool},
\begin{align*}
\Pr\left[\left||P|-\sum_{p\in P'} w'(p)\right|> 0.1 |P|\right] \leq 8\left(\frac{|P|\lambda\cdot\frac{1}{\phi}+\lambda^2\cdot\frac{1}{\phi^2}}{0.01|P|^2}\right)^{\lambda/2} \leq \frac{\delta}{10},
\end{align*}
where the last step follows from that $\lambda\geq 40\log(1/\delta)$, $\phi\geq 1000\lambda/T\geq 1000\lambda/|P|$.


Next, let us consider event \ref{it:sampling_estimating_cost}. 
Consider an arbitrary $\hat{B}=(\hat{b}_0,\hat{b}_1,\cdots,\hat{b}_k)$ which satisfies that $\forall i\in\{0,1,\cdots, k\},$ either $\hat{b}_i\in |R_i\cap P|\pm \xi T$ or $\hat{b}_i\in (1\pm \xi) |R_i\cap P|$.
Let $\hat{\pi}:P\rightarrow Z$ be a transferred assignment mapping corresponding to $(\mathcal{H},\hat{B},\xi,T)$.
Let $\hat{i}^* = \arg\max_{i\in[k]} \hat{b}_i$
\begin{claim}\label{cla:each_cluster_is_large}
$\forall i\in[k]$, either $s(\hat{\pi})_i=0$ or $s(\hat{\pi})_i\geq \xi T$.
\end{claim}
\begin{proof}
Consider $i\not = \hat{i}^*$.
If $\exists p\in P$ such that $\hat{\pi}(p)= z_i$, then by the definition of transferred assignment mapping, we have $\hat{b}_i\geq 2\xi T$ and thus $s(\hat{\pi})_i= |R_i\cap P|\geq \min(\hat{b}_i-\xi T,0.5 \hat{b}_i)\geq \xi T$. 
The above argument implies that $\forall i\not = \hat{i}^*$, either $s(\hat{\pi})_i=0$ or $s(\hat{\pi})\geq \xi T$.

For $\hat{i}^*$, if $\hat{b}_{\hat{i}^*}\geq 2\xi T$, then $s(\hat{\pi})_{\hat{i}^*}\geq |R_{\hat{i}^*}\cap P|\geq \min(\hat{b}_{\hat{i}^*}-\xi T,0.5 \hat{b}_{\hat{i}^*})\geq \xi T$.
If $\hat{b}_{\hat{i}^*}< 2\xi T$, then $\forall i\not =\hat{i}^*$, $\hat{b}_i<2\xi T$ which implies that $s(\hat{\pi})_{\hat{i}^*}=|P|\geq T\geq \xi T$.
\end{proof}

For $i\in[k],$ let $\hat{R}_i=\{p\in P\mid \hat{\pi}(p)=z_i\}$. 
We have
\begin{align*}
\E_{P'} \left[ \sum_{p\in P'\cap \hat{R}_i} w'(p)\cdot \dist^r(p,z_i) \right] = \sum_{p\in \hat{R}_i} \phi\cdot \frac{1}{\phi} \cdot \dist^r(p,z_i) = \sum_{p\in \hat{R}_i} \dist^r(p,z_i).
\end{align*}
We only need to handle the situation when $\hat{R}_i\not=\emptyset,$ i.e., $s(\hat{\pi})_i\geq \xi T$.
Consider two cases, the first case is that $\dist(z_i,\hat{R}_i)\leq \sqrt{d}g$.
In this case, by Lemma~\ref{lem:limited_independent_bound}, we have:
\begin{align*}
&\Pr\left[ \left| \sum_{p\in P'\cap \hat{R}_i} w'(p) \cdot \dist^r(p, z_i) - \sum_{p\in \hat{R}_i} \dist^r(p,z_i)\right| >  \xi |\hat{R}_i| (\sqrt{d}g)^r\right]\\
\leq & 8\left(\frac{\left(\sum_{p\in\hat{R}_i} \dist^r(p,z_i)\right)\cdot \lambda\cdot \left(\frac{1}{\phi}\cdot (\dist(z_i,\hat{R}_i)+\sqrt{d}g)^r\right) + \lambda^2 \cdot \left(\frac{1}{\phi}\cdot (\dist(z_i,\hat{R}_i)+\sqrt{d}g)^r\right)^2 }{\xi^2|\hat{R}_i|^2(\sqrt{d}g)^{2r}}\right)^{\lambda/2}\\
\leq & 8 \left( \frac{\frac{\lambda}{\phi}\cdot |\hat{R}_i|(\dist(z_i,\hat{R}_i)+\sqrt{d}g)^{2r}+\frac{\lambda^2}{\phi^2}\cdot (\dist(z_i,\hat{R}_i)+\sqrt{d}g)^{2r}}{\xi^2|\hat{R}_i|^2 (\sqrt{d}g)^{2r}} \right)^{\lambda/2}\\
\leq & 8 \left( \frac{\frac{\lambda}{\phi}\cdot |\hat{R}_i|(2\sqrt{d}g)^{2r}+\frac{\lambda^2}{\phi^2}\cdot (2\sqrt{d}g)^{2r}}{\xi^2|\hat{R}_i|^2 (\sqrt{d}g)^{2r}} \right)^{\lambda/2}\\
\leq & \frac{\delta}{ 10 k^2 2^k},
\end{align*}
where the first and the second step follows from triangle inequality, and the last step follows from that $\phi\geq 10\cdot 2^{2r}\lambda/(\xi^3T)\geq 10\cdot 2^{2r}\lambda/(\xi^2|\hat{R}_i|)$ and $\lambda \geq 100k\log(k/\delta)$.
The second case is that $\dist(z_i,\hat{R}_i)> \sqrt{d}g$.
In this case, by Lemma~\ref{lem:limited_independent_bound}, we have:
\begin{align*}
&\Pr\left[ \left| \sum_{p\in P'\cap \hat{R}_i} w'(p) \cdot \dist^r(p, z_i) - \sum_{p\in \hat{R}_i} \dist^r(p,z_i)\right| \geq  \xi \sum_{p\in \hat{R}_i}\dist^r(p,z_i)\right]\\
\leq & 8\left( \frac{ \left(\sum_{p\in\hat{R}_i}\dist^r(p,z_i)\right)\cdot  \lambda\cdot \left(\frac{1}{\phi}\cdot (\dist(z_i,\hat{R}_i)+\sqrt{d}g)^r\right) + \lambda^2\cdot \left(\frac{1}{\phi}\cdot \left(\dist(z_i,\hat{R}_i)+\sqrt{d}g\right)^r\right)^2 }{ \xi^2 \left(\sum_{p\in \hat{R}_i} \dist^r(p,z_i)\right)^2 } \right)^{\lambda/2} \\
\leq & 8\left(\frac{\frac{\lambda}{\phi}\cdot |\hat{R}_i|(\dist(z_i,\hat{R}_i)+\sqrt{d}g)^{2r}+\frac{\lambda^2}{\phi^2}\cdot (\dist(z_i,\hat{R}_i)+\sqrt{d}g)^{2r}}{\xi^2 |\hat{R}_i|^2 \dist^{2r}(z_i,\hat{R}_i)}\right)^{\lambda/2}\\
\leq & 8\left(\frac{\frac{\lambda}{\phi}\cdot |\hat{R}_i|(2\dist(z_i,\hat{R}_i))^{2r}+\frac{\lambda^2}{\phi^2}\cdot (2\dist(z_i,\hat{R}_i))^{2r}}{\xi^2 |\hat{R}_i|^2 \dist^{2r}(z_i,\hat{R}_i)}\right)^{\lambda/2}\\
\leq & \frac{\delta}{10k^22^k},
\end{align*}
where the first and the second step follows from triangle inequality, and the last step follows from $\phi\geq 10\cdot 2^{2r}\lambda/(\xi^3T)\geq 10\cdot 2^{2r}\lambda/(\xi^2|\hat{R}_i|)$ and $\lambda\geq 100k\log(k/\delta)$.

Thus, we know that with probability at least $1-\frac{\delta}{10k^2 2^k}$, 
\begin{align*}
\left|\sum_{p\in P'\cap \hat{R}_i} w'(p) \cdot \dist^r(p, z_i) - \sum_{p\in \hat{R}_i} \dist^r(p,z_i) \right| \leq \xi \left(|\hat{R}_i|(\sqrt{d}g)^r + \sum_{p\in\hat{R}_i}\dist^r(p,z_i)\right).
\end{align*}
By taking union bound over $i\in[k]$, with probability at least $1-\frac{\delta}{10k2^k}$, we have:
\begin{align*}
&\left|\sum_{p\in P'} w'(p) \dist^r(p,\hat{\pi}(p))-\cost(\hat{\pi})\right|\\
\leq & \sum_{i=1}^k \left|\sum_{p\in P'\cap \hat{R}_i} w'(p)\cdot \dist^r(p,z_i)-\sum_{p\in\hat{R}_i}\dist^r(p,z_i)\right|\\
\leq & \sum_{i=1}^k \xi \left(|\hat{R}_i|(\sqrt{d}g)^r + \sum_{p\in\hat{R}_i}\dist^2(p,z_i)\right)\\
\leq & \xi \left(|P|(\sqrt{d}g)^r + \cost(\hat{\pi})\right).
\end{align*}
Notice that though different $\hat{B}$ may induce a different assignment mapping $\hat{\pi}$, the total number of possible $\hat{\pi}$ cannot be too large.
This is because $\hat{i}^*$ only has $k$ choices and for $i\in \{0,1,\cdots,k\}\setminus\{\hat{i}^*\},$ either every point $p\in R_i\cap P$ is assigned to $z_{\hat{i}^*}$ or every point $p\in R_i\cap P$ is assigned to $z_i$.
Based on this observation, the total number of different assignment mapping $\hat{\pi}$ is upper bounded by $k\cdot 2^k$. 
By taking union bound over all possible $\hat{\pi}$, with probability at least $1-\delta/10$, for any possible choice of $\hat{B}$, we have 
\begin{align*}
\left|\sum_{p\in P'} w'(p)\dist^r(p,\hat{\pi}(p))-\cost(\hat{\pi})\right|\leq \xi \left(|P|(\sqrt{d}g)^r+\cost(\hat{\pi})\right).
\end{align*}
Condition on event \ref{it:bi_is_good}, $B$ is a possible choice of $\hat{B}$. 
Thus, with probability at least $1-\delta/10$,
\begin{align*}
\left|\sum_{p\in P'} w'(p)\dist^r(p,\pi(p))-\cost(\pi)\right|\leq \xi \left(|P|(\sqrt{d}g)^r+\cost({\pi})\right).
\end{align*}
Notice that, by the construction of $\pi$ and $\pi'$, $\forall p\in P',$ we have $\pi'(p)=\pi(p)$.
Thus, with probability at least $1-\delta/10$,
\begin{align*}
|\cost(\pi')-\cost(\pi)|\leq \xi(|P|(\sqrt{d}g)^r+\cost(\pi))
\end{align*}

Finally, let us consider event \ref{it:sampling_does_not_change_to_much}.
Similar to the argument for event \ref{it:sampling_estimating_cost}, let us still consider an arbitrary $\hat{B}=(\hat{b}_0,\hat{b}_1,\cdots,\hat{b}_k)$ which satisfies that $\forall i\in\{0,1,\cdots, k\},$ either $\hat{b}_i\in |R_i\cap P|\pm \xi T$ or $\hat{b}_i\in (1\pm \xi) |R_i\cap P|$.
Let $\hat{\pi}:P\rightarrow Z$ be a transferred assignment mapping corresponding to $(\mathcal{H},\hat{B},\xi,T)$.
Let $\hat{i}^* = \arg\max_{i\in[k]} \hat{b}_i$.
For $i\in[k],$ let $\hat{R}_i=\{p\in P\mid \hat{\pi}(p)=z_i\}$. 
We have
\begin{align*}
\E_{P'}\left[\sum_{p\in P'\cap \hat{R}_i} w'(p)\right] = \sum_{p\in \hat{R}_i} \phi\cdot \frac{1}{\phi} = |\hat{R}_i|= s(\hat{\pi})_i.
\end{align*}
By Claim~\ref{cla:each_cluster_is_large}, $s(\hat{\pi})_i$ is either $0$ or at least $\xi T$.
We only need to handle the situation when $s(\hat{\pi})_i$ is at least $\xi T$.
By Lemma~\ref{lem:limited_independent_bound},
\begin{align*}
    &\Pr\left[\left|\sum_{p\in P'\cap \hat{R}_i} w'(p) - s(\hat{\pi})_i\right| > \xi s(\hat{\pi})_i\right]\\
\leq & 8\left(\frac{|\hat{R}_i|\cdot \lambda\cdot \frac{1}{\phi} + \lambda^2\cdot \frac{1}{\phi^2}}{\xi^2 |\hat{R}_i|^2}\right)^{\lambda/2}\\
\leq & \frac{\delta}{10k^2 2^k},   
\end{align*}
where the last step follows from that $\phi\geq 10\lambda/(\xi^3T)\geq 10\lambda/(\xi^2|\hat{R}_i|)$ and $\lambda\geq 100k\log(k/\delta)$.
By taking union bound over all $i\in [k]$, with probability at least $1-\frac{\delta}{10k2^k}$,
\begin{align*}
\sum_{i=1}^k\left|\sum_{p\in P'\cap \hat{R}_i} w'(p) - s(\hat{\pi})_i \right| \leq \xi\cdot \sum_{i=1}^k s(\hat{\pi})_i = \xi |P|.
\end{align*}
By taking union bound over all the possible different $\hat{\pi}$, with probability at least $1-\delta/10$, for any possible choice of $\hat{B}$,
we have
\begin{align*}
\sum_{i=1}^k\left|\sum_{p\in P'\cap \hat{R}_i} w'(p) - s(\hat{\pi})_i \right| \leq \xi |P|.
\end{align*}
Condition on event \ref{it:bi_is_good}, $B$ is a possible choice of $\hat{B}$. Thus, with probability at lest $1-\delta/10$,
\begin{align*}
\sum_{i=1}^k\left|\sum_{p\in P'\cap \hat{R}_i} w'(p) - s(\pi)_i \right| \leq \xi |P|.
\end{align*}
Notice that, by the construction of $\pi$ and $\pi'$, $\forall p\in P',$ we have $\pi'(p)=\pi(p)$.
Thus, with probability at least $1-\delta/10$, $\|s(\pi)-s(\pi')\|_1\leq \xi |P|$.

By taking union bound over all four events, we complete the proof.
\end{proof} 

Notice that the previous lemma only works for $P$ and a fixed set of assignment half-spaces.
To make the above argument work for $P$ and all sets of assignment half-spaces, we just need to slightly raise the sampling probability and take union bound over all possible sets of assignment half-spaces.

\begin{lemma}[Estimation is good for all choices of centers and half-spaces]\label{lem:core_for_all}
Given a threshold $T\in\mathbb{R}_{\geq 0}$, let $P\subseteq[\Delta]^d$ be a point set such that each point has weight $w(p)=1$ and $|P|\geq T$.
Furthermore, $\forall p,q\in P$, $\dist(p,q)\leq \sqrt{d}g$ for some $g\in\mathbb{R}_{\geq 0}$.
Let $\xi,\delta\in(0,0.5),k\in\mathbb{Z}_{\geq 1},L=\log\Delta$.
Let $P'$ be a random subset of $P$ such that each point $p\in P$ is chosen $\lambda$-wise independently with probability $\phi$, where $\lambda=2000k^3dL\lceil\log(k/\delta)\rceil,\phi=\min\left(1, \frac{1000\cdot 2^{2r}\lambda}{\xi^3T}\right)$.
Let each $p\in P'$ have weight $w'(p)=1/\phi$. 
With probability at least $1-\delta/2$, the following event happens:
\begin{itemize}
\item For any choice of centers $Z=\{z_1,z_2,\cdots,z_k\}\subset[\Delta]^d$ with $|Z|=k$ and any choice of a set of assignment half-spaces $\mathcal{H}=\{H_{(i,j)}\mid i<j \in[k]\}$ corresponding to $Z$, $B^{Z,\mathcal{H}}=(b^{Z,\mathcal{H}}_0,b^{Z,\mathcal{H}}_1,\cdots,b^{Z,\mathcal{H}}_k)$ satisfies
that
$\forall i\in\{0,1,\cdots,k\}$, either $b^{Z,\mathcal{H}}_i\in (1\pm \xi)\cdot |R_i^{Z,\mathcal{H}}\cap P|$ or $b^{Z,\mathcal{H}}_i\in  |R_i^{Z,\mathcal{H}}\cap P|\pm \xi T$,
where $(R_0^{Z,\mathcal{H}},R_1^{Z,\mathcal{H}},\cdots,R_k^{Z,\mathcal{H}})$ are regions induced by $\mathcal{H}$ (Definition~\ref{def:regions}), and $\forall i\in\{0,1,\cdots,k\},b_i^{Z,\mathcal{H}}=\sum_{p\in R_i^{Z,\mathcal{H}}\cap P'} w'(p)$ , and furthermore,
\begin{align*}
\begin{array}{lll}
&|\cost(\pi'_{Z,\mathcal{H}})-\cost(\pi_{Z,\mathcal{H}})|\leq \xi(|P|(\sqrt{d}g)^r+\cost(\pi_{Z,\mathcal{H}}))\\
& \text{and}\\
& \|s(\pi'_{Z,\mathcal{H}})-s(\pi_{Z,\mathcal{H}})\|_1\leq \xi |P|,
\end{array}
\end{align*}
where both $\pi_{Z,\mathcal{H}}:P\rightarrow Z,\pi'_{Z,\mathcal{H}}:P'\rightarrow Z$  are transferred assignment mappings corresponding to $(\mathcal{H},B^{Z,\mathcal{H}},\xi,T)$.
\end{itemize}
In addition, with probability at least $1-\delta/2$, $\sum_{p\in P'} w'(p)\geq 0.9 T$.
\end{lemma}
\begin{proof}
By event \ref{it:preserving_the_most_weights} of Lemma~\ref{lem:core_for_each}, with probability at least $1-\delta/10$, $\sum_{p\in P'} w'(p)\geq 0.9T$.

By Lemma~\ref{lem:core_for_each} again, for any fixed $Z$ and $\mathcal{H}$, the second event happens with probability at least $1-\frac{\delta}{10\Delta^{2dk^2}}$, the second event happens.
Since $Z\subset[\Delta]^d$ and $|Z|=k$, the total number of possible choices of $Z$ is at most $\Delta^{dk}$.
For a particular choice of $Z$, since $\mathcal{H}$ contains ${k\choose 2} \leq k^2$ half-spaces and there are $\Delta^d$ possible choices for each half-space, the total number of possible $\mathcal{H}$ corresponding to $Z$ is at most $\Delta^{dk^2}$.
Hence the total number of possible choices of $(Z,\mathcal{H})$ is at most $\Delta^{dk}\cdot\Delta^{dk^2}\leq \Delta^{2dk^2}$.
By taking union bound over all possible choices of $(Z,\mathcal{H})$, with probability at least $1-\delta/10$, the second event happens.
\end{proof}

It is enough to prove the correctness of our algorithm. 
In Lemma~\ref{lem:offline_success_prob}, we will prove that if $o$ is a good approximation of $\OPT^{(r)}_{k\text{-clus}}$, Algorithm~\ref{alg:coreset_construction} does not output FAIL with high probability.
In the following lemma, we suppose this happens.
Furthermore, we can assume that the estimated size $\tau(C\cap Q),\tau\left(\bigcup_{j=1}^{s_i} Q_{i,j}\right),\tau(Q_{i,j})$ in Algorithm~\ref{alg:heavy_light} and Algorithm~\ref{alg:coreset_construction} are good estimations of $|C\cap Q|,~\left|\bigcup_{j=1}^{s_i} Q_{i,j}\right|$ and $|Q_{i,j}|$ respectively.
For offline algorithm, it is easy to compute the exact value of $|C\cap Q|,~\left|\bigcup_{j=1}^{s_i} Q_{i,j}\right|$ and $|Q_{i,j}|$.
For streaming and distributed algorithm, we will explain how to get good estimations with high probability in Section~\ref{sec:streaming}.
In the following lemma, we will show that if $o$ is a suitable choice, then we can output a strong coreset with high probability.

In high level, to prove the correctness, we will apply Lemma~\ref{lem:transferred_also_good} and Lemma~\ref{lem:core_for_all} for each part $Q_{i,j}$ of which $|Q_{i,j}|$ is $\Omega(\gamma T_i(o))$. 
We can show that the total error induced by each part $Q_{i,j}$ is relatively small.

\begin{lemma}[Correctness of the construction]\label{lem:correctness_offline}
Suppose the following conditions are satisfied:
\begin{enumerate}
    \item Algorithm~\ref{alg:coreset_construction} does not output FAIL,
    \item $o\leq \OPT^{(r)}_{k\text{-clus}}$,
    \item all estimated size $\tau\left(\bigcup_{j=1}^{s_i} Q_{i,j}\right)$, $\tau(Q_{i,j})$ in Algorithm~\ref{alg:coreset_construction} and $\tau(C\cap Q)$ in Algorithm~\ref{alg:heavy_light} called by Algorithm~\ref{alg:coreset_construction} are good (Definition~\ref{def:good_estimation_alg2}, Definition~\ref{def:good_estimation_alg1}).
\end{enumerate}
Let $Q'\subseteq Q,w':Q'\rightarrow \mathbb{R}_{\geq 0}$ be the output of Algorithm~\ref{alg:coreset_construction}.
With probability at least $0.99$, $\forall t\geq |Q|/k,Z\subset[\Delta]^d$ with $|Z|=k$,
\begin{align*}
\begin{array}{ccc}
\cost_{(1+\eta)t}(Q,Z) \leq (1+\varepsilon)\cost_t(Q',Z,w')    &  \text{and} & \cost_{(1+\eta)t}(Q',Z,w')\leq (1+\varepsilon)\cost_t(Q,Z). 
\end{array}
\end{align*}
\end{lemma}
\begin{proof}
By line~\ref{sta:lb_of_each_final_part} of Algorithm~\ref{alg:coreset_construction}, $\forall i\in\{0,1,\cdots,L\},\forall P\in\mathcal{P}^I_i$, we have $|P|\geq \min(0.9\tau(P),\tau(P)-0.1\gamma T_i(o))\geq  0.5\gamma T_i(o)$.
Consider an arbitrary $P\in \mathcal{P}^I_i$.
Since $\forall p,q\in P, c_{i-1}(p)=c_{i-1}(q)$, we have $\dist(p,q)\leq \sqrt{d}g_{i-1}=\sqrt{d}\cdot 2g_i$.
Let $P'=P\cap Q'$.
Let $\mathcal{E}(P)$ be the following events:
\begin{enumerate}
\item $\sum_{p\in P'} w'(p)\geq 0.9\cdot 0.5\gamma T_i(o)=0.45\gamma T_i(o)$.
\item For any choice of centers $Z=\{z_1,z_2,\cdots,z_k\}\subset [\Delta]^d$ with $|Z|=k$ and any choice of a set of assignment half-spaces $\mathcal{H}=\{H_{(j,j')}\mid j<j'\in[k]\}$ corresponding to $Z$, $B^{P,Z,\mathcal{H}}=(b_0^{P,Z,\mathcal{H}},b_1^{P,Z,\mathcal{H}},\cdots,b_k^{P,Z,\mathcal{H}})$ satisfies that
$\forall j\in\{0,1,\cdots,k\}$, ether $b_j^{P,Z,\mathcal{H}}\in(1\pm \xi)\cdot |R_j^{Z,\mathcal{H}}\cap P|$ or $b_j^{P,Z,\mathcal{H}}\in |R_j^{Z,\mathcal{H}}\cap P|\pm\xi\cdot 0.5\gamma T_i(o)$, where $(R_0^{Z,\mathcal{H}},R_1^{Z,\mathcal{H}},\cdots,R_k^{Z,\mathcal{H}})$ are regions (Definition~\ref{def:regions}) induced by $\mathcal{H}$, and $b_j^{P,Z,\mathcal{H}}=\sum_{p\in P\cap Q'\cap R_j^{Z,\mathcal{H}}} w'(p)$,
and furthermore,
\begin{align*}
&|\cost(\pi'_{P',Z,\mathcal{H}})-\cost(\pi_{P,Z,\mathcal{H}})|\leq \xi (|P|(\sqrt{d}\cdot 2g_i)^r+\cost(\pi_{P,Z,\mathcal{H}}))     \\
&\text{and} \\
& \|s(\pi'_{P',Z,\mathcal{H}})-s(\pi_{P,Z,\mathcal{H}})\|_1\leq \xi |P|,
\end{align*}
where both $\pi_{P,Z,H}:P\rightarrow Z,\pi'_{P',Z,H}:P'\rightarrow Z$ are transferred assignment mappings corresponding to $(\mathcal{H},B^{P,Z,\mathcal{H}},\xi,T)$.
\end{enumerate}
By Lemma~\ref{lem:core_for_all}, with probability at least $1-1/(10^8(k+d^{1.5r})L)$, $\mathcal{E}(P)$ happens.
Notice that $\sum_{i=0}^L|\mathcal{P}^I_i|\leq \sum_{i=0}^{L}s_i\leq 20000(k+d^{1.5r})L$.
By taking union bound over all $i\in\{0,1,\cdots,L\},P\in \mathcal{P}^I_i$, with probability at least $0.99$, $\forall i\in\{0,1,\cdots,L\},P\in\mathcal{P}^I_i$, event $\mathcal{E}(P)$ happens.
In the remaining of the proof, we condition on all $\mathcal{E}(P)$.

Firstly, let us focus on proving $\cost_{(1+\eta)t}(Q',Z,w')\leq (1+\varepsilon)\cost_t(Q,Z).$
Let $Q^I=\bigcup_{i=0}^L \bigcup_{P\in\mathcal{P}^I_i} P$.
According to Lemma~\ref{lem:opt_halfspace}, there is a set of assignment half-spaces $\mathcal{H}^I$ corresponding to $Z$ such that $\mathcal{H}^I=\{H^I_{(j,j')}\mid j<j'\in [k]\}$ is valid for $Q^I$ and 
\begin{align}\label{eq:costQI_is_equal_to_costpiI}
\cost_t(Q^I,Z)=\cost(\pi^I),
\end{align}
where $\pi^I:Q^I\rightarrow Z$ is an assignment mapping corresponding to $\mathcal{H}^I$.
For $i\in\{0,1,\cdots,L\},P\in\mathcal{P}^I_i$, let $\pi_P^I:P\rightarrow Z$ be the assignment mapping such that $\forall p\in P, \pi^I_P(p)=\pi^I(p)$.
Let $\hat{\pi}^I_P:P\rightarrow Z$ be a transferred assignment mapping corresponding to $(\mathcal{H}^I,B^{P,Z,\mathcal{H}^I},\xi,0.5\gamma T_i(o))$, where $B^{P,Z,\mathcal{H}^I}=(b_0^{P,Z,\mathcal{H}^I},b_1^{P,Z,\mathcal{H}^I},\cdots,b_k^{P,Z,\mathcal{H}^I})$ is the same as defined in the event $\mathcal{E}(P)$, i.e., $\forall j\in\{0,1,\cdots,k\},b_j^{P,Z,\mathcal{H}^I}=\sum_{p\in P\cap Q'\cap R_j^{Z,\mathcal{H}^I}}w'(p)$, where $(R_0^{Z,\mathcal{H}^I},R_1^{Z,\mathcal{H}^I},\cdots,R_k^{Z,\mathcal{H}^I})$ are regions (Definition~\ref{def:regions}) induced by $\mathcal{H}^I$. 
Let $\hat{\pi}^I_{P\cap Q'}:P\cap Q'\rightarrow Z$ be a transferred assignment mapping which is also corresponding to $(\mathcal{H}^I,B^{P,Z,\mathcal{H}^I},\xi,0.5\gamma T_i(o))$, where each point $p\in P\cap Q'$ has weight $w'(p)$.
We have:
\begin{align}\label{eq:long_eq1_part1}
\cost_t(Q,Z)
\geq \cost_t(Q^I,Z)
= \cost(\pi^I)
= \sum_{i=0}^L \sum_{P\in \mathcal{P}_i^I} \cost(\pi_P^I),
\end{align}
where the first step follows from that $Q^I$ is a subset of $Q$, the second step follows from Equation~\eqref{eq:costQI_is_equal_to_costpiI}.
Consider $i\in\{0,1,\cdots,L\}$ and $P\in\mathcal{P}^I$.
Algorithm~\ref{alg:coreset_construction} tells us $|P|\geq 0.5\gamma T_i(o)$.
Since all points in $P$ are in the same cell in $G_{i-1}$, we know that $\forall p,q\in P$, $\dist(p,q)\leq\sqrt{d}\cdot 2g_i$.
Then by Lemma~\ref{lem:transferred_also_good},
\begin{align}\label{eq:long_eq1_part2}
(1+2^{r+4}k^2\cdot \xi)\cdot \cost(\pi_P^I)\geq \cost(\hat{\pi}_P^I) -\xi \cdot 2^{r+1}k\cdot 0.5\gamma T_i(o)\cdot  (\sqrt{d}\cdot 2g_i)^r.
\end{align}
Since $2^{r+4}k^2\xi \leq \varepsilon/10$ and $\xi\leq \frac{\varepsilon}{2000\cdot 2^{2r}k(k+d^{1.5r})L}$, we have:
\begin{align}\label{eq:long_eq1_part3}
&(1+\varepsilon/10)\cdot\cost_t(Q,Z)\notag\\
 \geq& (1+\varepsilon/10)\cdot\sum_{i=0}^L \sum_{P\in \mathcal{P}_i^I} \cost(\pi_P^I) & \text{(Equation~\eqref{eq:long_eq1_part1})}\notag\\
\geq&  \sum_{i=0}^L \sum_{P\in \mathcal{P}_i^I} \left( \cost(\hat{\pi}_P^I) -\xi \cdot 2^{r+1}k\cdot 0.5\gamma T_i(o)\cdot  (\sqrt{d}\cdot 2g_i)^r\right) & \left(2^{r+4}k^2 \xi\leq \frac{\varepsilon}{10}\text{ and Equation~\eqref{eq:long_eq1_part2}}\right)\notag\\
\geq& \sum_{i=0}^L \sum_{P\in \mathcal{P}_i^I} \left( \cost(\hat{\pi}_P^I) -\frac{\varepsilon}{1000 (k+d^{1.5r})L} \cdot  0.5 T_i(o)\cdot  (\sqrt{d}g_i)^r\right) & \left(\xi\leq \frac{\varepsilon}{2000\cdot 2^{2r}k(k+d^{1.5r})L}\right)\notag\\
=& \sum_{i=0}^L \sum_{P\in \mathcal{P}_i^I} \left( \cost(\hat{\pi}_P^I) -\frac{\varepsilon}{1000 (k+d^{1.5r})L} \cdot   0.005o\right) & \left(T_i(o)=\frac{0.01o}{(\sqrt{d}g_i)^r}\right)\notag\\
= & \sum_{i=0}^L \sum_{P\in \mathcal{P}_i^I} \cost(\hat{\pi}_P^I) - \frac{\varepsilon}{1000 (k+d^{1.5r})L} \cdot  0.005o\cdot\sum_{i=0}^L|\mathcal{P}_i^I|\notag\\
\overset{(a)}{\geq} & \sum_{i=0}^L \sum_{P\in \mathcal{P}_i^I} \cost(\hat{\pi}_P^I) - \frac{\varepsilon}{10}\cdot o & \left(\sum_{i=0}^L |\mathcal{P}_i^I|\leq 20000(k+d^{1.5r})L\right)\notag\\
\overset{(b)}{\geq} & \sum_{i=0}^L \sum_{P\in \mathcal{P}_i^I} \cost(\hat{\pi}_P^I) - \frac{\varepsilon}{10}\cdot \cost_t(Q,Z), 
\end{align}
where step (a) follows from $\sum_{i=0}^L |\mathcal{P}_i^I|\leq \sum_{i=0}^L s_i\leq 20000(k+d^{1.5r})L$ which is according to Algorithm~\ref{alg:coreset_construction}, and  step (b) follows from $o\leq \OPT^{(r)}_{k\text{-clus}}\leq \cost_t(Q,Z)$.
Consider $i\in\{0,1,\cdots,L\}$ and $P\in\mathcal{P}^I$.
Event $\mathcal{E}(P)$ shows that 
\begin{align}\label{eq:long_eq1_part4}
(1+\xi)\cdot \cost(\hat{\pi}_P^I) \geq \cost(\pi_{P\cap Q'}^I) - \xi\cdot |P| (\sqrt{d}\cdot 2g_i)^r.
\end{align}
Since $\xi\leq \frac{\varepsilon}{4000\cdot 2^r(kL+d^{1.5r})L}$, we have:
\begin{align}\label{eq:long_eq1_part5}
&(1+\xi)\cdot \sum_{i=0}^L\sum_{P\in \mathcal{P}_i^I} \cost(\hat{\pi}^I_P) \notag\\
\geq & \sum_{i=0}^L\sum_{P\in \mathcal{P}_i^I} \left(\cost(\hat{\pi}_{P\cap Q'}^I)-\xi \cdot |P|(\sqrt{d}\cdot 2g_i)^r\right) & \text{(Equation~\eqref{eq:long_eq1_part4})}\notag\\
\overset{(a)}{\geq} & \sum_{i=0}^L\sum_{P\in \mathcal{P}_i^I} \cost(\hat{\pi}_{P\cap Q'}^I)-\sum_{i=0}^L \xi \cdot 2\cdot 10^4(kL+d^{1.5r})T_i(o)(\sqrt{d}\cdot 2g_i)^r \notag\\
\geq & \sum_{i=0}^L\sum_{P\in \mathcal{P}_i^I} \cost(\hat{\pi}_{P\cap Q'}^I)-\sum_{i=0}^L \frac{5\varepsilon}{L}\cdot  T_i(o)\cdot (\sqrt{d}g_i)^r &\left(\xi\leq \frac{\varepsilon}{4000\cdot 2^r(kL+d^{1.5r})L}\right)\notag\\
= & \sum_{i=0}^L\sum_{P\in \mathcal{P}_i^I} \cost(\hat{\pi}_{P\cap Q'}^I)- \sum_{i=0}^{L} \frac{5\varepsilon}{L}\cdot 0.01o & (T_i(o)=0.01o/(\sqrt{d}g_i)^r)\notag\\
\geq & \sum_{i=0}^L\sum_{P\in \mathcal{P}_i^I} \cost(\hat{\pi}_{P\cap Q'}^I)- \frac{\varepsilon}{10}\cdot o\notag \\
\overset{(b)}{\geq} & \sum_{i=0}^L\sum_{P\in \mathcal{P}_i^I} \cost(\hat{\pi}_{P\cap Q'}^I)- \frac{\varepsilon}{10}\cdot \cost_t(Q,Z),
\end{align}
where step (a) follows from $\sum_{P\in\mathcal{P}_i^I} |P|\leq \sum_{j=1}^{s_i}|Q_{i,j}|\leq 2\cdot 10^4(kL+d^{1.5r})T_i(o)$ and step (b) follows from that $o\leq \OPT^{(r)}_{k\text{-clus}}\leq \cost_t(Q,Z)$.
Consider the total weights of points assigned to each center by $\hat{\pi}_{P\cap Q'}^I$ for all $P\in\mathcal{P}_i^I$ and all $i\in\{0,1,\cdots,L\}$. 
We have:
\begin{align*}
&\left\|\sum_{i=0}^L\sum_{P\in\mathcal{P}_i^I} s(\hat{\pi}^I_{P\cap Q'})-s(\pi^I)\right\|_{\infty} \\
\leq & \left\|\sum_{i=0}^L\sum_{P\in\mathcal{P}_i^I} \left(s(\hat{\pi}^I_{P\cap Q'})-s(\hat{\pi}_P^I)\right) \right\|_{\infty} + \left\|\sum_{i=0}^L\sum_{P\in\mathcal{P}_i^I} \left(s(\hat{\pi}_P^I) - s(\pi^I_P)\right)\right\|_{\infty} & \text{(triangle inequality)}\\
\leq & \sum_{i=0}^L\sum_{P\in\mathcal{P}_i^I}\left\| s(\hat{\pi}^I_{P\cap Q'})-s(\hat{\pi}_P^I) \right\|_{\infty} + \sum_{i=0}^L\sum_{P\in\mathcal{P}_i^I}\left\| s(\hat{\pi}_P^I) - s(\pi^I_P)\right\|_{\infty} & \text{(triangle inequality)}\\
\leq & \sum_{i=0}^L\sum_{P\in\mathcal{P}_i^I}\left\| s(\hat{\pi}^I_{P\cap Q'})-s(\hat{\pi}_P^I) \right\|_{1} + \sum_{i=0}^L\sum_{P\in\mathcal{P}_i^I}\left\| s(\hat{\pi}_P^I) - s(\pi^I_P)\right\|_{1} & \text{($\forall x\in\mathbb{R}^k,\|x\|_{\infty}\leq \|x\|_1$)}\\
\leq  &\sum_{i=0}^L\sum_{P\in\mathcal{P}_i^I} \xi |P| + \sum_{i=0}^L\sum_{P\in\mathcal{P}_i^I}\left\| s(\hat{\pi}_P^I) - s(\pi^I_P)\right\|_{1} & \text{ (event $\mathcal{E}(P)$)}\\
\leq & \sum_{i=0}^L\sum_{P\in\mathcal{P}_i^I} \xi |P| + \sum_{i=0}^L\sum_{P\in\mathcal{P}_i^I}16k\xi |P| & \text{(Lemma~\ref{lem:transferred_also_good})}\\
\leq & \xi|Q^I|+16k\xi|Q^I|\leq \eta t & \left(\xi\leq \frac{\eta}{20k^2},Q^I\subseteq Q,t\geq \frac{|Q|}{k}\right).
\end{align*}
Since $\|s(\pi^I)\|_{\infty}\leq t$, we know that 
\begin{align*}
    \sum_{i=0}^L\sum_{P\in \mathcal{P}_i^I} \cost(\hat{\pi}_{P\cap Q'}^I) \geq \cost_{(1+\eta)t}(Q',Z,w').
\end{align*}
By combining above inequality with Equation~\eqref{eq:long_eq1_part3} and Equation~\eqref{eq:long_eq1_part5}, we have:
\begin{align*}
&(1+\varepsilon/10)\cdot \cost_t(Q,Z) \\
\geq& \sum_{i=0}^L \sum_{P\in \mathcal{P}_i^I} \cost(\hat{\pi}_P^I) - \frac{\varepsilon}{10}\cdot \cost_t(Q,Z) \\
\geq& \frac{1}{1+\xi} \left(\sum_{i=0}^L\sum_{P\in \mathcal{P}_i^I} \cost(\hat{\pi}_{P\cap Q'}^I)- \frac{\varepsilon}{10}\cdot \cost_t(Q,Z) \right)- \frac{\varepsilon}{10}\cdot \cost_t(Q,Z)\\
\geq & \frac{1}{1+\xi} \left(\cost_{(1+\eta)t}(Q',Z,w')- \frac{\varepsilon}{10}\cdot \cost_t(Q,Z) \right)- \frac{\varepsilon}{10}\cdot \cost_t(Q,Z).\\
\end{align*}
Since $\xi\leq \varepsilon/10$ and $\varepsilon\in(0,0.5)$, we can conclude that
\begin{align*}
\cost_{(1+\eta)t}(Q',Z,w')\leq (1+\varepsilon)\cdot \cost_t(Q,Z).
\end{align*}

Next, let us focus on proving $(1+\varepsilon)\cdot \cost_t(Q',Z,w')\geq \cost_{(1+\eta)t}(Q,Z)$.
We only need to consider the case when $\cost_t(Q',Z,w')\not=\infty$.
Since $Q'$ has at most $L+1$ different weights, according to Lemma~\ref{lem:opt_halfspace}, there are $m$ sets of assignment half-spaces $\mathcal{H}'^{(0)}=\left\{H'^{(0)}_{(j,j')}\mid j<j'\in[k]\right\},\mathcal{H}'^{(1)}=\left\{H'^{(1)}_{(j,j')}\mid j<j'\in[k]\right\},\cdots,\mathcal{H}'^{(L)}=\left\{H'^{(L)}_{(j,j')}\mid j<j'\in[k]\right\}$ corresponding to $Z$ such that $\forall i\in\{0,1,\cdots,L\}$, $\mathcal{H}'^{(i)}$ is valid for $Q'_i$, and 
\begin{align}
\cost_t(Q',Z,w')=\sum_{i=0}^L \cost(\pi_{Q'_i}),\label{eq:costprime_is_equal_to_costpiprime}
\end{align}
where $\pi_{Q'_i}:Q'_i\rightarrow Z$ is an assignment mapping corresponding to $\mathcal{H}'^{(i)}$.
For $i\in\{0,1,\cdots,L\},P\in\mathcal{P}^I_i$, let $\pi_{P\cap Q'_i}:P\cap Q'_i\rightarrow Z$ be the assignment mapping such that $\forall p\in P\cap Q'_i, \pi_{P\cap Q_i'}(p)=\pi_{Q'_i}(p)$, and let $\hat{\pi}_{P\cap Q'_i}:P\cap Q'_i\rightarrow Z$ be a transferred assignment mapping corresponding to $(\mathcal{H}'^{(i)},B^{P,Z,\mathcal{H}'^{(i)}},\xi,0.5\gamma T_i(o))$, where $B^{P,Z,\mathcal{H}'^{(i)}}=(b_0^{P,Z,\mathcal{H}'^{(i)}},b_1^{P,Z,\mathcal{H}'^{(i)}},\cdots,b_k^{P,Z,\mathcal{H}'^{(i)}})$ is the same as defined in the event $\mathcal{E}(P)$, i.e., $\forall j\in\{0,1,\cdots,k\},b_j^{P,Z,\mathcal{H}'^{(i)}}=\sum_{p\in P\cap Q'_i\cap R_j^{Z,\mathcal{H}'^{(i)}}} w'(p)$, where $(R_0^{Z,\mathcal{H}‘^{(i)}},R_1^{Z,\mathcal{H}‘^{(i)}},\cdots,R_k^{Z,\mathcal{H}‘^{(i)}})$ are regions (Definition~\ref{def:regions}) induced by $\mathcal{H}'^{(i)}$.
Let $\hat{\pi}_{P}:P\rightarrow Z$ also be a transferred assignment mapping corresponding to $(\mathcal{H}'^{(i)},B^{P,Z,\mathcal{H}'^{(i)}},\xi,0.5\gamma T_i(o))$.

We have:
\begin{align}\label{eq:long_eq2_part1}
\cost_t(Q',Z,w')
=  \sum_{i=0}^L \cost(\pi_{Q'_i})
=  \sum_{i=0}^L \sum_{P\in\mathcal{P}^I_i} \cost(\pi_{P\cap Q'_i}),
\end{align}
where the second step follows from Equation~\eqref{eq:costprime_is_equal_to_costpiprime}.
Notice that when we compute $\cost(\pi_{Q'_i})$ and $\cost(\pi_{P\cap Q'_i})$, each point $p$ has weight $w'(p)$.
For $i\in\{0,1,\cdots,L\}$ and $P\in\mathcal{P}^I$, consider the points sampled from $P$, i.e., $P\cap Q'_i$.
Event $\mathcal{E}(P)$ tells us that $\sum_{p\in P\cap Q'_i} w'(p)\geq 0.9\cdot 0.5\gamma T_i(o)$.
Since all points in $P$ are in the same cell in $G_{i-1}$, we know that $\forall p,q\in P$, $\dist(p,q)\leq \sqrt{d}\cdot 2g_i$ which implies that $\forall p,q\in P\cap Q'_i$, $\dist(p,q)\leq \sqrt{d}\cdot 2g_i$.
Then by Lemma~\ref{lem:transferred_also_good}, we have:
\begin{align}\label{eq:long_eq2_part2}
(1+2^{r+4}k^2\xi)\cdot \cost(\pi_{P\cap Q'_i})\geq \cost(\hat{\pi}_{P\cap Q'_i})-\xi\cdot 2^{r+1} k \cdot 0.5\gamma T_i(o)\cdot (\sqrt{d}\cdot 2g_i)^r.
\end{align}
Since $2^{r+4}k^2\xi\leq \varepsilon/10$ and $\xi\leq \frac{\varepsilon}{2000\cdot 2^{2r}k(k+d^{1.5r})L}$, we have:
\begin{align}\label{eq:long_eq2_part3}
&(1+\varepsilon/10)\cdot \cost_t(Q',Z,w')\notag\\
\geq & (1+\varepsilon/10)\cdot \sum_{i=0}^L\sum_{P\in\mathcal{P}_i^I}\cost(\pi_{P\cap Q'_i}) & \text{(Equation~\eqref{eq:long_eq2_part1})}\notag\\
\geq & \sum_{i=0}^L\sum_{P\in\mathcal{P}_i^I} \left(\cost(\hat{\pi}_{P\cap Q'_i})-\xi\cdot 2^{r+1} k \cdot 0.5\gamma T_i(o)\cdot (\sqrt{d}\cdot 2g_i)^r\right) & \left(2^{r+4}k^2\xi\leq \frac{\varepsilon}{10}\text{ and Equation~\eqref{eq:long_eq2_part2}}\right)\notag\\
\geq & \sum_{i=0}^L\sum_{P\in\mathcal{P}_i^I} \left(\cost(\hat{\pi}_{P\cap Q'_i}) - \frac{\varepsilon}{2000(k+d^{1.5r})L}\cdot T_i(o)\cdot (\sqrt{d}g_i)^r\right) & \left(\xi\leq \frac{\varepsilon}{2000\cdot 2^{2r}k(k+d^{1.5r})L}\right)\notag\\
= &\sum_{i=0}^L\sum_{P\in\mathcal{P}_i^I} \left(\cost(\hat{\pi}_{P\cap Q'_i}) - \frac{\varepsilon}{2000(k+d^{1.5r})L}\cdot 0.01o\right) & \left(T_i(o)=\frac{0.01o}{(\sqrt{d}g_i)^r}\right) \notag\\
= & \sum_{i=0}^L\sum_{P\in\mathcal{P}_i^I} \cost(\hat{\pi}_{P\cap Q'_i}) - \frac{\varepsilon}{2000(k+d^{1.5r})L}\cdot 0.01o \cdot \sum_{i=0}^L |\mathcal{P}_i^I| \notag\\
\overset{(a)}{\geq} &  \sum_{i=0}^L\sum_{P\in\mathcal{P}_i^I} \cost(\hat{\pi}_{P\cap Q'_i}) - \frac{\varepsilon}{10}\cdot o & \left(\sum_{i=0}^L |\mathcal{P}_i^I|\leq 20000(k+d^{1.5r})L\right)\notag\\
\overset{(b)}{\geq} & \sum_{i=0}^L\sum_{P\in\mathcal{P}_i^I} \cost(\hat{\pi}_{P\cap Q'_i}) - \frac{\varepsilon}{10}\cdot \cost_{(1+\eta)t}(Q,Z),
\end{align}
where step (a) follows from $\sum_{i=0}^L |\mathcal{P}_i^I|\leq \sum_{i=0}^L s_i\leq 20000(k+d^{1.5r})L$ which is according to Algorithm~\ref{alg:coreset_construction}, and step (b) follows from $o\leq \OPT^{(r)}_{k\text{-clus}}\leq \cost_{(1+\eta)t}(Q,Z)$.
Consider $i\in\{0,1,\cdots,L\}$ and $P\in\mathcal{P}^I$.
Event $\mathcal{E}(P)$ shows that
\begin{align}\label{eq:long_eq2_part4}
\cost(\hat{\pi}_{P\cap Q'_i})\geq (1-\xi)\cost(\hat{\pi}_P) - \xi |P|(\sqrt{d}\cdot 2g_i)^r.
\end{align}
Since $\xi\leq \frac{\varepsilon}{4000\cdot 2^r(kL+d^{1.5r})L}$,
\begin{align}\label{eq:long_eq2_part5}
&\sum_{i=0}^L\sum_{P\in\mathcal{P}_i^I} \cost(\hat{\pi}_{P\cap Q'_i})\notag\\
\geq & \sum_{i=0}^L\sum_{P\in\mathcal{P}_i^I} \left((1-\xi)\cdot \cost(\hat{\pi}_{P})-\xi |P|(\sqrt{d}\cdot 2g_i)^r\right)& \text{(Equation~\eqref{eq:long_eq2_part4})}\notag\\
\overset{(a)}{\geq} & (1-\xi) \sum_{i=0}^L\sum_{P\in\mathcal{P}_i^I} \cost(\hat{\pi}_{P}) - \sum_{i=0}^L \xi\cdot 20000(kL+d^{1.5r})T_i(o)(\sqrt{d}\cdot 2g_i)^r\notag\\
\geq & (1-\xi) \sum_{i=0}^L\sum_{P\in\mathcal{P}_i^I} \cost(\hat{\pi}_{P}) - \sum_{i=0}^L \frac{5\varepsilon}{L}\cdot T_i(o)\cdot (\sqrt{d}g_i)^r &\left(\xi\leq \frac{\varepsilon}{4000\cdot 2^r(kL+d^{1.5r})L}\right)\notag\\
= & (1-\xi) \sum_{i=0}^L\sum_{P\in\mathcal{P}_i^I} \cost(\hat{\pi}_{P}) - \sum_{i=0}^L \frac{5\varepsilon}{L}\cdot 0.01o & (T_i(o)=0.01o/(\sqrt{d}g_i)^r)\notag\\
\geq &(1-\xi) \sum_{i=0}^L\sum_{P\in\mathcal{P}_i^I} \cost(\hat{\pi}_{P}) - \frac{\varepsilon}{10}\cdot o\notag\\
\overset{(b)}{\geq}& (1-\xi)\sum_{i=0}^L \sum_{P\in\mathcal{P}_i^I}\cost(\hat{\pi}_{P}) -\frac{\varepsilon}{10}\cdot \cost_{(1+\eta)t}(Q,Z),
\end{align}
where step (a) follows from $\sum_{P\in\mathcal{P}_i^I} |P|\leq \sum_{j=1}^{s_i}|Q_{i,j}|\leq 2\cdot 10^4(kL+d^{1.5r})T_i(o)$ and step (b) follows from that $o\leq \OPT^{(r)}_{k\text{-clus}}\leq \cost_{(1+\eta)t}(Q,Z)$.
Consider the total number of points assigned to each center by $\hat{\pi}_P$ for all $P\in\mathcal{P}_i^I$ and all $i\in\{0,1,\cdots,L\}$.
We have:
\begin{align*}
&\left\|\sum_{i=0}^L\sum_{P\in\mathcal{P}_i^I} s(\hat{\pi}_P)-\sum_{i=0}^L s(\pi_{Q_i'})\right\|_{\infty}\\
= & \left\|\sum_{i=0}^L\sum_{P\in\mathcal{P}_i^I} \left(s(\hat{\pi}_P)-s(\pi_{P\cap Q'_i})\right)\right\|_{\infty}\\
\leq & \sum_{i=0}^L\sum_{P\in\mathcal{P}_i^I}\left\| s(\hat{\pi}_P)-s(\pi_{P\cap Q'_i})\right\|_{\infty} & \text{(triangle inequality)}\\
\leq & \sum_{i=0}^L\sum_{P\in\mathcal{P}_i^I}\left(\left\| s(\hat{\pi}_P)-s(\hat{\pi}_{P\cap Q'_i})\right\|_{\infty}+\left\| s(\hat{\pi}_{P\cap Q'_i})-s(\pi_{P\cap Q'_i})\right\|_{\infty}\right) & \text{(triangle inequality)}\\
\leq & \sum_{i=0}^L\sum_{P\in\mathcal{P}_i^I}\left(\left\| s(\hat{\pi}_P)-s(\hat{\pi}_{P\cap Q'_i})\right\|_{1}+\left\| s(\hat{\pi}_{P\cap Q'_i})-s(\pi_{P\cap Q'_i})\right\|_{1}\right) & \left(\forall x\in\mathbb{R}^k, \|x\|_\infty\leq \|x\|_1\right)\\
\leq & \sum_{i=0}^L \sum_{P\in\mathcal{P}_i^I} \xi |P| + \sum_{i=0}^L \sum_{P\in\mathcal{P}_i^I} \left\| s(\hat{\pi}_{P\cap Q'_i})-s(\pi_{P\cap Q'_i})\right\|_{1} & \text{(event $\mathcal{E}(P)$)}\\
\leq & \sum_{i=0}^L \sum_{P\in\mathcal{P}_i^I} \xi |P| + \sum_{i=0}^L \sum_{P\in\mathcal{P}_i^I} 16k\xi \left(\sum_{p\in P\cap Q'_i} w'(p)\right) & \text{(Lemma~\ref{lem:transferred_also_good})}\\
= & \sum_{i=0}^L \sum_{P\in\mathcal{P}_i^I} \xi |P| + \sum_{i=0}^L \sum_{P\in\mathcal{P}_i^I} 16k\xi \left\|s\left(\hat{\pi}_{P\cap Q'_i}\right)\right\|_1 & \left(\left\|s\left(\hat{\pi}_{P\cap Q'_i}\right)\right\|_1 = \sum_{p\in P\cap Q'_i} w'(p)\right)\\
\leq & \sum_{i=0}^L \sum_{P\in\mathcal{P}_i^I} \xi |P| + \sum_{i=0}^L \sum_{P\in\mathcal{P}_i^I} 16k\xi (\|s(\hat{\pi}_{P})\|_1+ \xi |P|) & \text{(event $\mathcal{E}(P)$)}\\
= & \sum_{i=0}^L \sum_{P\in\mathcal{P}_i^I} \xi |P| + \sum_{i=0}^L \sum_{P\in\mathcal{P}_i^I} 16k\xi (1+ \xi )|P| & (\|s(\hat{\pi}_P)\|_1=|P|)\\
\leq & \xi |Q^I| + 16k\xi(1+\xi) |Q^I|\leq \eta t/4 &\left(\xi\leq \frac{\eta}{160k^2},Q^I\subseteq Q,t\geq \frac{|Q|}{k}\right).
\end{align*}
Since $\|\sum_{i=0}^L s(\pi_{Q'_i})\|_{\infty}\leq t$, we know that
\begin{align}\label{eq:long_eq2_part6}
    \sum_{i=0}^L \sum_{P\in\mathcal{P}_i^I}\cost(\hat{\pi}_P)\geq \cost_{(1+\eta/4)t}(Q^I,Z).
\end{align}
Notice that $\forall i\in\{0,1,\cdots,L\},j\in[s_i]$ with $Q_{i,j}\not\in \mathcal{P}_i^I$, we know that $|Q_{i,j}|\leq 2\gamma T_i(o)$ due to Algorithm~\ref{alg:coreset_construction}.
Because $\gamma\leq \min\left(\frac{\eta}{80\cdot 2^rkL},\frac{\varepsilon}{40000\cdot 2^{2r}(k+d^{1.5r})L}\right)$ and we can apply Lemma~\ref{lem:small_parts_removal}, we have
\begin{align*}
\cost_{(1+\eta/4)t}(Q^I,Z)\geq \cost_{(1+\eta/4)^2t}(Q,Z)/(1+\varepsilon/10)\geq \cost_{(1+\eta)t}(Q,Z)/(1+\varepsilon/10),
\end{align*}
where the last step follows from $(1+\eta/4)^2\leq 1+\eta$.
By Equation~\eqref{eq:long_eq2_part6}, we have:
\begin{align*}
\sum_{i=0}^L\sum_{P\in \mathcal{P}_i^I}\cost(\hat{\pi}_P)\geq \cost_{(1+\eta)t}(Q,Z)/(1+\varepsilon/10).
\end{align*}
By Equation~\eqref{eq:long_eq2_part5}, we have:
\begin{align*}
\sum_{i=0}^L\sum_{P\in\mathcal{P}_i^I} \cost(\hat{\pi}_{P\cap Q'_i}) \geq (1-\xi) \cost_{(1+\eta)t}(Q,Z)/(1+\varepsilon/10) - \frac{\varepsilon}{10}\cdot \cost_{(1+\eta)t}(Q,Z).
\end{align*}
By Equation~\eqref{eq:long_eq2_part3}, we have:
\begin{align*}
(1+\varepsilon/10)\cdot \cost_t(Q',Z,w') \geq (1-\xi) \cost_{(1+\eta)t}(Q,Z)/(1+\varepsilon/10) - 2\cdot \frac{\varepsilon}{10}\cdot \cost_{(1+\eta)t}(Q,Z).
\end{align*}
Since $\xi\leq \varepsilon/10$ and $\varepsilon\leq 0.5$, we can reorder the terms in the above equation to conclude that
\begin{align*}
\cost_{(1+\eta)t}(Q,Z)\leq (1+\varepsilon)\cost_t(Q',Z,w').
\end{align*}

\end{proof}

Next, we consider the success probability and the size of the coreset.
As shown in Lemma~\ref{lem:number_of_center_cells}, $\mathcal{F}$ happens with probability at least $0.99$, i.e., with high probability, there should not be too many center cells.
Furthermore, as explained before previous lemma, we can suppose all the estimated sizes $\tau(C\cap Q),~\tau\left(\bigcup_j^{s_i}Q_{i,j}\right)$ and $\tau(Q_{i,j})$ are good estimations to $|C\cap Q|,~\left|\bigcup_{j=1}^{s_i} Q_{i,j}\right|$ and $|Q_{i,j}|$.
Condition on these events, Algorithm~\ref{alg:coreset_construction} does not output FAIL, and with high probability, the size of the outputted coreset is small.
\begin{lemma}[Success probability and the size]\label{lem:offline_success_prob}
Suppose the following conditions are satisfied:
\begin{enumerate}
    \item $\mathcal{F}$ happens (Lemma~\ref{lem:number_of_center_cells}),
    \item $\OPT^{(r)}_{k\text{-clus}}/10\leq o\leq \OPT^{(r)}_{k\text{-clus}}$,
     \item All estimated size $\tau\left(\bigcup_{j=1}^{s_i} Q_{i,j}\right)$ in Algorithm~\ref{alg:coreset_construction} and $\tau(C\cap Q)$ in Algorithm~\ref{alg:heavy_light} called by Algorithm~\ref{alg:coreset_construction} are good (Definition~\ref{def:good_estimation_alg2}, Definition~\ref{def:good_estimation_alg1}).
\end{enumerate}
Algorithm~\ref{alg:coreset_construction} does not return FAIL, and with probability at least $0.99$, 
\begin{align*}
|Q'|\leq \frac{8\cdot 10^{12}\cdot 2^{10(r+10)}rk^6d(k+d^{1.5r})^5L^{10}\log(kdL)}{\min(\varepsilon,\eta)^4}.
\end{align*}
\end{lemma}
\begin{proof}
By Lemma~\ref{lem:num_heavy_cells}, since event $\mathcal{F}$, $\sum_{i=0}^L s_i\leq 2000(k+d^{1.5r})L\cdot 10\leq 20000(k+d^{1.5r})L$ which implies that Algorithm~\ref{alg:coreset_construction} does not return FAIL in line~\ref{sta:return_fail_in_alg1}.

Recall that $Z^*\subset [\Delta]^d$ with $|Z^*|\leq k$ is the optimal solution of the standard $\ell_r$ $k$-clustering problem of $Q$, i.e., $\cost(Q,Z^*)=\OPT^{(r)}_{k\text{-clus}}$, and a cell $C\in G_i$ a center cell if $\dist(C,Z^*)\leq g_i/d$.
Due to $\mathcal{F}$, the total number of center cells is at most $2000kL$.
Let us consider $\sum_{j=1}^{s_i} |Q_{i,j}|$ for an arbitrary $i\in\{0,1,\cdots, L\}$.
We have:
\begin{align*}
\sum_{j=1}^{s_i} |Q_{i,j}|&=\sum_{C\in G_i:~C\text{ is crucial}}|C\cap Q|\\
&=\sum_{C\in G_i:~C\text{ is crucial, and is a center cell}}|C\cap Q| + \sum_{C\in G_i:~C\text{ is crucial, but is not a center cell}} |C\cap Q|\\
 &\leq 2000kL\cdot 1.1T_i(o) + \sum_{C\in G_i:~C\text{ is crucial, but is not a center cell}} |C\cap Q|\\
 &\leq 2000kL\cdot 1.1T_i(o) + \frac{\OPT^{(r)}_{k\text{-clus}}}{(g_i/d)^r}\\
 &\leq 2000kL\cdot 1.1T_i(o) + 100 d^{1.5r} T_i(o)\cdot \frac{\OPT^{(r)}_{k\text{-clus}}}{o}\\
 &\leq 2000kL\cdot 1.1T_i(o) + 1000 d^{1.5r} T_i(o) \\
  &\leq 4000 (kL+d^{1.5r})T_i(o),
\end{align*}
where the first step follows from the construction of $Q_{i,j}$, the third step follows from that the number of center cells is at most $2000kL$ and each crucial cell has at most $1.1T_i(o)$ points, the forth step follows from that each point $p$ in the non-center cell has distance to $Z^*$ at least $g_i/d$, the fifth step follows from the definition of $T_i(o)$, and the sixth step follows from $o\geq \frac{\OPT^{(r)}_{k\text{-clus}}}{10}$.
Thus, $\tau\left(\bigcup_{j=1}^{s_i} Q_{i,j}\right)\leq 10000(kL+d^{1.5r})T_i(o)$ which implies that Algorithm~\ref{alg:coreset_construction} does not return FAIL in line~\ref{sta:return_fail_in_alg2}.

Let us analyze the size of the coreset. 
We have
\begin{align*}
&\E[|Q'|] = \sum_{i=0}^L \E[|Q'_i|] \leq \sum_{i=0}^L \phi_i\sum_{j=1}^{s_i}|Q_{i,j}|\leq \sum_{i=0}^L \phi_i\cdot 4000(kL+d^{1.5r})T_i(o)\\
&\leq \frac{8\cdot 10^{10}\cdot 2^{10(r+10)}rk^6d(k+d^{1.5r})^5L^{10}\log(kdL)}{\min(\varepsilon,\eta)^4},
\end{align*}
where the last step follows from 
\begin{align*}
\phi_i\leq \frac{10^{7}\cdot 2^{10(r+10)}\cdot k^6d(k+d^{1.5r})^4L^8\log(kdL)}{\min(\varepsilon,\eta)^4T_i(o)}.
\end{align*}
By Markov's inequality, with probability at least $0.99$,
\begin{align*}
|Q'| \leq \frac{8\cdot 10^{12}\cdot 2^{10(r+10)}rk^6d(k+d^{1.5r})^5L^{10}\log(kdL)}{\min(\varepsilon,\eta)^4}.
\end{align*}
\end{proof}

The only thing remaining is to find a suitable parameter $o$ for Algorithm~\ref{alg:coreset_construction}.
Actually, we can enumerate $o$ exponentially and thus there must be some $o$ which is a good choice.
For the good choice of $o$, we can output the coreset with high probability.
Thus, we can conclude the following theorem.
\begin{theorem}[Offline algorithm]\label{thm:offline}
 Consider a point set $Q\subseteq [\Delta]^d$ which contains $n$ points and parameters $k\in\mathbb{Z}_{\geq 1},\varepsilon,\eta\in (0,0.5)$. 
 For constant $r\geq 1$, there is a randomized algorithm which outputs a subset of points $Q'\subseteq Q$ and weights $w':Q'\rightarrow \mathbb{R}_{>0}$ in time $O(nd\log^2(nd\Delta))$ such that with probability at least $0.9$,
\begin{enumerate}
\item $\forall t\geq n/k,Z\subset [\Delta]^d$ with $|Z|= k$, 
\begin{align*}
\begin{array}{ccc}
\cost_{(1+\eta)t}(Q,Z) \leq (1+\varepsilon)\cost_t(Q',Z,w')    &  \text{and} & \cost_{(1+\eta)t}(Q',Z,w')\leq (1+\varepsilon)\cost_t(Q,Z),
\end{array}
\end{align*}
\item $|Q'|\leq \poly(\varepsilon^{-1}\eta^{-1}kd\log \Delta)$,
\end{enumerate}
\end{theorem}
\begin{proof}
Algorithm~\ref{alg:coreset_construction} needs a parameter $o$ which is an approximation of $\OPT^{(r)}_{k\text{-clus}}$.
We can enumerate all possible $o\in\left\{1,2,4,\cdots,n\cdot\left(\sqrt{d}\Delta\right)^r\right\}$.
We choose the smallest $o$ such that Algorithm~\ref{alg:coreset_construction} does not output FAIL. 

Let us consider the running time of Algorithm~\ref{alg:coreset_construction}.
In line~\ref{sta:call_alg1} of Algorithm~\ref{alg:coreset_construction}, we call Algorithm~\ref{alg:heavy_light}.
In Algorithm~\ref{alg:heavy_light}, for each $p\in Q$, we can update the number of points in $c_i(p)$ for $i\in\{-1,0,1,\cdots,L\}$.
Then, for each $p\in Q$, we can check whether $c_i(p)$ is heavy or not for $i\in\{-1,0,1,\cdots,L\}$.
Thus, the total running time of Algorithm~\ref{alg:heavy_light} is $O(ndL)$.
In Algorithm~\ref{alg:coreset_construction}, for each point $p\in Q$ we should find the level $i\in\{0,1,\cdots,L\}$ such that $c_i(p)$ is crucial which takes $O(dL)$ time.
To conclude, the total running time of Algorithm~\ref{alg:coreset_construction} is $O(ndL)$.
Thus, the overall running time is $O(ndL)\cdot \log(n\cdot (\sqrt{d}\Delta)^r)=O(nd\log^2(nd\Delta))$.

Let us consider the correctness. 
According to Lemma~\ref{lem:number_of_center_cells}, with probability at least $0.99$, $\mathcal{F}$ happens.
According to Lemma~\ref{lem:offline_success_prob}, if $\OPT^{(r)}_{k\text{-clus}}/10\leq o\leq \OPT^{(r)}_{k\text{-clus}}$, Algorithm~\ref{alg:coreset_construction} does not output FAIL.
Thus we can find an $o\leq \OPT^{(r)}_{k\text{-clus}}$ such that Algorithm~\ref{alg:coreset_construction} does not output FAIL with probability at least $0.99$.
If Algorithm~\ref{alg:coreset_construction} does not output FAIL and $o\leq \OPT^{(r)}_{k\text{-clus}}$,
then, by Lemma~\ref{lem:correctness_offline}, with probability at least $0.99$, $\forall t\geq n/k,Z\subset [\Delta]^d$ with $|Z|= k$, 
\begin{align*}
\begin{array}{ccc}
\cost_{(1+\eta)t}(Q,Z) \leq (1+\varepsilon)\cost_t(Q',Z,w')    &  \text{and} & \cost_{(1+\eta)t}(Q',Z,w')\leq (1+\varepsilon)\cost_t(Q,Z),
\end{array}
\end{align*}
and furthermore, according to Lemma~\ref{lem:offline_success_prob}, with probability at least $0.99$, $|Q'|\leq \poly(\varepsilon^{-1}\eta^{-1}kdL)$.
\end{proof}

\subsection{Assignment Construction via Coreset}\label{sec:assignment_via_coreset}
In classic $k$-clustering problem, once centers are determined, each point should be assigned to the closest center.
But in capacitated $k$-clustering problem, even if the centers are determined, it is non trivial to assign points to centers.
In this section, we will discuss how to construct a good assignment for the input point set $Q$ given $k$ centers $Z=\{z_1,z_2,\cdots,z_k\}$ and the coreset $(Q',w')$ obtained by our construction.

Firstly, given a capacity $t'\geq \frac{1}{k}\cdot \max(\sum_{q\in Q'} w'(q),|Q|)$, we want to find an assignment $\pi':Q'\rightarrow Z$ such that $\cost(\pi')\leq (1+\varepsilon)\cost_{t'}(Q',Z,w')$ and $\|s(\pi')\|_{\infty}\leq (1+\eta)t'$. 
  Given centers $Z$, finding an assignment satisfying the capacity constraint for weighted points in general is NP-hard since we can reduce bin packing problem to the such feasibility problem.
 If we relax the problem to the fractional version, i.e., the weight of a point can be split to multiple centers, then the optimal assignment for the relaxed problem can be solved by the minimum-cost flow~\cite{bblm14}. 
 Given a fractional assignment, we can use the following way to reduce the number of points of which weight is split to multiple centers:
 \begin{enumerate}
     \item Build a bipartite graph as the following: create a vertex for each point and each center, and add an edge between a point vertex and a center vertex if there is a non-zero fraction of the weight of the point assigned to the center.
     \item Find an arbitrary simple cycle in the bipartite graph. 
     If there is no cycle, finish the procedure.
     Suppose the cycle corresponds to points $p_1,p_2,\cdots,p_m$ and centers $z_1,z_2,\cdots,z_m$ where $p_i$ connects to both $z_i$ and $z_{(i\bmod m) + 1}$. 
     Notice that $\sum_{i=1}^m \dist^r(p_i,z_i)=\sum_{i=1}^m \dist^r(p_i,z_{(i\bmod m)+1})$ since the given fractional assignment is optimal.
     \item Suppose $a$ is the minimum value of the weight assigned from $p_i$ to $z_i$ for all $i\in[m]$.
     For each $p_i$, move $a$ weights from $z_i$ to $z_{(i\bmod m)+1}$.
     \item Repeat above steps for the new fractional assignment.
 \end{enumerate}
 In each iteration of the above procedure, we can remove an edge between a point and a center.
 Thus, it only takes polynomial running time.
 Since at the end of the above procedure there is no cycle in the constructed bipartite graph, the number of points of which weight is split to multiple centers is at most $k-1$.
 For each of the $k-1$ points, we modify its assignment to make all of its weight assigned to the closest center.
 Thus, we can obtain an integral assignment $\pi':Q'\rightarrow Z$.
 Furthermore, we know that $\|s(\pi')\|_{\infty}\leq t'+(k-1)\cdot \max_{p\in Q'} w'(p)$.
 For $p\in Q'$, by algorithm~\ref{alg:coreset_construction}, if $p\in Q_{i,j}$, then $|Q_{i,j}|\geq 0.5\gamma T_i(o)$ and $w'(p)\leq \xi^3\gamma T_i(o)$ which implies that $w'(p)\leq \eta|Q|/k^2$ (due to the choice of $\xi$).
 Therefore, we can conlude that 
 \begin{align*}
 \|s(\pi')\|_{\infty}\leq (1+\eta)t',
 \end{align*}
 and 
 \begin{align*}
 \cost(\pi')\leq \cost_{t'}(Q',Z,w').
 \end{align*}
 
 Notice that $\pi'$ may not be represented by a small number of sets of assignment half-spaces.
 Thus, we need to apply the switching argument similar to the proof of Lemma~\ref{lem:opt_halfspace} to modify $\pi'$.
 The modification of $\pi'$ can be done by the following procedure:
 \begin{enumerate}
\item For each $Q'_i$ (see Algorithm~\ref{alg:coreset_construction}) do the following:
\begin{enumerate}
    \item Let $\pi'_{Q'_i}:Q'_i\rightarrow Z$ be the assignment mapping satisfying $\forall p\in Q'_i,\pi'_{Q'_i}(p)=\pi'(p)$.
    \item Since points in $Q'_i$ have the same weight, we can use minimum-cost flow to find an assignment mapping $\tilde{\pi}_{Q'_i}:Q'_i\rightarrow Z$ such that $s(\tilde{\pi}_{Q'_i})=s(\pi'_{Q'_i})$ and $\cost(\tilde{\pi}_{Q'_i})$ is minimized.
    \item If exists $p,q\in Q'_i,$ such that $\tilde{\pi}_{Q'_i}(p)=z_j$, $\tilde{\pi}_{Q'_i}(q)=z_{j'}$ $(j<j')$, $\dist^r(q,z_{j})-\dist^r(q,z_{j'})=\dist^r(p,z_{j})-\dist^r(p,z_{j'})$ and the alphabetic order of $q$ is smaller than $p$ (due to the optimality of $\tilde{\pi}_{Q'_i}$, $\dist^r(q,z_{j})-\dist^r(q,z_{j'})<\dist^r(p,z_{j})-\dist^r(p,z_{j'})$ can never happen),  
    then switch the assigned center of $p$ and the assigned center of $q$, i.e., $\tilde{\pi}_{Q'_i}(q)\gets z_j,\tilde{\pi}_{Q'_i}(p)\gets z_{j'}$.\label{it:switching}
    \item Repeat the above step until no switching happens.
    Let $\pi''_{Q'_i}:Q'_i\rightarrow Z$ be the final $\tilde{\pi}_{Q'_i}$ after all switching.
\end{enumerate}
\item Let $\pi'':Q'\rightarrow Z$ satisfy $\forall i\in\{0,1,\cdots,L\},p\in Q_i$, $\pi''(p)=\pi''_{Q'_i}(p)$.
 \end{enumerate}
 Consider step~\ref{it:switching}. 
 After each switching, $\forall l\in [s(\tilde{\pi}_{Q'_i})_j]$, the alphabetic order of the point assigned to $z_j$ with the $l$-th smallest alphabetic order can not increase.
 Thus, the total running time of the above procedure can be done in polynomial time.
Consider the properties of $\pi''$.
It is easy to see that $\cost(\pi'')=\sum_{i=0}^L\cost(\tilde{\pi}_{Q'_i})\leq \cost(\pi')$ and $s(\pi'')=\sum_{i=0}^L s(\tilde{\pi}_{Q'_i})=s(\pi')$.
Furthermore, by Defition~\ref{def:assignment_halfspace}, for each $i\in\{0,1,\cdots,L\}$, we can compute a set of assignment half-spaces $\mathcal{H}^i=\{H^i_{(j,j')}\mid j<j'\}$ such that $\forall p\in Q'_i$, $\pi''(p)=z_j$ if and only if $\forall j'\not=j, p\in H^i_{(j,j')}$.

It is enough to construct an assignment mapping $\pi:Q\rightarrow Z$ for the original point set.
For each $i\in \{0,1,\cdots,L\}$, for each $P\in\mathcal{P}_i^I$ (see Algorithm~\ref{alg:coreset_construction} for $\mathcal{P}_i^I$), we can construct a transferred assignment mapping $\pi_{P}:P\rightarrow Z$ corresponding to $(\mathcal{H}^i,B^{P,i},\xi,0.5\gamma T_i(o))$, where $B^{P,i}=(b_0^{P,i},b_1^{P,i},\cdots,b_k^{P,i})$ and $b_j^{P,i}=\sum_{p\in P\cap Q'_i:\forall j'\not=j, p\in H^i_{(j,j')}} w'(p)$.
According to the proof of Theorem~\ref{thm:offline}, we can condition on that $o\leq \OPT^{(r)}_{k\text{-clus}}$.
Similar to the proof of Lemma~\ref{lem:correctness_offline}, condition on $\mathcal{E}(P)$ for all $P$, we can show that
\begin{align*}
\sum_{i=0}^L\sum_{P\in\mathcal{P}_i^I}\sum_{p\in P} \dist^r(p,\pi_P(p))\leq (1+O(\varepsilon)) \sum_{p\in Q'} w'(p)\cdot \dist^r(p,\pi''(p)),
\end{align*} 
and
\begin{align*}
\left\|\sum_{i=0}^L\sum_{p\in\mathcal{P}_i^I}s(\pi_P)\right\|_{\infty}\leq (1+O(\eta))\cdot \|s(\pi'')\|_{\infty}.
\end{align*}
Then we construct $\pi:Q\rightarrow Z$ as the following:
\begin{enumerate}
\item if $\exists i\in\{0,1\cdots,L\},P\in \mathcal{P}_i^I$ such that $p\in P$, let $\pi(p)\gets \pi_P(p)$;
\item otherwise, let $\pi(p)\gets \arg\min_{z\in Z}\dist(p,z)$.
\end{enumerate}
According to the proof of Lemma~\ref{lem:small_parts_removal} (see Appendix), we can show that
 \begin{align*}
 |\{p\in Q\mid \forall i\in\{0,1,\cdots,L\},P\in \mathcal{P}_i^I,p\not\in P\}|\leq O(\eta)\cdot |Q|/k 
 \end{align*}
 and
 \begin{align*}
\sum_{p\in Q:\forall i\in\{0,1,\cdots,L\},P\in\mathcal{P}_i^I,p\not\in P} \dist^r(p,Z) \leq O(\varepsilon)\cdot  \cost\left(\bigcup_{i=0}^L\bigcup_{P\in \mathcal{P}_i^I} P,Z\right).
 \end{align*}
 Thus, we can conclude that 
 \begin{align*}
 \sum_{p\in Q} \dist^r(p,\pi(p)) \leq (1+O(\varepsilon)) \cdot \sum_{p\in Q'} w'(p)\cdot \dist^r(p,\pi''(p))\leq (1+O(\varepsilon)) \cost_{t'}(Q',Z,w'),
 \end{align*}
 and
 \begin{align*}
 \|s(\pi)\|_{\infty}\leq (1+O(\eta))\cdot t'.
 \end{align*}
 
 Notice that $\mathcal{P}_i^I$ can be determined by the heavy cells outputted by Algorithm~\ref{alg:heavy_light} and the estimated number of points in its children cells.
 By the above argument, if we store this information together with the coreset $(Q',w')$, we can determine the desired assignment mapping $\pi$ for any capacity $t'$ and centers $Z$ in $\poly(|Q'|)$ time. 
\section{Coreset in Streaming and Distributed Model}\label{sec:streaming}
\subsection{Estimating Number of Points via Sampling}
In this section, we discuss how to obtain the estimation $\tau(C\cap Q)$ in line~\ref{sta:estimated_size} of Algorithm~\ref{alg:heavy_light}, the estimation $\tau\left(\bigcup_{j=1}^{s_i} Q_{i,j}\right)$ in line~\ref{sta:return_fail_in_alg2} and the estimation $\tau(Q_{i,j})$ in line~\ref{sta:lb_of_each_final_part} of Algorithm~\ref{alg:coreset_construction}.
For convenience, we suppose that no two points in the input point set share the same coordinate\footnote{If two points share the same coordinate, we can assume that each point has a unique tag so we can distinguish them.}. 
We can obtain the estimation by the procedure in Algorithm~\ref{alg:estimating_points}.

\begin{algorithm}
\small
\caption{Estimation of Number of Points via Sampling}\label{alg:estimating_points}
\begin{enumerate}
\item Let $\lambda'\gets 100dL$.
\item 
Let $h_0,h_1,\cdots,h_L:[\Delta]\rightarrow \{0,1\}$ be $\lambda'$-wise independent hash functions, where $\forall i\in \{0,1,\cdots,L\},p\in[\Delta]^d$, it satisfies $\Pr[h_i(p)=1]=\psi_i=\min\left(\frac{10^6\lambda'}{T_i(o)},1\right)$.
\item For $i\in\{0,1,2,\cdots,L\}$ and each cell $C\in G_i$, set $\tau(C\cap Q)\gets \frac{1}{\psi_i}\cdot \sum_{p\in C\cap Q} h_i(p)$. 
Use these estimated values in Algorithm~\ref{alg:heavy_light}.
\item Let $h'_0,h'_1,\cdots,h'_L:[\Delta]\rightarrow \{0,1\}$ be $\lambda'$-wise independent hash functions, where $\forall i\in\{0,1,\cdots,L\},p\in[\Delta]^d$, it satisfies $\Pr[h'_i(p)=1] = \psi'_i = \min\left(\frac{10^6\lambda'}{\gamma T_i(o)},1\right)$.
\item In Algorithm~\ref{alg:coreset_construction}, for $i\in\{0,1,\cdots,L\}$ set 
\begin{align*}
\tau\left(\bigcup_{j=1}^{s_i} Q_{i,j}\right) \gets \frac{1}{\psi'_i} \sum_{C\in G_i:C\text{ is crucial}} \sum_{p\in C\cap Q} h'_i(p),
\end{align*}
and for $j\in [s_i]$ set 
\begin{align*}
\tau(Q_{i,j})\gets \frac{1}{\psi'_i} \sum_{C\in G_i:C\text{ is a crucial child of the }j\text{-th heavy cell in }G_{i-1}}\sum_{p\in C\cap Q} h'_i(p). 
\end{align*}
\end{enumerate}
\end{algorithm}

\begin{lemma}\label{lem:obtain_tau}
With probability at least $0.99$, Algorithm~\ref{alg:estimating_points} satisfies following conditions:
\begin{enumerate}
\item $\forall i\in\{0,1,\cdots,L\},C\in G_i,$ $\tau(C\cap Q)$ is good (Definition~\ref{def:good_estimation_alg1}). 
\item $\forall i\in\{0,1,\cdots,L\},$ $\tau\left(\bigcup_{j=1}^{s_i} Q_{i,j}\right)$ is good (Definition~\ref{def:good_estimation_alg2}).
\item $\forall i\in\{0,1,\cdots,L\},j\in[s_i]$, $\tau(Q_{i,j})$ is good (Definition~\ref{def:good_estimation_alg2}).
\end{enumerate}
\end{lemma}
\begin{proof}
Consider an arbitrary $i\in\{0,1,\cdots,L\}$ and an arbitrary cell $C\in G_i$.
Notice that $\E[\tau(C\cap Q)]=\frac{1}{\psi_i}\sum_{p\in C\cap Q}\E[h_i(p)]=|C\cap Q|$.
If $|C\cap Q|\geq T_i(o)$, by Lemma~\ref{lem:limited_independent_bound}, we have
\begin{align*}
\Pr[|\tau(C\cap Q)-|C\cap Q||\geq 0.01 |C\cap Q|] \leq 8\left(\frac{|C\cap Q|\lambda'\cdot\frac{1}{\psi_i}+\lambda'^2\cdot\frac{1}{\psi_i^2}}{0.01^2|C\cap Q|^2}\right)^{\lambda'/2} \leq 0.001\cdot\frac{1}{\Delta^d}.
\end{align*}
Similarly, if $|C\cap Q|<T_i(o)$, then by Lemma~\ref{lem:limited_independent_bound}, 
\begin{align*}
\Pr[|\tau(C\cap Q)-|C\cap Q||\geq 0.1 T_i(o)]\leq 8\left(\frac{|C\cap Q|\lambda'\cdot \frac{1}{\psi_i}+\lambda'^2\cdot\frac{1}{\psi_i^2}}{0.1^2|T_i(o)|^2}\right)^{\lambda'/2}\leq 0.001\cdot \frac{1}{\Delta^d}.
\end{align*}
Since there are at most $(\Delta/2^i)^d$ non-empty cells in $G_i$, the total number of non-empty cells is at most $2\Delta^d$.
By taking union bound over all such cells, with probability at least $0.998$, $\forall i\in\{0,1,\cdots,L\},C\in G_i,$ either $\tau(C\cap Q)\in |C\cap Q|\pm 0.1 T_i(o)$ or $\tau(C \cap Q)\in (1\pm 0.01)\cdot |C\cap Q|$.

By using Lemma~\ref{lem:limited_independent_bound} similar to the above argument, with probability at least $0.998$, $\forall i\in\{0,1,\cdots,L\},$ either $\tau\left(\bigcup_{j=1}^{s_i} Q_{i,j}\right)\in \sum_{j=1}^{s_i}|Q_{i,j}| \pm 0.1 T_i(o)$ or $\tau\left(\bigcup_{j=1}^{s_i} Q_{i,j}\right)\in (1\pm 0.1)\sum_{j=1}^{s_i}|Q_{i,j}|$.
We also have that, with probability $0.998$, $\forall i\in\{0,1,\cdots,L\},j\in[s_i]$, either $\tau(Q_{i,j})\in |Q_{i,j}|\pm 0.1\gamma T_i(o)$ or $\tau(Q_{i,j})\in (1\pm 0.1)|Q_{i,j}|$.
By taking union bound over all failure events, we complete the proof.
\end{proof}


\subsection{Coreset Construction over Dynamic Stream}
In this section, we will discuss how to implement the coreset construction in the streaming model. 
We consider the streaming model which allows both insertion and deletion. 
The description of the model is in the following.
\paragraph{Dynamic streaming model.} Initially, $Q$ is an empty point set. 
There is a stream of insertions and deletions, $(p_1,\pm)$, $(p_2,\pm),\cdots$, where $(p_i,+)$ denotes inserting a point $p_i\in [\Delta]^d$ into $Q$, and $(p_i,-)$ denotes deleting $p_i$ from $Q$. 
Each deletion $(p_i,-)$ guarantees that $p_i$ is in $Q$ before deletion.
A dynamic streaming algorithm is allowed a single pass over the stream.
At the end of the stream, the algorithm stores some information regarding $Q$. 
The space complexity of the algorithm is the total number of bits used by the algorithm during the stream.

In this section, we will introduce a dynamic streaming algorithm which can output a coreset for $\ell_r$ balanced $k$-clustering using space $\poly(\varepsilon^{-1}kdL)$ bits for constant $r$.

\begin{lemma}[Lemma 19 in~\cite{hsyz18}]\label{lem:useful_tool}
For $i\in\{0,1,\cdots,L\}$, $\alpha,\beta\in \mathbb{Z}_{\geq 1},\delta\in(0,0.5)$, there is a dynamic streaming algorithm \textsc{Storing}$(G_i,\alpha,\beta,\delta)$  which uses $O(\alpha\beta dL\cdot \log^2(\alpha\beta/\delta))$ bits to process a stream of insertion and deletion of points such that 
\begin{enumerate}
    \item 
    if the algorithm does not output FAIL, the algorithm will return 
    \begin{itemize}
    \item a set $\mathcal{C}$ of all non-empty cells,
    \item the number of points $f(C)$ in each cell $C\in\mathcal{C}$, 
    \item the set $S$ of points in all the non-empty cells that contain at most $\beta$ points,
    \end{itemize}
    \item if $|\mathcal{C}|\leq \alpha$, then with probability at least $1-\delta$, the algorithm does not output FAIL.
\end{enumerate}
\end{lemma}

Next, let us describe how to use above subroutine to implement our coreset construction algorithm. 
The idea is that we only store some information described by small number of bits, and at the end of the stream, we can use this information to implement Algorithm~\ref{alg:heavy_light}, Algorithm~\ref{alg:coreset_construction} and Algorithm~\ref{alg:estimating_points}.
The description is in Algorithm~\ref{alg:streaming}.

\begin{algorithm}
\small
\caption{Coreset Construction over a Dynamic Stream} \label{alg:streaming}
\begin{enumerate}
\item Let $\lambda\gets 10^6rk^3dL\lceil\log(kdL)\rceil$.
\item Let $h_0,h_1,\cdots,h_L,h'_0,h'_1,\cdots,h'_L,\hat{h}_0,\hat{h}_1,\cdots,\hat{h}_L:[\Delta]^d\rightarrow \{0,1\}$ be $\lambda$-wise independent hash functions:
\begin{align*}
&\forall i\in\{0,1,\cdots,L\}, p\in[\Delta]^d,\\
&\Pr[h_i(p)=1] = \psi_i & (\psi_i\text{ defined in Algorithm~\ref{alg:estimating_points}}),\\
&\Pr[h'_i(p)=1] = \psi'_i & (\psi_i'\text{ defined in Algorithm~\ref{alg:estimating_points}}),\\
&\Pr[\hat{h}_i(p) = 1] = \phi_i & (\phi_i\text{ defined in Algorithm~\ref{alg:heavy_light}}).
\end{align*}
\item For the input stream $(p_1,\pm),(p_2,\pm),\cdots$, create $3(L+1)$ sub-streams. For each $i\in\{0,1,\cdots,L\}$:
\begin{itemize}
\item There is a sub-stream which contains all the insertions/deletions related to all the points $p_j$, where $h_i(p_j)=1 \Biggm/ h'_i(p_j)=1 \Biggm/ \hat{h}_i(p_j)=1$.
Run a subroutine $\textsc{Storing}\left(G_i,\alpha_i,\beta_i,0.001\cdot\frac{1}{3(L+1)}\right)\biggm/ \textsc{Storing}\left(G_i,\alpha'_i,\beta'_i,0.001\cdot\frac{1}{3(L+1)}\right)\biggm/\textsc{Storing}\left(G_i,\hat{\alpha}_i,\hat{\beta}_i,0.001\cdot\frac{1}{3(L+1)}\right)$ (Lemma~\ref{lem:useful_tool}) on such sub-stream in parallel.
Let $\mathcal{C}_i\subset G_i,f:\mathcal{C}_i\rightarrow \mathbb{Z}_{\geq 0},S_i\subseteq [\Delta]^d \biggm/ \mathcal{C'}_i\subset G_i,f':\mathcal{C'}_i\rightarrow \mathbb{Z}_{\geq 0},S'_i\subseteq [\Delta]^d \biggm/ \hat{\mathcal{C}}_i\subset G_i,\hat{f}:\hat{\mathcal{C}}_i\rightarrow \mathbb{Z}_{\geq 0},\hat{S}_i\subseteq [\Delta]^d$ be the output of the subroutine, where $\alpha_i=10^6(k+d^{1.5r}\psi_i T_i(o))L^2, \alpha'_i=10^6(k+d^{1.5r}\psi'_i T_i(o))L^2,\hat{\alpha}_i=10^6(k+d^{1.5r}\phi_i T_i(o))L^2$ and $\beta_i=\beta'_i=1,\hat{\beta}_i=4\cdot  10^6 (k+d^{1.5r})L^2\phi_i T_i(o)$. \label{it:running_storing}
\end{itemize}
\item For each $i\in\{0,1,\cdots, L\}$, and each cell $C\in G_i$, if $C\in\mathcal{C}_i$, set $\tau(C\cap Q)\gets \frac{1}{\psi_i}\cdot f_i(C)$; set $\tau(C\cap Q)\gets 0$ otherwise.
Run Algorithm~\ref{alg:heavy_light} based on such $\tau(C\cap Q)$ for all cells $C$. \label{it:step3}
\item For each $i\in\{0,1,\cdots,L\}$, set
$
\tau\left(\bigcup_{j=1}^{s_i} Q_{i,j}\right) \gets \frac{1}{\psi'_i}\sum_{C\in\mathcal{C}'_i:C\text{ is crucial}} f'_i(C),
$
and $\forall j\in[s_i]$ set
$
\tau\left(Q_{i,j}\right)\gets \frac{1}{\psi'_i} \sum_{C\in \mathcal{C}'_i:C\text{ is a crucial child of the $j$-th heavy cell in }G_{i-1}} f'_i(C).
$ \label{it:step4}
\item Run Algorithm~\ref{alg:coreset_construction} based on all $\tau\left(\bigcup_{j=1}^{s_i} Q_{i,j}\right)$ and $\tau\left(Q_{i,j}\right)$. 
To obtain $Q'_i$ in line~\ref{sta:sampling_step} of Algorithm~\ref{alg:coreset_construction}, do the following:
$
Q'_i\gets \bigcup_{p\in \hat{S}_i : c_i(p)\text{ is crucial}, \tau(Q_{i,j})\geq \gamma T_i(o)} \{p\},
$
$\text{ where }c_{i-1}(p)\text{ is the }j\text{-th heavy cell in }G_{i-1}$.
Outputs $Q'$ and $w'$ returned by Algorithm~\ref{alg:coreset_construction}.
\label{it:step5}
\end{enumerate}
\end{algorithm}

\begin{lemma}[Correctness of the streaming algorithm]\label{lem:correctness_streaming}
If $o\leq \OPT^{(r)}_{k\text{-clus}}$ and none of the subroutines of Algorithm~\ref{alg:streaming} outputs FAIL, with probability at least $0.97$, the output $Q',w'$ of Algorithm~\ref{alg:streaming} satisfies that $\forall t\geq |Q|/k,Z\subset[\Delta]^d$ with $|Z|=k$,
\begin{align*}
\begin{array}{ccc}
\cost_{(1+\eta)t}(Q,Z) \leq (1+\varepsilon)\cost_t(Q',Z,w')    &  \text{and} & \cost_{(1+\eta)t}(Q',Z,w')\leq (1+\varepsilon)\cost_t(Q,Z). 
\end{array}
\end{align*}
\end{lemma}

\begin{proof}
Due to Lemma~\ref{lem:useful_tool}, since none of subroutine \textsc{Storing} outputs FAIL, $\forall i\in\{0,1,\cdots, L\}, C\in G_i$, 
\begin{itemize}
\item $f_i(C)=\sum_{p\in C\cap Q} h_i(p)$ if $C\in \mathcal{C}_i$ and $\sum_{p\in C\cap Q}h_i(p)=0$ otherwise,
\item $f'_i(C)=\sum_{p\in C\cap Q} h'_i(p)$ if $C\in \mathcal{C}'_i$ and $\sum_{p\in C\cap Q} h'_i(p)=0$ otherwise.
\end{itemize}
Thus, in step~\ref{it:step3} of Algorithm~\ref{alg:streaming}, 
\begin{align*}
\tau(C\cap Q) = \frac{1}{\psi_i} \sum_{p\in C\cap Q} h_i(p).
\end{align*}
Similarly, in step~\ref{it:step4} of Algorithm~\ref{alg:streaming},
\begin{align*}
&\tau\left(\bigcup_{j=1}^{s_i} Q_{i,j}\right) =\frac{1}{\psi'_i} \sum_{C\in G_i:C\text{ is crucial}} \sum_{p\in C\cap Q} h'_i(p),\\
&\tau\left(Q_{i,j}\right) = \frac{1}{\psi'_i} \sum_{C\in G_i: C\text{ is a rucial child of the }j\text{-th heavy cell in }G_{i-1}}\sum_{p\in C\cap Q} h'_i(p).
\end{align*}
By Lemma~\ref{lem:obtain_tau}, with probability at least $0.99$,
\begin{enumerate}
\item $\forall i\in\{0,1,\cdots,L\},C\in G_i,$ either $\tau(C\cap Q)\in |C\cap Q|\pm 0.1 T_i(o)$ or $\tau(C \cap Q)\in (1\pm 0.01)\cdot |C\cap Q|$,
\item $\forall i\in\{0,1,\cdots,L\},$ either $\tau\left(\bigcup_{j=1}^{s_i} Q_{i,j}\right)\in \sum_{j=1}^{s_i}|Q_{i,j}| \pm 0.1 T_i(o)$ or $\tau\left(\bigcup_{j=1}^{s_i} Q_{i,j}\right)\in (1\pm 0.1)\sum_{j=1}^{s_i}|Q_{i,j}|$,
\item $\forall i\in\{0,1,\cdots,L\},j\in[s_i]$, either $\tau(Q_{i,j})\in |Q_{i,j}|\pm 0.1\gamma T_i(o)$ or $\tau(Q_{i,j})\in (1\pm 0.1)|Q_{i,j}|$.
\end{enumerate}

Consider step~\ref{it:step5} of Algorithm~\ref{alg:streaming}. 
Notice that since Algorithm~\ref{alg:coreset_construction} does not output FAIL, we have $\forall i\in\{0,1,\cdots, L\}, \sum_{j=1}^{s_i} |Q_{i,j}|\leq 2\cdot 10^4(kL+d^{1.5r})T_i(o)$.
Thus,
\begin{align*}
\forall i\in\{0,1,\cdots,L\}, \E\left[\sum_{j=1}^{s_i} \sum_{p\in Q_{i,j}} \hat{h}_i(p) \right] \leq \phi_i\cdot 2\cdot 10^4(kL+d^{1.5r})T_i(o).
\end{align*}
By Markov's inequality and union bound over all $i\in\{0,1,\cdots, L\}$, with probability at least $0.99$,
\begin{align*}
\forall i\in\{0,1,\cdots,L\}, \sum_{j=1}^{s_i} \sum_{p\in Q_{i,j}} \hat{h}_i(p) \leq  \phi_i\cdot 4\cdot  10^6  (k+d^{1.5r})L^2 T_i(o),
\end{align*}
which implies that
\begin{align*}
\forall i\in\{0,1,\cdots,L\},j\in[s_i], \{p\in Q_{i,j}\mid \hat{h}_i(p)=1\}\subseteq \hat{S}_i. 
\end{align*}
Thus, the construction of $Q'_i$ in step~\ref{it:step5} of Algorithm~\ref{alg:streaming} is equivalent to the construction of $Q'_i$ in line~\ref{sta:sampling_step} of Algorithm~\ref{alg:coreset_construction}.
Due to the correctness (Lemma~\ref{lem:correctness_offline}) of Algorithm~\ref{alg:streaming}, we complete the proof.
\end{proof}

\begin{lemma}[Space complexity and success probability of the streaming algorithm]\label{lem:success_prob_streaming}
For constant $r$, the space needed by Algorithm~\ref{alg:streaming} is at most $\poly(\varepsilon^{-1}\eta^{-1}kdL)$. 
Furthermore, if $\OPT^{(r)}_{k\text{-clus}}/10\leq o\leq \OPT^{(r)}_{k\text{-clus}}$, then with probability at least $0.95$, none of the subroutine of Algorithm~\ref{alg:streaming} outputs FAIL, and the output $Q'$ of Algorithm~\ref{alg:estimating_points} satifies $|Q'|\leq \poly(\varepsilon^{-1}\eta^{-1}kdL)$.
\end{lemma}

\begin{proof}
Suppose $r$ is a constant.
Since all the hash functions are $\lambda$-wise independent, the total number of random bits needed is at most $\poly(\varepsilon^{-1}\eta^{-1}kdL)$.
Since $\forall i\in\{0,1,\cdots, L\}$, $\psi_i,\psi'_i,\phi_i\leq \poly(\varepsilon^{-1}\eta^{-1}kdL)/T_i(o)$, we have $\alpha_i,\alpha'_i,\hat{\alpha}_i,\beta_i,\beta'_i,\hat{\beta}_i\leq \poly(\varepsilon^{-1}\eta^{-1}kdL)$.
By Lemma~\ref{lem:useful_tool}, the space complexity of each subroutine \textsc{Storing} (Lemma~\ref{lem:useful_tool}) is at most $\poly(\varepsilon^{-1}\eta^{-1}kdL)$ bits.
Since there are $3(L+1)$ parallel subroutines of \textsc{Storing}, the total space needed is at most $\poly(\varepsilon^{-1}\eta^{-1}kdL)$ bits.

Recall that $Z^*\subset [\Delta]^d$ with $|Z^*|\leq k$ is the optimal solution of the standard $\ell_r$ $k$-clustering problem of $Q$.
A cell $C\in G_i$ a center cell if $\dist(C,Z^*)\leq g_i/d$.
By Lemma~\ref{lem:number_of_center_cells}, $\mathcal{F}$ happens with probability at least $0.99$, i.e., the total number of center cells is at most $2000kL$.

Consider $i\in\{0,1,\cdots,L\}$. 
The number of points in non-center cells in $G_i$ is at most 
\begin{align*}
\frac{\OPT^{(r)}_{k\text{-clus}}}{(g_i/d)^2}\leq 1000d^{1.5r} T_i(o).
\end{align*}
Thus, for $i\in\{0,1,\cdots,L\}$, we have:
\begin{align*}
&\E\left[\left|\left\{ C\in G_i\mid \exists p\in C\cap Q,h_i(p)=1 \right\}\right|\right] \leq 2000kL + \psi_i \cdot 1000 d^{1.5r} T_i(o),\\
&\E\left[\left|\left\{ C\in G_i\mid \exists p\in C\cap Q,h'_i(p)=1 \right\}\right|\right] \leq 2000kL + \psi'_i \cdot 1000 d^{1.5r} T_i(o),\\
&\E\left[\left|\left\{ C\in G_i\mid \exists p\in C\cap Q,\hat{h}_i(p)=1 \right\}\right|\right] \leq 2000kL + \phi_i \cdot 1000 d^{1.5r} T_i(o).
\end{align*}
By Markov's inequality and union bound, with probability at least $0.99$, for $i\in\{0,1,\cdots,L\}$, we have:
\begin{align*}
&\left|\left\{ C\in G_i\mid \exists p\in C\cap Q,h_i(p)=1 \right\}\right| \leq 10^6(k+d^{1.5r}\psi_i T_i(o))L^2 \leq \alpha_i,\\
&\left|\left\{ C\in G_i\mid \exists p\in C\cap Q,h'_i(p)=1 \right\}\right| \leq 10^6(k+d^{1.5r}\psi'_i T_i(o))L^2\leq \alpha'_i,\\
&\left|\left\{ C\in G_i\mid \exists p\in C\cap Q,\hat{h}_i(p)=1 \right\}\right| \leq 10^6(k+d^{1.5r}\phi_i T_i(o))L^2\leq \hat{\alpha}_i,
\end{align*}
and thus none of subroutine \textsc{Storing} outputs FAIL.
According to the proof of Lemma~\ref{lem:correctness_streaming}, with probability at least $0.99$, 
\begin{enumerate}
\item $\forall i\in\{0,1,\cdots,L\},C\in G_i,$ $\tau(C\cap Q)$ is good (Definition~\ref{def:good_estimation_alg1}), 
\item $\forall i\in\{0,1,\cdots,L\},$  $\tau\left(\bigcup_{j=1}^{s_i} Q_{i,j}\right)$ is good (Definition~\ref{def:good_estimation_alg2}), 
\item $\forall i\in\{0,1,\cdots,L\},j\in[s_i]$, $\tau(Q_{i,j})$ is good (Definition~\ref{def:good_estimation_alg2}). 
\end{enumerate}
According to Lemma~\ref{lem:offline_success_prob}, with probability at least $0.99$, the subroutine Algorithm~\ref{alg:coreset_construction} does not output FAIL. 
By union bound, Algorithm~\ref{alg:streaming} does not output FAIL with probability at least $0.95$.
\end{proof}

\begin{theorem}[Streaming algorithm]\label{thm:streaming}
Suppose the input point set $Q\subseteq[\Delta]^d$ is obtained by a dynamic stream. 
For a constant $r\geq 1$, given $\varepsilon,\eta\in(0,0.5),k\in\mathbb{Z}_{\geq 1}$, there is a dynamic streaming algorithm which uses one pass over the stream and with probability at least $0.9$ can output a subset $Q'\subseteq Q$ and weights $w':Q'\rightarrow \mathbb{R}_{>0}$ such that 
\begin{enumerate}
\item $\forall t\geq |Q|/k,Z\subset[\Delta]^d$ with $|Z|=k$,
\begin{align*}
\begin{array}{ccc}
\cost_{(1+\eta)t}(Q,Z) \leq (1+\varepsilon)\cost_t(Q',Z,w')    &  \text{and} & \cost_{(1+\eta)t}(Q',Z,w')\leq (1+\varepsilon)\cost_t(Q,Z),
\end{array}
\end{align*}
\item $|Q'|\leq \poly(\varepsilon^{-1}\eta^{-1}kd\log\Delta)$.
\end{enumerate}
Furthermore, the space complexity of the algorithm is at most $\poly(\varepsilon^{-1}\eta^{-1}kd\log\Delta)$.
\end{theorem}
\begin{proof}
We can enumerate $o\in\{1,2,4,\cdots, \Delta^d\cdot (\sqrt{d}\Delta)^r\}$ and run Algorithm~\ref{alg:streaming} in parallel for each possible $o$.
By \cite{hsyz18}, there is a dynamic streaming algorithm which uses one pass over the stream and can give a $2$-approximation to $\OPT^{(r)}_{k\text{-clus}}$ with probability at least $0.99$.
We can run such algorithm in parallel, and at the end of the stream, we can find an $o$ such that $\OPT^{(r)}/10\leq o\leq \OPT^{(r)}$.
We conclude the proof by applying Lemma~\ref{lem:correctness_streaming} and Lemma~\ref{lem:success_prob_streaming} for such $o$.
\end{proof}

\subsection{Coreset Construction in Distributed Model} \label{subsec:dist}
Next we convert our streaming algorithm to a distributed algorithm. 
The distributed model is described as the following.
\paragraph{Distributed model.} We study the same model as in~\cite{kvw14,wz16,bwz16,swz17,swz19}.
There are $s$ machines, where the $i$-th machine holds a subset $Q^{(i)}\subseteq[\Delta]^d$ of input points.
There is one machine which is called coordinator.
The communication is only allowed between machines and the coordinator.
The communication cost of a protocol is the total number of bits needed to communicate between machines and the coordinator.

We show that there is a communication efficient distributed protocol which has a similar behavior as the subroutine shown in Lemma~\ref{lem:useful_tool}.

\begin{lemma}\label{lem:distributed_storing}
Suppose a point set is partitioned on $s$ machines where each machine holds a subset of points.
For $\alpha,\beta\in \mathbb{Z}_{\geq 1},i\in\{0,1,\cdots,L\}$,
there is a distributed protocol which requires $O(s\cdot \alpha\beta dL\log d)$ bits such that the protocol either leaves $\mathcal{C}\subset G_i$, $f:\mathcal{C}\rightarrow \mathbb{Z}_{\geq 1}$ and $S\subseteq [\Delta]^d$ on the coordinator or outputs FAIL on the coordinator, where
\begin{itemize}
\item  $\mathcal{C}$ is the set which contains all the non-empty cells,
\item  $f(C)$ denotes the number of points in the cell $C\in \mathcal{C}$,
\item  $S$ contains points in all non-empty cells that contains at most $\beta$ points,
\end{itemize}
and furthermore, if $|\mathcal{C}|\leq \alpha$, the protocol does not output FAIL on the coordinator.
\end{lemma}
\begin{proof}
The protocol is described as the following:
\begin{enumerate}
    \item The coordinator sends the randomly shifted vector to each machine such that each machine learns the grid $G_i$.
    \item The $j$-th machine finds non-empty cells $\mathcal{C}^{(j)}\subseteq G_i$ based on its local point set, computes the number of local points $f^{(j)}(C)$ of each cell $C\in\mathcal{C}^{(j)}$ and construct $S^{(j)}$ to be the set of local points in all cells that contains at most $\beta$ local points.
    \item If $\left|\mathcal{C}^{(j)}\right|\leq \alpha$, the $j$-th machine sends $\mathcal{C}^{(j)}, f^{(j)}$ and $S^{(j)}$ to the coordinator.
    Otherwise, the $j$-th machine sends FAIL to the coordinator.
    \item The coordinator outputs FAIL if it receives any FAIL from other machines. Otherwise, $\mathcal{C}\gets\bigcup_{j=1}^s \mathcal{C}^{(j)},$ for each cell $ C\in \mathcal{C}, f(C)\gets \sum_{j:C\in \mathcal{C}^{(j)}} f^{(j)}(C)$ and $S\gets \bigcup_{j=1}^s S^{(j)}$.
\end{enumerate}
The first step needs $O(dL\log d)$ bits per machine.
Consider the third step. 
The $j$-th machine needs $O(\alpha dL)$ bits to represent $\mathcal{C}^{j}$, needs $O(\alpha \cdot dL)$ bits to represent $f^{(j)}$, and needs $O(\alpha\beta dL)$ bits to represent $S^{(j)}$.
Thus, the overall communication cost is at most $O(s\cdot \alpha\beta dL\log d)$ bits.

Notice that if $|\mathcal{C}|\leq \alpha$, then $|\mathcal{C}^{(j)}|\leq \alpha$ for all $j\in[s]$ since $\mathcal{C}^{j}\subseteq \mathcal{C}$.
If the total number of points in a cell $C$ is at most $\beta$, then any machine can hold at most $\beta$ local points in $C$.
The above argument concludes the proof of correctness.
\end{proof}

\begin{theorem}\label{thm:distributed}
Suppose a point set $Q\subseteq[\Delta]^d$ is partitioned into $s$ machines. 
For a constant $r\geq 1$, given $\varepsilon,\eta\in(0,0.5),k\in\mathbb{Z}_{\geq 1}$, there is a distributed protocol which on termination with probability at least $0.9$ leaves a subset of points $Q'\subseteq Q$ and weights $w':Q'\rightarrow \mathbb{R}_{>0}$ on the coordinator such that
\begin{enumerate}
\item $\forall t\geq |Q|/k,Z\subset[\Delta]^d$ with $|Z|=k$,
\begin{align*}
\begin{array}{ccc}
\cost_{(1+\eta)t}(Q,Z) \leq (1+\varepsilon)\cost_t(Q',Z,w')    &  \text{and} & \cost_{(1+\eta)t}(Q',Z,w')\leq (1+\varepsilon)\cost_t(Q,Z),
\end{array}
\end{align*}
\item $|Q'|\leq \poly(\varepsilon^{-1}\eta^{-1}kd\log\Delta)$.
\end{enumerate}
Furthermore, the total communication cost of the protocol is at most $s\cdot \poly(\varepsilon^{-1}\eta^{-1}kd\log\Delta)$ bits.
\end{theorem}

\begin{proof}
By \cite{fl11,bfl16,bflsy17,hsyz18}, there is a distributed protocol using $s\cdot \poly(\varepsilon^{-1}\eta^{-1}kd\log\Delta)$ bits of communication leaves a $2$-approximation of $\OPT^{(r)}_{k\text{-clus}}$ on the coordinator with $0.99$ probability.
The the coordinator can broadcast the approxtion to every machine, and all the machines can agree on the same $o$ such that $\OPT^{(r)}_{k\text{-clus}}/10\leq o\leq \OPT^{(r)}_{k\text{-clus}}$.

Then the protocol simulates Algorithm~\ref{alg:streaming}. 
For step~\ref{it:running_storing} of Algorithm~\ref{alg:streaming}, we can use the protocol shown in Lemma~\ref{lem:distributed_storing} instead.
Due to the choices of $\alpha_i,\alpha'_i,\hat{\alpha}_i,\beta_i,\beta'_i,\hat{\beta}_i$, the total communication cost is at most $s\cdot \poly(\varepsilon^{-1}\eta^{-1}kdL)$.
For the remaining steps of Algorithm~\ref{alg:streaming}, we can simulate them on the coordinator. 
We conclude our proof by applying Lemma~\ref{lem:correctness_streaming} and Lemma~\ref{lem:success_prob_streaming}.
\end{proof}


\addcontentsline{toc}{section}{References}
\bibliographystyle{alpha}
\bibliography{ref}

\newcommand{\etalchar}[1]{$^{#1}$}
\begin{thebibliography}{GMMO00}

\bibitem[ABM{\etalchar{+}}18]{adamczyk2018constant}
Marek Adamczyk, Jaros{\l}aw Byrka, Jan Marcinkowski, Syed~M Meesum, and
  Micha{\l} W{\l}odarczyk.
\newblock Constant factor fpt approximation for capacitated k-median.
\newblock {\em arXiv preprint arXiv:1809.05791}, 2018.

\bibitem[ASS17]{an2017lp}
Hyung-Chan An, Mohit Singh, and Ola Svensson.
\newblock Lp-based algorithms for capacitated facility location.
\newblock {\em SIAM Journal on Computing}, 46(1):272--306, 2017.

\bibitem[BBLM14]{bblm14}
MohammadHossein Bateni, Aditya Bhaskara, Silvio Lattanzi, and Vahab Mirrokni.
\newblock Distributed balanced clustering via mapping coresets.
\newblock In {\em Advances in Neural Information Processing Systems (NIPS)},
  pages 2591--2599, 2014.

\bibitem[BFL16]{bfl16}
Vladimir Braverman, Dan Feldman, and Harry Lang.
\newblock New frameworks for offline and streaming coreset constructions.
\newblock {\em arXiv preprint arXiv:1612.00889}, 2016.

\bibitem[BFL{\etalchar{+}}17]{bflsy17}
Vladimir Braverman, Gereon Frahling, Harry Lang, Christian Sohler, and Lin~F
  Yang.
\newblock Clustering high dimensional dynamic data streams.
\newblock In {\em ICML}. \url{https://arxiv.org/pdf/1706.03887}, 2017.

\bibitem[BIP{\etalchar{+}}16]{birw16}
Arturs Backurs, Piotr Indyk, Eric Price, Ilya Razenshteyn, and David~P
  Woodruff.
\newblock Nearly-optimal bounds for sparse recovery in generic norms, with
  applications to k-median sketching.
\newblock In {\em Proceedings of the Twenty-Seventh Annual ACM-SIAM Symposium
  on Discrete Algorithms}, pages 318--337. SIAM, 2016.

\bibitem[BKS12]{bks12b}
Binay Bhattacharya, Tsunehiko Kameda, and Zhao Song.
\newblock Computing minmax regret 1-median on a tree network with
  positive/negative vertex weights.
\newblock In {\em International Symposium on Algorithms and Computation}, pages
  588--597. Springer, 2012.

\bibitem[BKS14]{bks14}
Binay Bhattacharya, Tsunehiko Kameda, and Zhao Song.
\newblock A linear time algorithm for computing minmax regret 1-median on a
  tree network.
\newblock {\em Algorithmica}, 70(1):2--21, 2014.

\bibitem[BLLM16]{braverman2016clustering}
Vladimir Braverman, Harry Lang, Keith Levin, and Morteza Monemizadeh.
\newblock Clustering problems on sliding windows.
\newblock In {\em Proceedings of the twenty-seventh annual ACM-SIAM symposium
  on Discrete algorithms}, pages 1374--1390. Society for Industrial and Applied
  Mathematics, 2016.

\bibitem[BR94]{br94}
Mihir Bellare and John Rompel.
\newblock Randomness-efficient oblivious sampling.
\newblock In {\em Foundations of Computer Science, 1994 Proceedings., 35th
  Annual Symposium on}, pages 276--287. IEEE, 1994.

\bibitem[BRU16]{byrka2016approximation}
Jaros{\l}aw Byrka, Bartosz Rybicki, and Sumedha Uniyal.
\newblock An approximation algorithm for uniform capacitated k-median problem
  with $1+\epsilon$ capacity violation.
\newblock In {\em International Conference on Integer Programming and
  Combinatorial Optimization}, pages 262--274. Springer, 2016.

\bibitem[BWZ16]{bwz16}
Christos Boutsidis, David~P Woodruff, and Peilin Zhong.
\newblock Optimal principal component analysis in distributed and streaming
  models.
\newblock In {\em Proceedings of the 48th Annual ACM SIGACT Symposium on Theory
  of Computing (STOC)}, pages 236--249. ACM,
  \url{https://arxiv.org/pdf/1504.06729}, 2016.

\bibitem[Che09]{chen09}
Ke~Chen.
\newblock On coresets for k-median and k-means clustering in metric and
  euclidean spaces and their applications.
\newblock {\em SIAM Journal on Computing}, 39(3):923--947, 2009.

\bibitem[COP03]{cop03}
Moses Charikar, Liadan O'Callaghan, and Rina Panigrahy.
\newblock Better streaming algorithms for clustering problems.
\newblock In {\em Proceedings of the thirty-fifth annual ACM symposium on
  Theory of computing}, pages 30--39. ACM, 2003.

\bibitem[DL16]{demirci2016constant}
G{\"o}kalp Demirci and Shi Li.
\newblock Constant approximation for capacitated $ k $-median with
  $(1+\epsilon)$-capacity violation.
\newblock {\em arXiv preprint arXiv:1603.02324}, 2016.

\bibitem[FL11]{fl11}
Dan Feldman and Michael Langberg.
\newblock A unified framework for approximating and clustering data.
\newblock In {\em Proceedings of the 43rd {ACM} Symposium on Theory of
  Computing, {STOC} 2011, San Jose, CA, USA, 6-8 June 2011}, pages 569--578,
  2011.

\bibitem[FMS07]{fms07}
Dan Feldman, Morteza Monemizadeh, and Christian Sohler.
\newblock A ptas for k-means clustering based on weak coresets.
\newblock In {\em Proceedings of the twenty-third annual symposium on
  Computational geometry}, pages 11--18. ACM, 2007.

\bibitem[FS05]{fs05}
Gereon Frahling and Christian Sohler.
\newblock Coresets in dynamic geometric data streams.
\newblock In {\em Proceedings of the thirty-seventh annual ACM symposium on
  Theory of computing (STOC)}, pages 209--217. ACM, 2005.

\bibitem[GMMO00]{gmmo00}
Sudipto Guha, Nina Mishra, R.~Motwani, and L.~O'Callaghan.
\newblock Clustering data streams.
\newblock In {\em FOCS}, pages 359--366, 2000.

\bibitem[HPK05]{hk05}
Sariel Har-Peled and Akash Kushal.
\newblock Smaller coresets for k-median and k-means clustering.
\newblock In {\em Proceedings of the twenty-first annual symposium on
  Computational geometry}, pages 126--134. ACM, 2005.

\bibitem[HPM04]{hm04}
Sariel Har-Peled and Soham Mazumdar.
\newblock On coresets for k-means and k-median clustering.
\newblock In {\em Proceedings of the thirty-sixth annual ACM symposium on
  Theory of computing}, pages 291--300. ACM, 2004.

\bibitem[HSYZ18]{hsyz18}
Wei Hu, Zhao Song, Lin~F Yang, and Peilin Zhong.
\newblock Nearly optimal dynamic $ k $-means clustering for high-dimensional
  data.
\newblock {\em arXiv preprint arXiv:1802.00459}, 2018.

\bibitem[KVW14]{kvw14}
Ravindran Kannan, Santosh~S Vempala, and David~P Woodruff.
\newblock Principal component analysis and higher correlations for distributed
  data.
\newblock In {\em Proceedings of The 27th Conference on Learning Theory
  (COLT)}, pages 1040--1057, 2014.

\bibitem[Li17]{li2017uniform}
Shi Li.
\newblock On uniform capacitated k-median beyond the natural lp relaxation.
\newblock {\em ACM Transactions on Algorithms (TALG)}, 13(2):22, 2017.

\bibitem[LLMR15]{llmr15}
Silvio Lattanzi, Stefano Leonardi, Vahab Mirrokni, and Ilya Razenshteyn.
\newblock Robust hierarchical k-center clustering.
\newblock In {\em Proceedings of the 2015 Conference on Innovations in
  Theoretical Computer Science}, pages 211--218. ACM, 2015.

\bibitem[MMR19]{mmr19}
Konstantin Makarychev, Yury Makarychev, and Ilya Razenshteyn.
\newblock Performance of johnson-lindenstrauss transform for k-means and
  k-medians clustering.
\newblock In {\em Proceedings of the 51st Annual ACM SIGACT Symposium on Theory
  of Computing}, pages 1027--1038. ACM, 2019.

\bibitem[SWZ17]{swz17}
Zhao Song, David~P Woodruff, and Peilin Zhong.
\newblock Low rank approximation with entrywise $\ell_1$-norm error.
\newblock In {\em Proceedings of the 49th Annual Symposium on the Theory of
  Computing (STOC)}. ACM, \url{https://arxiv.org/pdf/1611.00898}, 2017.

\bibitem[SWZ19]{swz19}
Zhao Song, David~P Woodruff, and Peilin Zhong.
\newblock Relative error tensor low rank approximation.
\newblock In {\em Proceedings of the Thirtieth Annual ACM-SIAM Symposium on
  Discrete Algorithms}, pages 2772--2789. Society for Industrial and Applied
  Mathematics, 2019.

\bibitem[WZ16]{wz16}
David~P Woodruff and Peilin Zhong.
\newblock Distributed low rank approximation of implicit functions of a matrix.
\newblock In {\em 32nd IEEE International Conference on Data Engineering
  (ICDE)}. \url{https://arxiv.org/pdf/1601.07721}, 2016.

\bibitem[XHX{\etalchar{+}}19]{xu2019constant}
Yicheng Xu, Rolf H.Mohring, Dachuan Xu, Yong Zhang, and Yifei Zou.
\newblock A constant parameterized approximation for hard-capacitated k-means.
\newblock {\em arXiv preprint arXiv:1901.04628}, 2019.

\end{thebibliography}

\appendix
\section{Missing Details of Section~\ref{sec:points_partitioning}}\label{sec:missing_proofs_partition}

\begin{fact}
Suppose $o\leq \OPT^{(r)}_{k\text{-clus}}$. If the estimated size $\tau(C\cap Q)$ in line~\ref{sta:estimated_size} of Algorithm~\ref{alg:heavy_light} is good (Definition~\ref{def:good_estimation_alg1}) for every cell, then there is a unique cell $C\in G_{-1}$ which contains $[\Delta]^d$ and is marked as heavy.
\end{fact}
\begin{proof}
As mentioned, since each cell in $G_{-1}$ has side length $2\Delta$, there is a unique cell $C\in G_{-1}$ such that $[\Delta]^d\subset C$.
Thus, $C\cap Q=Q$. 
Notice that $\OPT^{(r)}_{k\text{-clus}}\leq (\sqrt{d}\Delta)^r |Q|\leq (\sqrt{d}g_{-1})^r |Q|$.
Since $T_{-1}(o)=0.01\cdot o/(\sqrt{d}g_{-1})^r\leq 0.01\cdot \OPT^{(r)}_{k\text{-clus}}/(\sqrt{d}g_{-1})^r$, $C$ should be heavy due to Algorithm~\ref{alg:heavy_light}.
\end{proof}

\paragraph{Proof of Lemma~\ref{lem:number_of_center_cells}}
Suppose $Z^*=\{z_1^*,z_2^*,\cdots,z_k^*\}$.
For $i\in\{0,1,\cdots,L\},j\in[k]$, consider the grid $G_i$ and the center $z_j^*$.
For $l\in[d]$, let $X_l$ be the indicator random variable such that $X_l=1$ if and only if the distance between $z_j^*$ and the boundary of the $l$-th dimension of $G_i$ is at most $g_i/d$.
Notice that if $z_j^*$ is close to a boundary, the number of cells which is close to $z_j^*$ may increase by a factor of $2$.
Therefore, the number of cells which has distance to $z_j^*$ is at most
$2^{\sum_{l=1}^d X_l}$ of which expectation is at most
\begin{align*}
\prod_{l=1}^d \E\left[ 2^{X_l}\right] = \prod_{l=1}^d \left(1+\E[X_l]\right)=(1+2/d)^d\leq e^2.
\end{align*}

Thus, the expectation of total number of center cells is at most $(L+1)k\cdot e^2$.
By Markov's inequality, with probability at most $0.99$, the total number of center cells is at most $2000kL$.
\hfill{$\qed$}

\paragraph{Proof of Lemma~\ref{lem:num_heavy_cells}}
Condition on $\mathcal{F}$, the total number of center cells is at most $2000kL$ (Lemma~\ref{lem:number_of_center_cells}).
Consider a cell $C\in G_i$ which is not a center cell. 
If $C$ is heavy, then
\begin{align*}
\sum_{p\in C\cap Q} \dist^r(p,Z^*)\geq |C\cap Q|\cdot \left(\frac{g_i}{d}\right)^r\geq 0.9 T_i(o)\cdot \left(\frac{g_i}{d}\right)^r\geq \frac{o}{120d^{1.5r}}.
\end{align*}
Since $\sum_{C\in G_i}\sum_{p\in C\cap Q} \dist^r(p,Z^*)=\sum_{p\in Q}\dist^r(p,Z^*)= \OPT^{(r)}_{k\text{-clus}}$,
 the total number of heavy cells which are not center cells are at most $(L+1)\cdot \OPT^{(r)}_{k\text{-clus}}/(o/(120 d^{1.5r}))\leq 240 d^{1.5r} L\cdot \frac{\OPT_{k\text{-means}}}{o}$.
Together with the upper bound of number of center cells, we can conclude the proof.
\hfill{$\qed$}

\paragraph{Proof of Lemma~\ref{lem:small_parts_removal}}
Let $Q^I=Q\setminus Q^N$.
Since $Q^I$ is a subset of $Q$, $\cost_t(Q^I,Z)\leq \cost_t(Q,Z)$ is obviously true by the definition of $\cost_t(\cdot,\cdot)$.
In the remaining of the proof, let us focus on proving $\cost_{(1+\eta)\cdot t}(Q,Z)\leq (1+\varepsilon)\cost_{t}(Q^I,Z)$.
Before we prove the statement, there are several observations.

\begin{claim}\label{cla:num_of_non_important_points}
$|Q^N|\leq \eta|Q|/k$.
\end{claim}
\begin{proof}
\begin{align*}
|Q^N| & = \sum_{i=0}^L \sum_{P\in\mathcal{P}^N_i} |P|\leq \sum_{i=0}^L \sum_{P\in\mathcal{P}_i^N} 2\gamma T_i(o)\leq \sum_{i=0}^L \sum_{P\in \mathcal{P}_i^N} 2^{r+2}\gamma|c_{i-1}(P)\cap Q|\leq  2^{r+2}\gamma(L+1) |Q|\leq \eta|Q|/k,
\end{align*}
where the third step follows from that $c_{i-1}(P)$ is a heavy cell and thus $|c_{i-1}(P)\cap Q|\geq 0.5T_{i-1}(o)\geq T_i(o)/2^{r+1}$, the fifth step follows from $\gamma\leq \eta/(2^{r+3}kL)$.
\end{proof}

\begin{claim}\label{cla:each_cell_will_contain_many}
For $i\in\{0,1,\cdots,L\}$, if $C\in G_{i-1}$ is marked as heavy by Algorithm~\ref{alg:heavy_light}, then $|C\cap Q^I|\geq (1-2^{r+2}(L-i)\gamma)|C\cap Q|\geq (1-2^{r+2}L\gamma)\cdot 0.5T_{i-1}(o)$.
\end{claim}
\begin{proof}
We prove it by induction.
Consider the case when $i=L$.
If there is no heavy cell in $G_{L-1}$, the claim is true.
Otherwise, consider a heavy cell $C\in G_{L-1}$.
Since all the children of $C$ are crucial, all the points in the children of $C$ will be in the same part of which size is at least $0.5T_{L-1}(o)> 2\gamma T_L(o)$.
Therefore $|C\cap Q^I|=|C\cap Q|$.

Suppose the claim is true for $i+1,i+2,\cdots, L$.
If there is no heavy cell in $G_i$, the claim is true.
Otherwise, consider a heavy cell $C\in G_{i-1}$.
There are two cases.
The first case is that all points from $Q$ in the crucial children of $C$ are also in $Q^I$.
In this case, we have
\begin{align*}
|C\cap Q^I|  &= \sum_{C'\in G_i:~C'\text{ is a crucial child of }C} |C'\cap Q^I| + \sum_{C'\in G_i:~C'\text{ is a heavy child of }C} |C'\cap Q^I|\\
& = \sum_{C'\in G_i:~C'\text{ is a crucial child of }C} |C'\cap Q| + \sum_{C'\in G_i:~C'\text{ is a heavy child of }C} (1-16(L-i-1)\gamma)|C'\cap Q|\\
&\geq (1-16(L-i)\gamma)|C\cap Q|.
\end{align*}
The second case is that none of the point from $Q$ in the curcial children of $C$ is in $Q^I$. In this case, we have
\begin{align*}
\sum_{C'\in G_i:~C'\text{ is a heavy child of }C} |C'\cap Q| &= |C\cap Q| - \sum_{C'\in G_i:~C'\text{ is a crucial child of }C} |C'\cap Q|\\
&\geq |C\cap Q| - 2\gamma T_i(o)\geq (1-2^{r+2}\gamma)|C\cap Q|,
\end{align*}
where the last step follows from that $|C\cap Q|\geq 0.5T_{i-1}(o)\geq 1/2^{r+1}\cdot T_i(o)$.
Thus,
\begin{align*}
|C\cap Q^I| & = \sum_{C'\in G_i:~C'\text{ is a heavy child of }C} |C'\cap Q^I|\\
& \geq (1-2^{r+2}(L-i-1)\gamma)\sum_{C'\in G_i:~C'\text{ is a heavy child of }C}|C'\cap Q|\\
& \geq (1-2^{r+2}(L-i-1)\gamma)(1-2^{r+2}\gamma) |C\cap Q|\\
&\geq (1-2^{r+2}(L-i)\gamma)|C\cap Q|.
\end{align*}
\end{proof}

It is enough to prove the lemma:
\begin{align*}
& \cost_{(1+\eta) t} (Q,Z)\\
\leq &\cost_{t+\eta|Q|/k}(Q,Z)\\
\leq & \cost_{t}(Q^I,Z)+\cost(Q^N,Z)\\
= & \cost_t(Q^I,Z) + \sum_{i=0}^L \sum_{P\in \mathcal{P}^N_i} \sum_{p\in P} \dist^r(p,Z)\\
\leq & \cost_t(Q^I,Z) + \sum_{i=0}^L \sum_{P\in\mathcal{P}^N_i} \sum_{p\in P} \left(2^{r-1}(\sqrt{d}g_{i-1})^r+ 2^{r-1}\frac{\sum_{q\in c_{i-1}(P)\cap Q^I}\dist^r(q,Z)}{|c_{i-1}(P)\cap Q^I|}\right)\\
\leq & \cost_t(Q^I,Z) + \sum_{i=0}^L \sum_{P\in\mathcal{P}^N_i} \sum_{p\in P} \left(2^{r-1}(\sqrt{d}g_{i-1})^r+ 2^{r+1}\frac{\cost(Q^I,Z)}{T_{i-1}(o)}\right)\\
\leq & \cost_t(Q^I,Z) + \sum_{i=0}^L |\mathcal{P}_i^N| \cdot 2\gamma T_i(o)\cdot \left(2^{r-1}(\sqrt{d}g_{i-1})^r+ 2^{r+1}\frac{\cost(Q^I,Z)}{T_{i-1}(o)}\right)\\
= & \cost_t(Q^I,Z) + \sum_{i=0}^L |\mathcal{P}_i^N| \cdot 0.02\gamma \left(2^{2r-1}o+ 2^{2r+1}\cost(Q^I,Z)\right)\\
\leq & \cost_t(Q^I,Z) + 20000(k+d^{1.5r})L\cdot 0.02\gamma \left(2^{2r-1}o+ 2^{2r+1}\cost(Q^I,Z)\right)\\
\leq &\cost_t(Q^I,Z) + 20000(k+d^{1.5r})L\cdot 0.02\gamma \left(2^{2r-1}\cost_{(1+\eta)t}(Q,Z)+ 2^{2r+1}\cost(Q^I,Z)\right)\\
\leq &\cost_t(Q^I,Z) + 1000(k+d^{1.5r})L\cdot 2^{2r}\gamma \left(\cost_{(1+\eta)t}(Q,Z)+ \cost(Q^I,Z)\right)\\
\leq &\cost_t(Q^I,Z) + \frac{\varepsilon}{4} \left(\cost_{(1+\eta)t}(Q,Z)+ \cost(Q^I,Z)\right)\\
\leq &\cost_t(Q^I,Z) + \frac{\varepsilon}{4} \left(\cost_{(1+\eta)t}(Q,Z)+ \cost_t(Q^I,Z)\right)
\end{align*}
where the first step follows from $t\geq |Q|/k$, the second step follows from Claim~\ref{cla:num_of_non_important_points}, the forth step follows from that since $c_{i-1}(P)\cap Q^I\not=\emptyset$ (Claim~\ref{cla:each_cell_will_contain_many}), there must be a point $q'\in c_{i-1}(P)\cap Q^I$ such that 
\begin{align*}
\dist^r(q',Z)\leq \frac{\sum_{q\in c_{i-1}(P)\cap Q^I} \dist^r(q,Z)}{|c_{i-1}(P)\cap Q^I|}
\end{align*}
by averaging argument and furthermore, due to $\dist(p,q')\leq \sqrt{d}g_{i-1}$ and Fact~\ref{fac:approx_tri}, we have 
\begin{align*}
\dist^r(p,Z)\leq 2^{r-1}\dist^r(p,q')+2^{r-1}\dist^r(q',Z)\leq 2^{r-1} \left(\sqrt{d} g_{i-1}\right)^r + 2^{r-1} \frac{\sum_{q\in c_{i-1}(P)\cap Q^I} \dist^r(q,Z)}{|c_{i-1}(P)\cap Q^I|},
\end{align*}
the fifth step follows from $|c_{i-1}(P)\cap Q^I|\geq 0.25 T_{i-1}(o)$ (Claim~\ref{cla:each_cell_will_contain_many}), the sixth step follows from $\forall P\in \mathcal{P}_i^N$, $|P|\leq 2\gamma T_i(o)$, the seventh step follows from $T_i(o)=0.01 o/(\sqrt{d}g_i)^r=2^rT_{i-1}(o)$, the eighth step follows from $|\mathcal{P}_i^N|\leq s_i$, the ninth step follows from $o\leq \OPT^{(r)}_{k\text{-clus}}\leq \cost_{(1+\eta)t}(Q,Z)$, and the last step follows from $\cost(Q^I,Z)\leq \cost_t(Q^I,Z)$.

By rearranging the term of the above inequality, we have $(1-\varepsilon/4)\cost_{(1+\eta)t}(Q,Z)\leq (1+\varepsilon/4)\cost_{t}(Q^I,Z)$. Since $\varepsilon\in(0,0.5)$, we can conclude that $\cost_{(1+\eta)t}(Q,Z)\leq (1+\varepsilon)\cost_t(Q^I,Z)$.

\hfill{$\qed$}

\nocite{bks14,llmr15,birw16,bks12b,mmr19}




\end{document}